\documentclass[12pt,oneside,reqno]{amsart}

\usepackage{geometry,verbatim}
\usepackage{mathabx,mathtools,amsmath,amssymb,enumitem,amsthm,nccmath}

\usepackage{bm} 

\usepackage{graphicx,url}

\usepackage{color,cite,verbatim}

\usepackage[svgnames]{xcolor}

\usepackage[colorlinks=true, linkcolor=Maroon, urlcolor=Maroon]{hyperref}

\hypersetup{citecolor=blue}

\usepackage{dsfont}
\usepackage{suetterl} 
\usepackage{pgfplots}	
\pgfplotsset{compat=1.18} 

\usepackage{caption}
\usepackage{subcaption}		
\usepackage{braket}
\usepackage{stmaryrd} 

\usepackage[T1]{fontenc}
\usepackage{currvita}

\usepackage{csquotes} 

\mathtoolsset{showonlyrefs} 

\usepackage{simpler-wick} 	
\newlength{\wdth}

\usepackage{tikz}
\usepgfplotslibrary{fillbetween}
\usetikzlibrary{patterns}

\usetikzlibrary{arrows}
\usetikzlibrary{arrows,positioning,arrows.meta} 

\allowdisplaybreaks		

\def\C{{\mathbb C}}
\def\N{{\mathbb N}}	
\def\R{{\mathbb R}}				 
\def\Z{{\mathbb Z}}				

\renewcommand{\S}{{\mathbb S}}

\def\bk{{\bf k}}

\def\bx{{\bf x}}
\def\by{{\bf y}}
\def\bp{\mathbf{p}}
\def\bprev{\cev{\mathbf{p}}}
\def\bpb{\overline{\mathbf{p}}}
\def\bk {\mathbf{k}}				
\def\bq {\mathbf{q}}				
\def\bs{\mathbf{s}}					

\def\cE{{\mathcal E}}
\def\cF{{\mathcal F}}
\def\cH{{\mathcal H}}
\def\cJ{{\mathcal J}}					 
\def\cK{{\mathcal K}}					
\def\cL{{\mathcal L}}
\def\cM{{\mathcal M}}
\def\cN{{\mathcal N}}
\def\cP{{\mathcal P}}
\def\cQ{{\mathcal Q}}
\def\cR{{\mathcal R}}
\def\cS{{\mathcal S}}
\def\cT{{\mathcal T}}
\def\cU{{\mathcal U}}
\def\cV{{\mathcal V}} 
\def\cW{{\mathcal W}}
\def\cX{{\mathcal X}} 
\def\cY{{\mathcal Y}}

\newcommand{\cev}[1]{\reflectbox{\ensuremath{\vec{\reflectbox{\ensuremath{#1}}}}}}

\renewcommand{\Re}{\operatorname{Re}}
\renewcommand{\Im}{\operatorname{Im}}			
\newcommand{\tr}{\operatorname{Tr}}   

\def\1{{\bf 1}}
\def\alg{{\mathfrak A}}

\newcommand{\scp}[2]{\langle#1,#2\rangle}
\def\jb#1{\langle#1\rangle}

\def\ad{a^\dagger}

\def\eqn*{\begin{align*}}			
	\def\eeqn*{\end{align*}}		
\def\eqn{\begin{eqnarray}}				  
	\def\eeeqn{\end{eqnarray}}

\def\prf{\begin{proof}}
	\def\endprf{\end{proof}} 

\makeatletter
\def\namedlabel#1#2{\begingroup					
	#2%
	\def\@currentlabel{#2}%
	\phantomsection\label{#1}\endgroup
}
\makeatother

\theoremstyle{plain}
\newtheorem{theorem}{Theorem}[section]		
\newtheorem{definition}[theorem]{Definition}				
\newtheorem{proposition}[theorem]{Proposition}		   

\newtheorem{lemma}[theorem]{Lemma}
\newtheorem{corollary}[theorem]{Corollary}
\newtheorem{remark}[theorem]{Remark}
\newtheorem*{remark*}{Remark}

\numberwithin{equation}{section}

\def\weyl{\cW}
\def\weyld{\cW^\dagger}
\def\bog{\cT} 								
\def\bprop{\cU_{\mathrm{Bog}}} 		 	   
\def\fluc{\cU_{\mathrm{fluc}}}						  
\def\bogd{\cT^\dagger}
\def\bpropd{\bprop^\dagger}
\def\flucd{\fluc^\dagger}

\def\whfluc{\widehat{\cU}_N}						   

\def\Hfluc{\cH_{\mathrm{fluc}}}

\def\Hcubcon{\wick{\settowidth{\wdth}{$\Hcub$}\hspace{.33\wdth}\c1{\vphantom{\Hcub}}\hspace{.33\wdth}\c1{\vphantom{\Hcub}}\hspace{-.66\wdth}\Hcub}}

\def\Hquart{\cH_{\mathrm{quart}}}
\def\Hquartcon{\wick{\settowidth{\wdth}{$\Hquart$}\hspace{.33\wdth}\c1{\vphantom{\Hquart}}\hspace{.33\wdth}\c1{\vphantom{\Hquart}}\hspace{-.66\wdth}\Hquart}}
\def\Hquartcond#1{\wick{\settowidth{\wdth}{$\Hquart$}\hspace{.33\wdth}\c1{\vphantom{\Hquart}}\hspace{.33\wdth}\c1{\vphantom{\Hquart}}\hspace{-.66\wdth}\Hquart(#1)^{(d)}}}
\def\Hquartconod#1{\wick{\settowidth{\wdth}{$\Hquart$}\hspace{.33\wdth}\c1{\vphantom{\Hquart}}\hspace{.33\wdth}\c1{\vphantom{\Hquart}}\hspace{-.66\wdth}\Hquart(#1)^{(cor)}}}

\def\HHFB{\cH_{\mathrm{HFB}}}
\def\HHFBd{\HHFB^{\mathrm{(d)}}}
\def\HHFBod{\HHFB^{\mathrm{(cor)}}}
\def\Hcub{\cH_{\mathrm{cub}}}
\def\HBEC{\cH_{\mathrm{BEC}}}

\def\HBECt{\HBEC^{(\phi_t)}}
\def\HHFBodt{\HHFB^{(\mathrm{cor},\phi_t)}}
\def\Hcubt{\Hcub^{(\phi_t)}}

\def\nd#1{\|#1\|_X}
\def\wnd#1{\|#1\|_Y}													
\def\fd{\nd{f_0}}

\def\vd{\wnd{\hat{v}}}

\DeclareMathOperator{\osg}{sign}		 

\DeclareMathOperator\diag{diag}

\newcommand{\Rem}{\operatorname{Rem}}   

\def\vep{{\varepsilon}}

\def\vol{|\Lambda|}
\def\lattice{{\Lambda^*}}
\def\nb{\mathcal{N}_b}

\def\dx#1{\mathrm{d}#1\ }

\def\fbar{\widetilde{f_0}}

\def\hb{\widetilde{h}}					
\def\fplus{f_0^{(+)}}				     

\def\totst#1{\jb{#1}^{(\mathrm{tot})}}
\def\ftot{f^{(\mathrm{tot})}}
\def\freg{f^{(\mathrm{ex})}}

\def\Phitot{\Phi^{(\mathrm{tot})}}
\def\gtot{g^{(\mathrm{tot})}}

\def\stepo{1}
\def\stept{2}
 
\def\phio{\phi^{(\stepo)}}
\def\phiob{\overline{\phi}^{(\stepo)}}
\def\Sigo{\Sigma^{(\stepo)}}
\def\Sigob{\overline{\Sigma}^{(\stepo)}}

\def\sigo{\sigma^{(\stepo)}}
\def\sigob{\overline{\sigma}^{(\stepo)}}

\def\Gamo{\Gamma^{(\stepo)}}

\def\gamo{\gamma^{(\stepo)}}

\def\Omo{\Omega^{(\stepo)}}

\def\ehfb{\cE_{\mathrm{HFB}}}

\def\phit{\phi^{(\stept)}}
\def\phitb{\overline{\phi}^{(\stept)}}
\def\Sigt{\Sigma^{(\stept)}} 
\def\Sigtb{\overline{\Sigma}^{(\stept)}}

\def\sigt{\sigma^{(\stept)}} 
\def\sigtb{\overline{\sigma}^{(\stept)}} 
\def\Gamt{\Gamma^{(\stept)}}

\def\gamt{\gamma^{(\stept)}}
\def\Omt{\Omega^{(\stept)}}

\def\hfbspace{\cX}

\def\Gamj{\Gamma^{(j)}}

\def\Sigj{\Sigma^{(j)}}

\def\gtr{\Gamma^{T}}

\def\str{\Sigma^{T}}

\def\strb{\overline{\Sigma}^{T}}

\def\symp{\cR}
\def\hgam{h_\Gamma}
\def\hsig{h_\Sigma}
\def\hsymp{\cH^{(\Gamma,\Sigma)}}
\def\hsympt{\cH^{(\Gamma_t,\Sigma_t)}}
\def\ssymp{\cS}
\def\sympbog{\cV}
\def\sympbogd{\cV^\dagger}
\def\sympspace{\cY}

\newcommand{\bbf}[2]{\mathbf{w}^{(#2,#1)}}
\newcommand{\bbfb}[2]{\overline{\mathbf{w}}^{(#2,#1)}}

\makeatother

\newcommand{\mh}[1]{{\color{black} #1}}

\makeatletter
\def\paragraph{\@startsection{paragraph}{4}%
  \z@\z@{-\fontdimen2\font}%
  {\normalfont\bfseries}}
\makeatother

\makeatletter
\def\@tocline#1#2#3#4#5#6#7{\relax
  \ifnum #1>\c@tocdepth 
  \else
    \par \addpenalty\@secpenalty\addvspace{#2}%
    \begingroup \hyphenpenalty\@M
    \@ifempty{#4}{%
      \@tempdima\csname r@tocindent\number#1\endcsname\relax
    }{%
      \@tempdima#4\relax
    }%
    \parindent\z@ \leftskip#3\relax \advance\leftskip\@tempdima\relax
    \rightskip\@pnumwidth plus4em \parfillskip-\@pnumwidth
    #5\leavevmode\hskip-\@tempdima
      \ifcase #1
       \or\or \hskip 2em \or \hskip 3em \else \hskip 4em \fi%
      #6\nobreak\relax
    \hfill\hbox to\@pnumwidth{\@tocpagenum{#7}}\par
    \nobreak
    \endgroup
  \fi}
\makeatother

\usepackage[affil-it]{authblk}
\usepackage{orcidlink}
\usepackage{amsaddr}
 
\begin{document}

	\title[Quantum fluctuations around BEC]
	{Derivation of renormalized Hartree-Fock-Bogoliubov and quantum Boltzmann equations in an interacting Bose gas}
	
	\author{Thomas Chen\orcidlink{0000-0003-2704-1454}}
	\address[T. Chen]{Department of Mathematics, University of Texas at Austin, Austin TX 78712, USA}
	\email{tc@math.utexas.edu} 
 
	\author{Michael Hott\orcidlink{0000-0003-4243-6585}}
	\address[M. Hott]{Department of Mathematics, University of Minnesota - Twin Cities, Minneapolis MN 55414, USA}
	\email{mhott@umn.edu} 


    \maketitle
    
	\begin{abstract}
		Our previous work \cite{chenhott} presented a rigorous derivation of quantum Boltzmann equations near a Bose-Einstein condensate (BEC). Here, we extend it with a complete characterization of the leading order fluctuation dynamics. For this purpose, we correct the latter via an appropriate Bogoliubov rotation, in partial analogy to the approach by Grillakis-Machedon et al. \cite{grmama}, in addition to the Weyl transformation applied in \cite{chenhott}. Based on the analysis of the third order expansion of the BEC wave function, and the second order expansions of the pair-correlations, we show that through a renormalization strategy, various contributions to the effective Hamiltonian can be iteratively eliminated by an appropriate choice of the Weyl and Bogoliubov transformations. This leads to a separation of renormalized Hartree-Fock-Bogoliubov (HFB) equations and quantum Boltzmann equations. A multitude of terms that were included in the error term in \cite{chenhott} are now identified as contributions to the HFB renormalization terms. Thereby, the error bound in the work at hand is improved significantly. To the given order, it is now sharp, and matches the order or magnitude expected from scaling considerations. Consequently, we extend the time of validity to $t\sim (\log N)^2$ compared to $t\sim (\log N/\log \log N)^2$ before. We expect our approach to be extensible to smaller orders in $\frac1N$.  
	\end{abstract}
	
	\bigskip

    \paragraph{Statements and Declarations} The authors do not declare financial or non-financial interests that are directly or indirectly related to the work submitted for publication.

    \bigskip
        
    \paragraph{Acknowledgements} M.H.'s research was supported by 'University of Texas at Austin Provost Graduate Excellence Fellowship' and NSF grants DMS-1716198 and DMS-2009800 through T.C. `T.C. gratefully acknowledges support by the NSF through the grant DMS-2009800, and the RTG Grant DMS-1840314 - {\em Analysis of PDE}. We thank Esteban C\'ardenas and Jacky Chong for helpful discussions.
	
	\tableofcontents

	\section{Introduction}

    \subsection{Summary of previous results}

    The Boltzmann (transport) equation describes the time-dependent behavior of the phase space probability distribution $f$ of fluids and gases. It takes the form
    \begin{equation}\label{eq-cBE}
        (\partial_t+p\cdot \nabla_x) \mh{f_t} \, = \, Q[\mh{f_t}] \, ,
    \end{equation}
    where, for a classical gas, 
    \begin{equation}\label{eq-cBE-op}
        Q^{\mathrm{(cl)}}[f](p) \, = \, \int_{\S^2}\dx{\omega}\int_{\R^3}\dx{p_*}b(\omega,|p-p_*|)(f(p')f(p_*')-f(p)f(p_*))\Big|_{\substack{p'=p+[\omega\cdot(p_*-p)]\omega\\p_*'=p_*-[\omega\cdot(p_*-p)]\omega}} \, .
    \end{equation}
    In order to derive this equation, Boltzmann imposed the \emph{Stosszahlansatz} or \emph{molecular chaos assumption}, which requires that \mh{all marginals of the joint distribution of the Bose gas particles factorize into products of one-particle marginals}. We will refer to this condition as \emph{propagation of \mh{factorization}}\footnote{\mh{In the context of the classical Boltzmann equation, it is referred to as \emph{propagation of chaos}. However, the notion for quantum particles differs from the classical concept, which is why we refrain from this nomenclature in the present context.}}. 

    \par To this day, the rigorous derivation of the Boltzmann equation, and possible corrections, for all physically relevant regimes remains a widely open problem. However, in some special cases, and for sufficiently short times, there has been progress both from a heuristic and mathematically rigorous perspective. Hilbert \cite{hilbert1902problemes} famously asked for a rigorous justification of Boltzmann's equations; the derivation of physical laws from a set of mathematical axioms is referred to as \emph{Hilbert's sixth problem}. A fundamental question that emerges in this context is how irreversibility of the (mesoscopic) Boltzmann equation, for which the entropy functional is non-decreasing, arises from the microscopic reversible many-body dynamics, see, e.g., \cite{lebowitz1993macroscopic}. 
    
    \par In the case of classical gases, starting from the classical Liouville equation for an interacting $N$-particle gas, Boltzmann's idea to show the propagation of factorization in the derivation of Boltzmann's equation has been made rigorous, albeit for times that allow for at most one collision. The first rigorous, though incomplete, results go back to Cercignani \cite{cerc} and Lanford \cite{lanford}, and were later completed \cite{uc,cerillpul,cegepe,gasrte}. A crucial insight in the derivation is that the joint distribution function only needs to factorize for particles that are \emph{about to collide}. The derivation of the classical Boltzmann equation remains an extraordinarily active area of research, see, e.g., \cite{ammipa,bogasrsi2023statistical,bogasrs21,den,gerasimenko2022propagation,pusi}. 
    \par For a quantum gas or fluid, Nordheim \cite{nordheim} proposed a quantum analogue of \eqref{eq-cBE}, for which the collision operator $Q$ is given by
    \begin{align}
        \begin{split}\label{eq-qBE}
        \mh{Q^{\mathrm{(lit)}}_4}[f](p) \,= 
        & \int \dx{\bp_4} \delta(p_1+p_2-p_3-p_4) \delta(E(p_1)+E(p_2)-E(p_3)-E(p_4))\\ 
        & |\cM_{22} (\bp_4)|^2(\delta(p-p_1)+\delta(p-p_2)-\delta(p-p_3)-\delta(p-p_4)) \\
		& \big((1\pm f(p_1))(1\pm f(p_2))f(p_3)f(p_4)-f(p_1)f(p_2)(1\pm f(p_3))(1\pm f(p_4))\big) \, ,
        \end{split}
    \end{align}
    where '+' refers to the case of bosons, and '-' to fermions. Boldface letters with subscripts will denote multivectors, such as $\bp_k=(p_1,p_2,\ldots,p_k)$, where $p_j\in\R^3$, and $k\in\N$. $\cM_{22} (\bp_4)$ is the scattering cross section relevant for this process. In the case of bosons, it has been shown \cite{esve2014,esve2015,lu14bos} that the solution to the quantum Boltzmann equation with collision operator \eqref{eq-qBE} develops a $\delta$-mass in finite time, corresponding to condensation in finite time, see also \cite{arno15,arno17,algatr} for related works. If we then decompose $f=\freg+n_c\delta$ into its regular and singular part, we have that \begin{equation}
        \mh{Q^{\mathrm{(lit)}}_4}[f] \, = \, n_c\mh{Q^{\mathrm{(lit)}}_3}[\freg] \, + \, \mh{Q^{\mathrm{(lit)}}_4}[\freg] \, - \, n_c\delta \int\dx{q}\mh{Q^{\mathrm{(lit)}}_3}[\freg](q) \, , 
    \end{equation}
    where 
    \begin{align}
        \begin{split}\label{eq-cqBE}
        \mh{Q^{\mathrm{(lit)}}_3}[f](p) \,= \, & \int \dx{\bp_3} \delta(p_1+p_2-p_3) \delta(E(p_1)+E(p_2)-E(p_3))\\ 
        & |\cM_{22} (\bp_3,0)|^2(\delta(p-p_1)+\delta(p-p_2)-\delta(p-p_3)) \\
		& \big((1+ f(p_1))(1+ f(p_2))f(p_3)-f(p_1)f(p_2)(1+f(p_3)) \big) \, .
        \end{split}
    \end{align}
    This leads to the coupled system
    \begin{equation}\label{eq-coupled-qBE}
        \begin{cases}
            \partial_t\mh{\freg_t} \, &= \, \mh{n_c(t)}\mh{Q^{\mathrm{(lit)}}_3}[\mh{\freg_t}] \, + \, \mh{Q^{\mathrm{(lit)}}_4}[\mh{\freg_t}] \\
            \partial_t \mh{n_c(t)} \, &= \, - \mh{n_c(t)} \int\dx{q}\mh{Q^{\mathrm{(lit)}}_3}[\mh{\freg_t}](q) \, .
        \end{cases} 
    \end{equation}
    If $n_c\gg1 $, $\mh{Q^{\mathrm{(lit)}}_3}[\freg]$ determines the leading order dynamics of $\freg$. In this work, we are interested in studying the emergence of \eqref{eq-cqBE} for an interacting quantum Bose gas with initial condensate density \mh{$n_c(0)=N\gg1$} and initial thermal excitation density $\freg_0\sim1$.
    \par In order to study the emergence of \eqref{eq-qBE}, we consider the Hamiltonian
    \begin{equation}\label{eq-gen-Ham}
       H_{N,\vep_1,\vep_2} \, := \, \sum_{j=1}^N\frac{\vep_1^2}2(-\Delta_{x_j}) \, + \, g_{N,\vep_1,\vep_2} \sum_{j<k}^N v\Big(\frac{x_j-x_k}{\vep_2}\Big) 
    \end{equation}
    acting on the bosonic Hilbert space $L^2(\Lambda)^{\otimes_s N}$, where $\Lambda\subseteq \R^3$ is a 3-torus. Given the Schr\"odinger equation
    \begin{equation}\label{eq-schroedinger}
        \mh{i\vep_1\partial_t \Psi_{N,\vep_1,\vep_2,t} \, = \, H_{N,\vep_1,\vep_2} \Psi_{N,\vep_1,\vep_2,t} \, ,}
    \end{equation}
    we are interested in the asymptotic behavior of the Wigner transforms
    \begin{equation}
        \mh{f_{N,\vep_1,\vep_2,t}^{(k)}(\bx_k,\bp_k) :=  \int_{\Lambda^N} \dx{\by_N} e^{i\bp_k\cdot \by_k} \overline{\Psi}_{N,\vep_1,\vep_2,t}(\bx_k+\frac{\vep_1}2\by_k,\by_{N-k})\Psi_{N,\vep_1,\vep_2,t}(\bx_k-\frac{\vep_1}2\by_k,\by_{N-k}),}
    \end{equation}
    where $\by_N=(\by_k,\by_{N-k})$. The propagation of chaos assumption in this case reads
    \begin{equation}
         f_{N,\vep_1,\vep_2,t}^{(k)} \, \approx \, \big(f_{N,\vep_1,\vep_2,t}^{(1)}\big)^{\otimes k} \, .
    \end{equation}
    at positive times $t>0$ for particles about to collide. The first mathematically rigorous results go back to Hugenholtz \cite{hu} and Ho-Landau \cite{hola}, where it was shown that the terms proportional to $g^2$ in the Duhamel expansion of $f_{N,\vep_1,\vep_2,t}^{(1)}$ give rise to a Boltzmann collision term, see also works by Benedetto et al. \cite{BCEP06,BCEP08}. Under the assumption of \mh{quantum propagation of factorization}, more precisely propagation of \emph{restricted quasifreeness}, see  Definition \ref{defi-quasifree} below, Erd\"os-Salmhofer-Yau \cite{ESY} showed that, for mesoscopic times $t\propto g^{-2}$, the second order Duhamel expansion of $f_{N,\vep_1,\vep_2,t}^{(1)}$ yields the quantum Boltzmann equation \eqref{eq-qBE}. Lukkarinen and Spohn \cite{LuSp1} later revisited this idea and stated conditions under which they derive a Boltzmann equation. However, they did not justify assumptions on the growth of certain moments for the evolution of $f^{(k)}_{N,t}$ that would correspond to a rigorous error control in the evolution.
    
    \par X. Chen and Guo \cite{xchguo} showed in the case $\Omega=\R^3$ that, if one assumes a sufficiently regular potential $v$, and if one assumes the convergence and sufficient regularity of the marginals $f_{N,\vep_1,\vep_2,t}^{(k)}$, then, in the mesoscopic weak-coupling regime, i.e., $\vep_1=\vep_2=N^{-1/3}$, $g_{N,\vep_1,\vep_2}=\sqrt{\vep_1}$, the limiting dynamics is given by a classical, \emph{quadratic} Boltzmann equation \eqref{eq-cBE-op}. 
    \mh{In a recent work \cite{xchenholmer2023Bol-derivation}, X. Chen and Holmer proved that, if the regularity $r$ of marginals of the many-body joint distribution is propagated and its non-negativity are propagated, there exists a critical threshold $r_c$ such that the following trichotomy holds: 1) If $r>r_c$, then the limiting dynamics of $f$ is trivial; 2) if $r<r_c$, the effective Boltzmann equation is ill-posed; 3) if $r=r_c$, the limiting Boltzmann dynamics is \emph{quadratic}, i.e., classical, see \eqref{eq-cBE-op}, instead of quantum, see \eqref{eq-qBE}. They study a factorization property that differs from that considered in the present work. Their factorization property, which is formulated in first quantization rather than second quantization, coincides with ours only in the asymptotic limit that the volume goes to infinity.} 
    \par \mh{The existence of both \emph{positive} and \emph{negative} \emph{conditional} results suggests that a more detailed analysis of the fluctuation dynamics is required, in order to understand the validity of (cubic or quartic) quantum Boltzmann dynamics.}

    \par In the case of fermions, C\'ardenas and one of the authors \cite{cardenas2023quantum} established the emergence of quantum Boltzmann fluctuations alongside bosonized self-interaction terms, beyond the Hartree-Fock approximation, while rigorously controling the error. Their result is unconditional and holds for sufficiently short times. The bosonized self-interaction terms have been analyzed in more detail in \cite{BNPSS2021,benedikter2023correlation,christiansen2022effective} in the stationary case, and \cite{BNPSS2021b} for the dynamical case. Their results show that these can be characterized as \emph{Random-Phase approximation}, as introduced by Bohm and Pines \cite{BohmPinesI,BohmPines1952II,bohm1953collective}.

    \par The \emph{mean-field regime} is determined by $\vep_1=\vep_2=1$ and $g=\frac{\lambda}N$, where $N$ is the number of particles. For bosons, a question of special interest is the persistence of a Bose-Einstein condensate, i.e., the $N\gg 1$ asymptotic behavior of the wave function $\Psi_{N,t}$ satisfying \eqref{eq-schroedinger} when, initially, $\Psi_{N,0}\approx \phi_0^{\otimes N}$. Going back to Ginibre-Velo \cite{GV} and Hepp \cite{he}, it has been shown that the wave function remains factorized, i.e., $\Psi_{N,t}\approx \phi_t^{\otimes N}$, where $\phi_t$ satisfies a \emph{nonlinear Hartree} (NLH) or \emph{Schr\"odinger} (NLS) equation, dependent on the scaling regime. In particular, this establishes persistence of the condensate for microscopic times $t=O(1)$. These works have been extended, and convergence has been shown in different topologies \cite{EY,esy2,frgrsc,grmama,KP,Liard,rosc} and for more singular scalings $\vep_2=N^{-\beta}$, $g=N^{3\beta-1}$, $\beta\in[0,1]$ \cite{brennecke2019gross,CPBBGKY,CHPS-1,xchho5,dietzelee2023,esy1,GM2,GM3,pickl1}. The case $\beta=1$ is of particular physical interest, as it describes rare but strong interactions, and is referred to as the \emph{Gross-Pitaevskii regime}. Another related problem is to study the asymptotic behavior of the ground state (energy) of \eqref{eq-gen-Ham}, see, e.g., \cite{basti2023second,brcasc,cena,lenaro,lenasc,lisesoyn,mitrouskas2023exponential,nam2023bogoliubov,nam2023exponential}.
    
    \par In order to study corrections to the leading order NLH/NLS BEC dynamics, it is necessary to include to study the dynamics of pair correlations, corresponding to thermal fluctuations. It has been shown \cite{grmama,GM1,GM2,MPP,NN2,pepiso} that these are governed by the nonlinear \emph{Hartree-Fock-Bogoliubov} (HFB) equations. Bach, Breteaux, Fr\"ohlich, Sigal and one of the authors \cite{BBCFS18,BBCFS22} have shown that the \emph{quasifree approximation} of the full dynamics, i.e., the restriction to states for which quantum propagation of factorization holds, is given by the HFB equations. 

    \par Another vibrant area of research is the (bosonic) ground state problem associated with \eqref{eq-gen-Ham}, see \cite{boccato2023bose,brscsc,caraci2023order,liseyn,lisesoyn,lewin2023positivedensity}. There, corrections to the leading order Gross-Pitaevskii energy are described by the \emph{Lee-Huang-Yang formula}, see, e.g., \cite{fournais_energy_2023,giuliani2009ground}, and higher oder corrections have also been described \cite{bopese,caraci2023order}. 
    
    \par In our previous work \cite{chenhott}, we studied the regime  $g=\frac{\lambda_N}N$ for some $\lambda_N\ll 1$ and $N\gg1$, on a 3-torus $\Lambda$, for mesoscopic times $t\propto \lambda_N^{-2}$. The Bose gas initially consisted of a translation invariant BEC of \emph{density} $N$, and translation invariant thermal fluctuations close to equilibrium of density $1$, described by a \emph{quasifree} state, see Definition \ref{defi-quasifree} below. \mh{In short, a quantum state is quasifree if its observable expectations satisfy a factorization property, that allows to reduce many-body expectations to single-particle expectations. The precise definition requires to properly introduce the mathematical framework needed to describe many-body dynamics, and we refer the reader to Section \ref{sec-results}.}
    \par Assuming $\hat{v}(0)=0$, we showed that, subleading to the HFB fluctuation dynamics, for kinetic times $t\propto\lambda^{-2}$, the density of thermal fluctuations obeys a (mollified) cubic Boltzmann equation with collision operator
    \begin{align}\label{eq-QBE-previous}
    \begin{split}
        \mh{Q^{\mathrm{(mol)}}_3}[f](p) \, = \, \mh{\frac1{N\vol}}&\sum_{\substack{p_j\in\lattice,\\j=1,2,3}}\delta_t(\Omega(p_1)+\Omega(p_2)-\Omega(p_3))\mh{\delta_{p_1+p_2,p_3}|\hat{v}(p_1)+\hat{v}(p_2)|^2}\\
        & \mh{\times \big(\delta_{p,p_1}+\delta_{p,p_2}-\delta_{p,p_3}\big)}\big(\tilde{f}(p_1)\tilde{f}(p_2)f(p_3)-f(p_1)f(p_2)\tilde{f}(p_3)\big) \, ,
        \end{split}
    \end{align}
    where $\lattice$ denotes the momentum (reciprocal) lattice \mh{and $\tilde{f}=1+f$}. These collisions correspond to those, where a BEC particle is either absorbed or emitted in the collision. Here $\delta_t$ is a mollification of a \mh{one-dimensional Dirac-$\delta$.} $\Omega(p)=\sqrt{E(p)(E(p)+2\lambda\hat{v}(p))}$ is the Bogoliubov dispersion, and we choose the BEC wavefunction to be constant equal to \mh{$\vol^{-\frac12}$}. Given a time scale shorter than $O(\vol^{\frac13})$, where we assumed $\Lambda$ to be a \mh{cubic} torus in three dimensions, we showed that the discrete Boltzmann operator in \eqref{eq-QBE-previous} can be approximated by a continuous Boltzmann operator, where the summation \mh{$\frac1{\vol}\sum_{p\in\lattice}$ is replaced by Riemann integrals $(2\pi)^3\int_{\R^3}\dx{p}$, $\delta_t$ by a Dirac-$\delta$ and $\vol\delta_{p,0}$ by a Dirac-$\delta$ $\delta(p)$}. For longer times, interference leads to additional resonance terms. Crucially to our work, we used the fact that the HFB dynamics preserves quasifreeness, and showed that the full dynamics approximately preserves restricted quasifreeness for times $t\propto\lambda^{-2}\propto \big(\frac{\log N}{\log \log N}\big)^\alpha$, $\alpha>0$. Our result was a first unconditional result, \mh{in the sense that propagation of quasifreeness for the considered time is proved}, and it provided a rigorous error control for sufficiently short times.

    \subsection{New contributions}
    
    In the present work, we revisit this problem from a different perspective. We show that corrections to the pair-absorption rate and to the BEC wave function can be absorbed into the HFB equations. This, in turn, leads to a renormalization of the HFB equations describing the coupled dynamics of the BEC wave function with the leading order thermal pair-excitations. After renormalizing, the corrections to the BEC and thermal fluctuation dynamics are given by pure Boltzmann collision terms. More precisely, let $(\phi_{N,t},\gamma_{N,t},\sigma_{N,t})$ denote the renormalized HFB fields, and $\Omega_{N,t}$ be the corresponding renormalized Bogoliubov dispersion. Let $\mh{\ftot_{N,t}}$ denote the total density of the Bose gas, $\mh{\Phitot_{N,t}}$ denote the full BEC wavefunction, and $\mh{\gtot_{N,t}}$ the full pair-absorption rate\mh{, each evolving according to the many-body dynamics.} These fields can be expanded in the form
    \mh{
    \begin{align}\label{eq-cor-exp}
        \begin{pmatrix}
        \Phitot_{N,t}\\
         \ftot_{N,t} \, - \, |\Phitot_{N,t}|^2\delta\\
         \gtot_{N,t} \, - \, (\Phitot_{N,t})^2\delta
        \end{pmatrix}
          \, =& \, \begin{pmatrix}
              \sqrt{N\vol}\phi_{N,t}\\
              \gamma_{N,t}\\
              \sigma_{N,t}
          \end{pmatrix} \, + \, \begin{pmatrix}
              a_{N,t} & 0 \\
              0 & A_{N,t}
          \end{pmatrix}\begin{pmatrix}
              \Phi_{N,t}\\
              f_{N,t}-|\Phi_{N,t}|^2\delta\\
              g_{N,t}-\Phi_{N,t}^2\delta
          \end{pmatrix} \, ,
    \end{align}
    }
    where \mh{$a_{N,t}\in\C$, $A_{N,t}\in\C^{2\times 2}$ depend} on the HFB fields \mh{$(\phi_{N,t},\gamma_{N,t},\sigma_{N,t},\Omega_{N,t})$}, and \mh{$\gamma_{N,t}$} and \mh{$\sigma_{N,t}$} are related by \mh{$|\sigma_{N,t}|^2=\gamma_{N,t}(1+\gamma_{N,t})$}. 
    
    \par \eqref{eq-cor-exp} can be interpreted as a gradual centering process: 
    \begin{enumerate}
        \item We collect all condensate terms to given order in the perturbation expansion and add them to the contribution coming from \mh{$\phi_{N,t}$}, and center the correlation functions \mh{$f_{N,t}$, $g_{N,t}$} w.r.t. \mh{$\phi_{N,t}$}. Then \mh{$\Phi_{N,t}$} denotes the correction terms to given order.
        \item Simultaneously, we subtract the leading order pair-correlations \mh{$(\gamma_{N,t},\sigma_{N,t})$} from the centered pair correlations. Then \mh{$f_{N,t},g_{N,t}$} denote the pair correlation corrections to given order. 
    \end{enumerate}
    We show that, \emph{if} \mh{$(\phi_{N,t},\gamma_{N,t},\sigma_{N,t},\Omega_{N,t})$} satisfy the renormalized HFB equations
    \mh{
    \begin{align}
    \begin{cases}\label{eq-HFB-ren-intro}
	i\partial_t \phi_{N,t}  \, &= \,  \frac{\lambda}{N} \big[ \big(\Gamma_{N,t}*(\hat{v}+\hat{v}(0))\big) (0) \phi_{N,t} \, + \, (\Sigma_{N,t}*\hat{v}) (0) \overline{\phi}_{N,t}\big] \, - \, 2\lambda\vol\hat{v}(0) |\phi_{N,t}|^2 \phi_{N,t} \, ,  \\
	\partial_t \gamma_{N,t} \, &= \,\frac{2\lambda}N \Im\big((\Sigma_{N,t}*\hat{v} ) \overline{\sigma}_{N,t}\big) \, , \\
	i\partial_t \sigma_{N,t} \, &= \, 2\big(E+\frac{\lambda}{N}\Gamma_{N,t}*(\hat{v}+\hat{v}(0))\big)\sigma_{N,t} \, + \, \frac{\lambda}{N} \big( \Sigma_{N,t}* \hat{v}\big) (1+2\gamma_{N,t}) \, ,\\
    \Omega_{N,t} \, &= \, E+\frac{\lambda}N\big(\Gamma_{N,t}*(\hat{v}+\hat{v}(0))\big) \, + \, \frac{\lambda}N\frac{\Re((\overline{\Sigma}_{N,t}*\hat{v})\sigma_{N,t})}{1+\gamma_{N,t}} \, ,
	\end{cases}
    \end{align}
    }
    where 
    \begin{align}\label{def-tot-HFB-intro}
        \begin{pmatrix}
         \mh{\Gamma_{N,t}} \\
         \mh{\Sigma_{N,t}}
        \end{pmatrix}(p) \, = \, N\vol\delta(p)\begin{pmatrix}
            |\mh{\phi_{N,t}}|^2\\
            \mh{\phi_{N,t}}^2
        \end{pmatrix} \, + \, \big((1+f_0(p)+f_0(-p))\begin{pmatrix}
            \mh{\gamma_{N,t}}\\\mh{\sigma_{N,t}}
        \end{pmatrix}(p) \, + \, \begin{pmatrix}
            \frac{f_0(p)+f_0(-p)}2 \\ 0
        \end{pmatrix} \, ,
    \end{align} 
    and $E(p)=\frac12|p|^2$, \emph{then} \mh{$(\Phi_{N,t},f_{N,t},g_{N,t})$} satisfy
    \begin{align}
                \mh{\Phi_{N,t}} - \Phi_0 \, & = \, \frac1{N^{\frac32}}\int_0^t\dx{s}\big(Q_3^{(\Phi)}[\mh{f_{N,\cdot}}](s)+Q_{3,3}^{(\Phi)}[\mh{f_{N,\cdot}}](s)\big) \, + \, O\Big(\frac{e^{c\vol\lambda t}}{N^2}\Big) \, , \label{eq-Phi-QBE-intro}\\
                \mh{f_{N,t}}-f_0 \, &= \, \frac1N\int_0^t\dx{s} Q_3[\mh{f_{N,\cdot}}](s) \, + \, O\Big(\frac{e^{c\vol\lambda t}}{N^{\frac32}}\Big) \, , \label{eq-f-QBE-intro}\\
                \mh{g_{N,t}}-g_0 \, &= \, \frac1N\int_0^t\dx{s} Q_3^{(g)}[\mh{f_{N,\cdot}}](s)\, + \, O\Big(\frac{e^{c\vol\lambda t}}{N^{\frac32}}\Big) \, ,\label{eq-g-QBE-intro}
            \end{align}
            where each $Q_j^{(k)}$ denotes a collection of cubic Boltzmann collision terms, for which the collision kernels depend on the renormalized HFB fields \mh{$(\phi_{N,t},\gamma_{N,t},\sigma_{N,t},\Omega_{N,t})$}. Crucially, the errors provided here are \emph{sharp} in orders of $N$.
            
            \par In particular, we show that, to leading orders, the BEC wave function \mh{$\Phitot_{N,t}$} and the pair correlations \mh{$(\ftot_{N,t},\gtot_{N,t})$} can each be decomposed into a part that satisfies renormalized HFB equations and a quantum Boltzmann correction part. This separation confirms the phenomenological paradigm that the HFB and QBE dynamics evolve according to different time scales; the HFB dynamics leads to fast oscillations, while the QBE determines the slow long-time dynamics. As in our previous work, our result is unconditional and rigorous, however, we extend the validity to times of order $t\propto \lambda^{-2}\sim (\log N)^2$, in contrast to \mh{$t\propto\lambda^{-2}\sim (\log N/\log\log N)^2$} before.

            \par The term $N^{-3/2}\int_0^t\dx{s} Q_3^{(\Phi)}[\mh{f_{N,\cdot}}](s)$ in the evolution of the BEC wave function justifies the Boltzmann term in the evolution of the condensate density, see \eqref{eq-coupled-qBE}. However, the additional $N^{-3/2}\int_0^t\dx{s} Q_{3,3}^{(\Phi)}[\mh{f_{N,\cdot}}](s)$ has previously not been included, despite of being of the same order of magnitude as the cubic Boltzmann term, see also \cite{potr20,retr,ZNG}. 
        
            \par Bo\ss mann et al. \cite{BPaPS,BoßmannPetrat2023WEef,BPePS} have computed a full expansion of the correction dynamics in terms of the coupling constant. In comparison, in our current and past work \cite{chenhott}, we have been able to characterize terms in the expansion as stemming either from a HFB contribution, or a quantum Boltzmann correction.
            
            \par We are able to rigorously control the error for times $t\propto
            \lambda^{-2}\sim(\log N)^2$. A more detailed calculation shows that, at order $\frac1{N^2}$, terms emerge that are neither of HFB nor of quantum Boltzmann type. 

    \section{Statement of results\label{sec-results}}

    \par \mh{We now introduce the second-quantization formalism underlying the results of this paper.}
    
	\par Let $\Lambda=[-L/2,L/2]^3/\sim$ denote a cubic torus of length $L$, and let $\lattice=(\frac{2\pi}L\Z)^3$ denote its reciprocal space. Let
    \begin{align}
        \cF \, := \, \C\oplus\bigoplus_{n\in\N} L^2(\Lambda)^{\otimes_{s}n}
    \end{align}
    denote the \emph{bosonic Fock space}, endowed with the inner product
    \begin{align}
        \scp{\Phi}{\Psi}_{\cF} \, = \, \sum_{n\in\N_0}\scp{\Psi^{(n)}}{\Phi^{(n)}}_{L^2(\Lambda^n)} \, .
    \end{align}
    For $\psi\in L^2(\Lambda)$, let
    \begin{align}
        (a(\psi)\Psi)^{(n-1)}(\bx_{n-1}) \, = \, \sqrt{n}\int_\Lambda\dx{x}\overline{\psi}(x)\Psi^{(n)}(x,\bx_{n-1})
    \end{align}
    denote the annihilation operator, and
    \begin{align}
        (\ad(\psi)\Psi)^{(n+1)} \, &:= \, \sqrt{n+1}P_{L^2(\Lambda)^{\otimes_{s} (n+1)}}\psi\otimes \Psi^{(n)} \\
        &= \, \frac{\sqrt{n+1}}{(n+1)!}\sum_{\pi\in \cS_{n+1}}\psi(x_{\pi(1)})\Psi^{(n)}(x_{\pi(2)},\ldots,x_{\pi(n+1)}) \, ,
    \end{align}
    the creation operator. In addition, we introduce the momentum-space annihilation/creation operators $a_p^{\#}:=a^{\#}(e^{ip\cdot})$. They satisfy the CCR
    \begin{align}\label{eq-CCR}
        [a_p,a_q] \, = \, [\ad_p,\ad_q] \, = \, 0 \, , \quad [a_p,\ad_q] \, = \, \vol \delta_{p,q} \, =: \, \delta_{\lattice}(p-q) \, .
    \end{align}
    In the following, we will omit the subscript $\lattice$, when referring to a momentum-$\delta$. For brevity, we will use the notation
    \begin{equation}\label{def-mom-int}
        \int_{\lattice}\dx{p} h(p) \, := \, \frac1{\vol}\sum_{p\in\lattice}h(p) \, .
    \end{equation}
    Again, we will omit the subscript $\lattice$, unless it is ambiguous. Our convention for the Fourier transform is
    \begin{align}\label{def-FT}
        \hat{h}(p) \, := \, \int_{\Lambda}\dx{x}e^{ip\cdot x}h(x) \, .
    \end{align}
    \par We consider the mean-field Hamiltonian
	\begin{align}\label{def-hamiltonian} 
	\cH_N \, := \,   \int \dx{p} E(p)\ad_pa_p +\frac\lambda{2N}\int\dx{\bp_4}\delta(p_1+p_2-p_3-p_4) \hat{v}(p_1-p_3) \, \ad_{p_1}\ad_{p_2}a_{p_3}a_{p_4} \, ,
	\end{align}
    where $E(p):=\frac12|p|^2$ denotes the free dispersion, and $v\geq0$ a pair potential such that $\hat{v}\geq0$. 
    \par We are interested in determining the leading order dynamics for a Bose gas governed by the Hamiltonian $\cH_N$. As an initial state, we will choose a Bose-Einstein condensate (BEC) of density $N$, surrounded by a thermal excitations that are described by a \emph{quasifree} state with density $O(1)$. 

    \subsection{Initial state} 
    
    Let 
	\begin{align}\label{def-weyl}
	\weyl[f] \, := \, \exp\big(\ad(f)-a(f)\big)
	\end{align} 
	denote the Weyl \mh{operator}, and define \mh{Bogoliubov operator}
	\begin{align}\label{def-bogrot}
	\bog[k] \, := \, \exp\Big(\frac12\int \dx{p} \big(k(p)\ad_p\ad_{-p}-\overline{k}(p)a_pa_{-p}\big)\Big) 
	\end{align}
    in the translation invariant case. Without loss of generality, we assume that $k$ is even. We then obtain the Weyl transform
	\begin{align}
		\weyld[f]a_p\weyl[f] \, = \,  a_p \, + \, \hat{f}(p) \, , \label{eq-weyl}
    \end{align}
    and the Bogoliubov transform
    \begin{align}
		\bogd[k]a_p\bog[k] \, = \,  \cosh\big(|k(p)|\big)a_p \, + \, \sinh\big(|k(p)|\big)\frac{k(p)}{|k(p)|}\ad_{-p} \, , \label{eq-bog-rot-raw}
	\end{align}
	see Lemma \ref{lem-bogoliubov}. \mh{As a result, the creation and annihilation operators defined by \eqref{eq-weyl} and \eqref{eq-bog-rot-raw} describe thermal fluctuations outside the BEC.}
    \par Let $\alg$ denotes the CCR algebra generated by the Weyl operators $\weyl[\psi]$, $\psi\in L^2(\Lambda)$, see \cite[section 5.2.3]{BR2}.
	
	\begin{definition}[Quasifree state]\label{defi-quasifree}
		Let $\jb{\cdot}$ be a state and 
		\begin{align}
		\jb{A}^{\mathrm{(cen)}} \, := \, \jb{(\weyl[\jb{a})]A\weyld[\jb{a}])}
		\end{align} 
		denote its centering. We say $\jb{\cdot}$ is \emph{quasifree} iff for all $n\in\N$
		\begin{equation} \label{eq-wick-thm}
			\begin{cases}
				\jb{a^{\#_1} a^{\#_2}\ldots a^{\#_{2n}}}^{\mathrm{(cen)}} \, &= \,  
				\wick{
					\c1{a}^{\#_1} \c2{a}^{\#_2} \c2{.}\c2{.}\c1{.} \c2{a}^{\#_{2n}}
				}\\
				& \quad +\mbox{all\ pair\ contractions}\\ 
				\jb{a^{\#_1} a^{\#_2}\ldots a^{\#_{2n-1}}}^{\mathrm{(cen)}}
    \, &= \, 0 
			\end{cases} \, ,			
		\end{equation}
		where $\wick{\c a^{\#_1} \c a^{\#_2}}:=\jb{a^{\#_1}a^{\#_2}}^{\mathrm{(cen)}}$. \eqref{eq-wick-thm} is referred to as {\normalfont Wick's Theorem}.
        \par A state is \emph{restricted quasifree} if \eqref{eq-wick-thm} holds for $n\leq n_0$ for some $n_0\in\N$.
	\end{definition}
\mh{
    \begin{remark}
        $\wick{\c a^{\#_1} \c a^{\#_2}}:=\jb{a^{\#_1}a^{\#_2}}^{\mathrm{(cen)}}$ accounts for a \emph{contraction}. The notation makes it easier to keep track of different terms in the Duhamel expansion.  
    \end{remark}
}
	\begin{definition}[Number conserving state]
		A state $\jb{\cdot}$ is called {\normalfont number conserving} iff $\jb{[A,\nb]} = 0$ for every observable $A\in\alg$.
	\end{definition}

    \begin{definition}[Translation invariance]
		A state $\jb{\cdot}$ is called {\normalfont translation invariant} iff \begin{equation}
		  \jb{\prod_{j=1}^na_{p_j}^{(\sigma_j)}} = \frac{\delta(\sum_{j=1}^n\sigma_jp_j)}\vol \jb{\prod_{j=1}^na_{p_j}^{(\sigma_j)}}    
		\end{equation}
        for all $p_1,\ldots,p_n$ and $\sigma_j=\pm1$.
	\end{definition}
	
	Let $\jb{\cdot}_0$ be a number conserving, quasifree, and translation invariant. More precisely, let
	\begin{align}\label{def-K}
	\cK \, := \, \int_{\lattice} \dx{p} K(p) \ad_pa_p 
	\end{align}
	be such that $K(p)\geq\kappa_0$ for some $\kappa_0>0$, and let
    \begin{align}\label{def-jb0}
        \jb{A}_0 \, := \, \frac{\tr(e^{-\cK}A)}{\tr(e^{-\cK})} \, .
    \end{align}
    Observe that, by being number conserving, $\jb{\cdot}_0$ is already centered. $\jb{\cdot}_0$ determines the thermal excitations beyond the HFB fluctuations.
	
	\par The full state describing the Bose gas is then given by
	\begin{align} \label{def-rho0}
		\mh{\totst{A}_{N,0}} \, := \, \jb{e^{-\cK}\bogd[k_0]\weyld[\sqrt{N\vol}\phi_0] A \weyl[\sqrt{N\vol}\phi_0]\bog[k_0]}_0
	\end{align} 
	for all $A\in\alg$. Note that \mh{$\totst{\cdot}_{N,0}$} is quasifree. \mh{Our initial state is chosen close to a Gibbs state, see, e.g., \cite[Theorem 2.3]{lenaseso}.} The initial value problem (IVP) associated with the Hamiltonian $\cH_N$ and the initial state $\totst{\cdot}_0$ is then given by the Liouville-von Neumann equation
    \mh{
	\begin{align}\label{eq-schroedinger-non-rel}
	i\partial_t \totst{A}_{N,t} \, = \, \totst{[A,\cH_N]}_{N,t}
	\end{align}
    }
	for all observables $A\in\alg$. Below, we impose assumptions on $v$ ensuring that $\cH_N$ is self-adjoint and that it induces a unitary evolution $e^{-it\cH_N}$. Notice that $\totst{\cdot}_t$ is not quasifree. However, we will show that for short enough times, $\totst{\cdot}_t$ is approximately quasifree. Consequently, the evolution of any expectation $\totst{A}_t$ is fully characterized by the first and second moments $\totst{a_0}_t$, and $\totst{\ad_p a_p}_t$, $\totst{a_pa_{-p}}_t$, respectively. 
    \par Of particular interest in the present work is the evolution of the density 
    \mh{
    \begin{align}\label{def-ftot}
        \ftot_{N,t}(p) \, := \, \frac{\totst{\ad_pa_p}_{N,t}}\vol \, .
    \end{align}
    }
    In order to study \mh{$\ftot_{N,t}$}, we will decompose the full dynamics into a BEC and a thermal fluctuation part. For the latter, we will further decompose the dynamics of the fluctuation particles into a quasifree part, which we will refer to as \emph{Hartree-Fock-Bogoliubov fluctuations}, and a collision part, which we will identify as the \emph{quantum Boltzmann fluctuations}.
    \par Our approach includes the time-behavior of leading order pair correlations via Bogoliubov rotations as employed by Grillakis-Machedon et al. \cite{grmama,GM1,GM2} to derive the leading order HFB dynamics. However, we extend the latter by a quasifree, number conserving, centered (translation invariant) state. It is the presence of these additional excitation states that allow for a (cubic) Boltzmann equation to arise.  

    \par \mh{Subsequently, we will suppress the dependence of states and fields on the BEC density $N$ from the notations for $(\Phitot,\ftot,\gtot)$ and the relative fields $(\Phi,f,g)$.}
    
    \subsection{Fluctuation dynamics} 
    \mh{In order to extract the evolution of thermal fluctuations surrounding the BEC, we pass to the relative evolution relative the BEC. To this end, we define the Bogoliubov propagation}
	\begin{align}\label{def-bogprop}
	\begin{cases}
		i\partial_t \bprop(t) &= \int \dx{p} \Omega_t(p)a_p^\dagger a_p \bprop(t)\, ,\\
		\bprop(0) &= \1 \, ,
	\end{cases}
	\end{align} 
	where $\Omega_t$ is an even function. We have that
	\begin{align}
	\bpropd(t)a_p\bprop(t) \, = \, e^{-i\int_0^t \dx{s} \Omega_s(p)}a_p \,. \label{eq-bog-prop}
	\end{align} 
	\mh{Recalling the Weyl transform and the Bogoliubov transform obtained from conjugation by $\weyl$ and $\bog$, respectively,} we introduce the \mh{evolution operator for the} \emph{fluctuation dynamics} defined by
	\begin{align}\label{def-fluc}
	\fluc(t) \, := \, e^{iS_t}\bpropd(t)\bogd[k_t]\weyld[\sqrt{N\vol}\phi_t]e^{-it\cH_N}\weyl[\sqrt{N\vol}\phi_0]\bog[k_0] \, ,
	\end{align} 	
	where $(\phi_t,k_t)$ and $S_t$ will be determined below. We will choose $\phi_t$ and $k_t$ in such a way that the leading order and next-to-leading order contributions of the full dynamics $e^{-i\cH_Nt}$ are determined by $(\phi_t,k_t)$.
 
    \par We assume that $\phi_t(x)\equiv \phi_t$ is translation invariant. We show in Lemma \ref{lem-hfluc} that the fluctuation dynamics satisfies
	\begin{align} \label{eq-fluc-dyn}
	\begin{cases}
		i\partial_t\fluc(t) & =  \Hfluc(t)\fluc(t) \, , \\
		\fluc(0) &= \1 \, ,
	\end{cases}		
	\end{align} 
	where\mh{, for an appropriate choice of $S_t$, $\Hfluc(t)$ takes the form}
	\begin{align}
		\Hfluc(t) \, = \, \HBEC(t)  \, + \, \HHFB(t) \, + \, \Hcub(t) \, + \, \Hquart(t) \, . \label{eq-Hfluc-decomp-0}
	\end{align} 
	Each of the terms in $\Hfluc(t)$ is a normal-ordered polynomial in $a$ and $\ad$, and their explicit expressions are given in Lemma \ref{lem-hfluc}. Here a monomial in $a$ and $\ad$ is {\it normal-ordered} iff all creation operators $\ad$ are on the left of all annihilation operators $a$. We choose $S_t$ in such a way that it absorbs all scalar terms. $\HBEC(t)$ denotes the BEC Hamiltonian and it is linear in $a^{\#}$, $\HHFB$ is the HFB Hamiltonian, which is quadratic in $a^{\#}$, $\Hcub(t)$ is cubic in $a^{\#}$ and accounts for cubic scattering processes, where one of the particles is being absorbed into or emitted from the BEC, and $\Hquart(t)$ is quartic in $a^{\#}$ and describes pair interactions.

	 Let 
	\begin{align}
		u_t(p) := & \cosh(|k_t(p)|) \, , \label{def-up} \\
		v_t(p) := & \sinh(|k_t(p)|)\frac{k_t(p)}{|k_t(p)|} \, . \label{def-vp}
	\end{align} 
    In particular, we can rewrite \eqref{eq-bog-rot-raw} as 
    \begin{align}\label{eq-bog-rot}
       \bogd[k_t]a_p\bog[k_t] \, = \, u_t(p)a_p \, +\, v_t(p)\ad_{-p} \, . 
    \end{align}
	Note that we have $u_t(p)^2-|v_t(p)|^2=1$. In addition, we introduce the scalar fields
	\begin{align}
		\gamma_t(p) :=&  |v_t(p)|^2 \, , \label{def-gamma} \\
		\sigma_t(p) :=&  u_t(p)v_t(p) \, . \label{def-sigma}
	\end{align} 
	Observe that, from this definition, we have the relation
	\begin{align}
	|\sigma_t(p)|^2 \, = \, u_t(p)^2 |v_t(p)|^2 \, = \, (1+\gamma_t(p))\gamma_t(p) \, . \label{eq-gamma-sigma-rel-0}
	\end{align} 
	Then the expressions for each of the normal-ordered terms in $\Hfluc(t)$ in \eqref{eq-Hfluc-decomp-0} are given in Lemma \ref{lem-hfluc} in the Appendix.
	\par We introduce the time-evolved \textit{relative} state
    \begin{equation}\label{def-rel-dynamics}
        \jb{A}_t \, := \, \jb{\flucd(t)A\fluc(t)}_0 \, .
    \end{equation} 
    It satisfies
    \begin{equation}\label{eq-rel-Liouville}
        i\partial_t\jb{A}_t \, = \, \jb{[A,\Hfluc(t)]}_t \, .
    \end{equation}
    We introduce the relative moments
    \begin{align}\label{def-rel-moments}
        \Phi_t \,:= \, \frac{\jb{a_0}_t}\vol \, , \quad f_t(p) \, := \, \frac{\jb{\ad_p a_p}_t}\vol \, , \quad g_t(p) \, := \, \frac{\jb{a_pa_{-p}}_t}\vol \, . 
    \end{align}
    Our choice of $\jb{\cdot}_0$ implies $\Phi_0=g_0=0$. Then we can rewrite the total density $\ftot$, see \eqref{def-ftot} for its definition and section \eqref{sec-dens-exp} for the derivation, as
    \begin{align}
        \ftot_t(p) \, =& \, \delta(p)\Big[N\vol|\phi_t|^2 \, + \, \Big(\sqrt{N\vol}e^{i\int_0^t \dx{s} \Omega_s(0)}\big(\phi_tu_t(0)+\overline{\phi}_tv_t(0)\big)\overline{\Phi}_t \, + \, \mathrm{h.c.}\Big)\Big] \\
        & \, + \, \gamma_t(p) \, + \, \big(1+\gamma_t(p)\big)f_t(p) \, + \, \gamma_t(p)f_t(-p) \, + \, \Big(e^{2i\int_0^t \dx{s} \Omega_s(p)}\sigma_t(p) \overline{g}_t(p)+\mathrm{h.c.}\Big) \, .
    \end{align}
    Our goal is to show that if $(\phi_t,\gamma_t,\sigma_t,\Omega_t)$ satisfy the \emph{renormalized} Hartree-Fock-Bogoliubov (HFB) equations, then the dynamics of $(\Phi_t,f_t,g_t)$ each are determined, to leading order, by a Boltzmann equation. More precisely, we will show that, to leading order, the dynamics of $f$ is given by a cubic Boltzmann equation, while the dynamics of $\Phi$ and $g$, to leading order, are driven by $f$ via a collision term, see \eqref{eq-Phi-QBE-intro}--\eqref{eq-g-QBE-intro}.

    \subsection{Main result}

    \mh{As a final step, we establish well-posedness of the leading order HFB dynamics, describing the BEC and the leading order thermal fluctuations. This is required for the statement of our main result.}
    \subsubsection{Renormalized HFB equations}
    
    Let
    \begin{align}
	\fplus(p) \, := \, \frac12\big(f_0(p)+f_0(-p)\big) \label{def-fplus}
	\end{align} 
	denote the even symmetrization of $f_0$. We introduce the (second order) renormalized HFB fields
    \begin{align}
	\Gamt :=&  (1+2\fplus)\gamma \, + \, \fplus \, + \, N\vol|\phi|^2\delta \, , \label{def-Gamt-intro} \\
	\Sigt :=&  (1+2\fplus)\sigma \, + \, N\vol\phi^2\delta \, . \label{def-Sigt-intro}
	\end{align}
    Then the renormalized HFB equations read
    \begin{align}
    \begin{cases}\label{eq-HFB-ren-2-intro}
	i\partial_t \phit_t  \, &= \,  \frac{\lambda}{N} \Big( \big(\Gamt_t*(\hat{v}+\hat{v}(0))\big) (0) \phit_t \, + \, (\Sigt_t*\hat{v}) (0) \phitb_t\Big) \\
	& \qquad - \, 2\lambda\vol\hat{v}(0) |\phit_t|^2 \phit_t \, ,  \\
	\partial_t \gamt_t \, &= \,\frac{2\lambda}N \Im\big((\Sigt_t*\hat{v} ) \sigtb_t\big) \, , \\
	i\partial_t \sigt_t \, &= \, 2\big(E+\frac{\lambda}{N}\Gamt_t*(\hat{v}+\hat{v}(0))\big)\sigt_t \, + \, \frac{\lambda}{N} \big( \Sigt_t* \hat{v}\big) (1+2\gamt_t) \, ,
	\end{cases}
    \end{align}
    and the corresponding Bogoliubov dispersion is given by
    \begin{align}\label{def-ren-Bog-disp-2}
        \Omt_t \, := \, E+\frac{\lambda}N\big(\Gamt_t*(\hat{v}+\hat{v}(0))\big) \, + \, \frac{\lambda}N\frac{\Re\big(\Sigtb*\hat{v})\sigt_t\big)}{1+\gamt_t} \, .
    \end{align}

    \begin{remark}\label{rem-mod-bog-disp}
        Observe that $\Omt_t$ is a \emph{modified} Bogoliubov dispersion. Instead, the \emph{regular} Bogoliubov dispersion, to leading order, is given by $\Omega_{\mathrm{Bog}}=\sqrt{E(E+2\lambda \hat{v})}$, where $E(p)=|p|^2/2$ is the free dispersion, see above and, e.g., \cite{chenhott,pepiso}.
    \end{remark}
    
    Omitting superscripts, we note that the HFB equations \eqref{eq-HFB-ren-2-intro} can be rewritten as
    \begin{align}\label{eq-HFB-system-intro}
	\begin{cases}
		i\partial_t \phi_t  \, &= \, \frac{\lambda}{N} \Big( \big(\Gamma_t*(\hat{v}+\hat{v}(0))\big) (0) \phi_t + (\Sigma_t*\hat{v}) (0) \overline{\phi}_t\Big)- 2\lambda\vol\hat{v}(0) |\phi_t|^2 \phi_t \, ,  \\
	\partial_t \Gamma_t \, &= \,  \frac{2\lambda}N \Im\big((\Sigma_t*\hat{v} ) \overline{\Sigma}_t\big) \, , \\
	i\partial_t \Sigma_t \, &= \, 2\big(E+\frac{\lambda}{N}\Gamma_t*(\hat{v}+\hat{v}(0))\big)\Sigma_t \, + \, \frac{\lambda}{N} \big( \Sigma_t* \hat{v}\big) (1+2\Gamma_t) \\
    &\quad \quad - \, 4N\lambda\vol^2\hat{v}(0) |\phi_t|^2 \phi_t^2\delta \, .
	\end{cases}	
 	\end{align}
    \mh{For any $1\leq a< \infty$, we introduce the rescaled $L^a(\lattice)$-norms
    \begin{align}
        \|f\|_{L^a(\lattice)} \, := \, \|f\|_a \, := \, \vol^{-\frac1a}\|f\|_{\ell^a(\lattice)}, \quad \|f\|_{L^\infty(\lattice)} \, := \, \|f\|_\infty:=\|f\|_{\ell^\infty(\lattice)} \, . \label{eq-Lp-norm}
    \end{align}
    }
    For any weight $\widetilde{w}:\lattice\to \R^+$, define the weighted $L^r$ space
	\begin{align}
		\mh{L^r_{\widetilde{w}}(\lattice) \, := \, \{f:\lattice \to \R \mid \|\widetilde{w}^{\frac1r}f\|_{L^r(\lattice)}<\infty\}\, , \label{def-energy-Lp}}
	\end{align} 
	endowed with the norm
	\begin{align}
		\mh{\|f\|_{L^r_{\widetilde{w}}} \, := \, \|\widetilde{w}^{\frac1r}f\|_r \, .}
	\end{align} 
    \mh{
    Based on the norms \eqref{eq-Lp-norm}, we abbreviate
    \begin{align}
        \nd{f} \, &:= \, \|f\|_1 \, + \, \|f\|_\infty \, ,\\
        \wnd{f} \, &:= \, \|f\|_{L^1_{\sqrt{1+E}}} \, + \, \|f\|_\infty \, .
    \end{align}
    }
	For all $j\in\N_0$, we introduce the function spaces
	\begin{align}
		\hfbspace^j := & \C\times\big(L^1_{(1+E)^j}(\lattice)\cap L^\infty(\lattice)\big)\times \big(L^2_{(1+E)^j}(\lattice)\cap L^\infty(\lattice)\big) \, , \label{def-hfbspace}
	\end{align} 
	endowed with the norm
	\begin{align}
		\|(\phi,\Gamma,\Sigma)\|_{\hfbspace^j} \, := \, |\phi| \, + \, \|\Gamma\|_{L^1_{(1+E)^j}} \, + \, \|\Gamma\|_\infty \, + \, \|\Sigma\|_{L^2_{(1+E)^j}} \, + \, \|\Sigma\|_\infty \, .
	\end{align} 
	In addition, we define
	\begin{align}
		\hfbspace^{-j} \, := \, (\hfbspace^j)' \, ,
	\end{align}  
	in the sense of Banach space duals.

	\begin{definition}[Mild solution]
		We call $(\phi,\Gamma,\Sigma)$ a {\normalfont mild solution} of \eqref{eq-HFB-system-intro} with initial datum $(\phi_0,\Gamma_0,\Sigma_0) \in \hfbspace^1$ iff there exists $T>0$ such that $(\phi,\Gamma,\Sigma)\in C^0_t\big([0,T),\hfbspace^1\big)\cap C^1_t\big([0,T),\hfbspace^{-1}\big)$ satisfies
		\begin{align} 
			\phi_t \, = \, &  \phi_0 \, - \, i \int_0^t \dx{s} \Big[\frac{\lambda}{N} \Big( \big(\Gamma_s*(\hat{v}+\hat{v}(0))\big) (0) \phi_s \, + \, (\Sigma_s*\hat{v}) (0) \overline{\phi}_s\Big) \\
			& - \,  2\lambda\vol\hat{v}(0) |\phi_s|^2 \phi_s\Big] \, , \label{eq-phi-HFB-mild} \\
			\Gamma_t \, = \, &  \Gamma_0 \, - \, \frac{2\lambda}N \int_0^t \dx{s} \Im\Big(\big(\overline{\Sigma}_s*\hat{v} \big) \, \Sigma_s\Big) \, , \label{eq-gamma-HFB-mild}\\
			\Sigma_t \, = \, &  e^{-2iEt}\Sigma_0 \, - \, i\int_0^t \dx{s} e^{-2iE(t-s)}\Big[2\big(\frac{\lambda}{N}\Gamma_s*(\hat{v}+\hat{v}(0))\big)\Sigma_s \\
			& + \, \frac{\lambda}{N} \big( \Sigma_s* \hat{v}\big) (1+2\Gamma_s) \, - \, 4N\lambda\vol^2\hat{v}(0)|\phi_s|^2\phi_s^2 \, \delta\Big] \label{eq-sigma-HFB-mild}
	\end{align} 
	for all $t\in[0,T)$.
	\end{definition}

    For the next result, we introduce the truncated fields
    \begin{align} \label{def-gtr-str}
        \gtr \, := \, \Gamma-N\vol|\phi|^2\delta \, , \quad \str \, := \, \Sigma-N\vol \phi^2\delta \, .
    \end{align} 
    
    \begin{proposition}[Global well-posedness] \label{prop-HFB-gwp}
		Assume that $\hat{v}\in L^1_{\sqrt{1+E}}\cap L^\infty(\lattice)$, and that $v\geq0$. Let $(\phi_0,\Gamma_0,\Sigma_0)\in\hfbspace^1$, such that $\gtr_0\geq0$ and $|\str_0|^2\leq (\gtr_0+1)\gtr_0$. Let $(\phi,\Gamma,\Sigma)\in C^0_t\big([0,T_0),\hfbspace^1\big)\cap C^1_t\big([0,T_0),\hfbspace^{-1}\big)$ be the associated unique maximal mild solution of \eqref{eq-HFB-system-intro} with existence time $T_0>0$. Then $T_0=\infty$, and $\gtr_t\geq0$ and $|\str_ t|^2\leq (\gtr_t+1)\gtr_t$ for all $t\geq0$.
	\end{proposition}
    
    \subsubsection{Boltzmann collision kernels}

    \mh{Our main result, Theorem \ref{thm-main}, depends on quantum Boltzmann collision kernels, which we introduce here. These are expressed in terms of the HFB fields $(\phi,\gamma,\sigma,\Omega)$, controlled by Proposition \ref{prop-HFB-gwp}.}

    \par In Lemma \ref{lem-hfluc}, we compute the cubic collision kernels
	\begin{align} 
        &\begin{split}\label{def-bbf03}
	\MoveEqLeft\bbf{0}{3}_t(\bp_3) \, :=\\
	&\sqrt{\vol}\Big(\big( u_t(p_1)u_t(p_2)v_t(p_3)\phi_t + v_t(p_1)v_t(p_2)u_t(p_3)\overline{\phi}_t\big)\big(\hat{v}(p_1)+\hat{v}(p_2)\big)\\
	& + \, \big( v_t(p_1)u_t(p_2)u_t(p_3)\phi_t + u_t(p_1)v_t(p_2)v_t(p_3)\overline{\phi}_t\big)\big(\hat{v}(p_2)+\hat{v}(p_3)\big)\\
	& + \, \big( u_t(p_1)v_t(p_2)u_t(p_3)\phi_t + v_t(p_1)u_t(p_2)v_t(p_3)\overline{\phi}_t\big)\big(\hat{v}(p_1)+\hat{v}(p_3)\big) \Big) \, ,
    \end{split}\\
    &\begin{split}\label{def-bbf12}
	\MoveEqLeft \bbf{1}{2}_t(\bp_3) \, := \\
	&\sqrt{\vol}\Big(\big(u_t(p_1)u_t(p_2)u_t(p_3)\phi_t+v_t(p_1)v_t(p_2)\overline{v_t}(p_3)\overline{\phi}_t\big)\big(\hat{v}(p_1)+\hat{v}(p_2)\big)\\
	& + \big(v_t(p_1)u_t(p_2)\overline{v_t}(p_3)\phi_t+u_t(p_1)v_t(p_2)u_t(p_3)\overline{\phi}_t\big)\big(\hat{v}(p_2)+\hat{v}(p_3)\big)\\
	& + \big(u_t(p_1)v_t(p_2)\overline{v_t}(p_3)\phi_t+v_t(p_1)u_t(p_2)u_t(p_3)\overline{\phi}_t\big)\big(\hat{v}(p_1)+\hat{v}(p_3)\big) \Big) 
	\end{split}
    \end{align}
	in the expression for $\Hcub(t)$. Analogously, we also obtain the quartic collision kernels
    \begin{align}
        \MoveEqLeft \bbf{0}{4}_t(\bp_4) \, :=\label{def-bbf04}\\
			& \big(u_t(p_1)u_t(p_2)v_t(p_3)v_t(p_4)  +  v_t(p_1)v_t(p_2)u_t(p_3)u_t(p_4)\big)\big(\hat{v}(p_1+p_3)+\hat{v}(p_2+p_3)\big)\\
			& + \, \big(u_t(p_1)v_t(p_2)u_t(p_3)v_t(p_4)+ v_t(p_1)u_t(p_2)v_t(p_3)u_t(p_4)\big)\big(\hat{v}(p_1+p_2)+\hat{v}(p_2+p_3)\big)\\
			& + \, \big(u_t(p_1)v_t(p_2)v_t(p_3)u_t(p_4)+ v_t(p_1)u_t(p_2)u_t(p_3)v_t(p_4)\big)\big(\hat{v}(p_1+p_2)+\hat{v}(p_1+p_3)\big) \, ,\\
			\MoveEqLeft \bbf{1}{3}_t(\bp_4) \, := \\
			&\big(u_t(p_1)u_t(p_2)v_t(p_3)u_t(p_4)+ v_t(p_1)v_t(p_2)u_t(p_3)\overline{v_t}(p_4)\big)\big(\hat{v}(p_1+p_3)+\hat{v}(p_2+p_3)\big) \\
			& + \big(u_t(p_1)v_t(p_2)u_t(p_3)u_t(p_4)+ v_t(p_1)u_t(p_2)v_t(p_3)\overline{v_t}(p_4)\big)\big(\hat{v}(p_1+p_2)+\hat{v}(p_2+p_3)\big) \\
			& + \, \big(v_t(p_1)u_t(p_2)u_t(p_3)u_t(p_4)+ u_t(p_1)v_t(p_2)v_t(p_3)\overline{v_t}(p_4)\big)\big(\hat{v}(p_1+p_2)+\hat{v}(p_1+p_3)\big) \, ,  \\
			\MoveEqLeft \bbf{2}{2}_t(\bp_4) \, := \label{def-bbf22}\\
			& \big(u_t(p_1)u_t(p_2)u_t(p_3)u_t(p_4) \, + \, v_t(p_1)v_t(p_2)\overline{v_t}(p_3)\overline{v_t}(p_3)\big)\big(\hat{v}(p_1-p_3)+\hat{v}(p_2-p_3)\big)\\
			& + \, \big(u_t(p_1)v_t(p_2)\overline{v_t}(p_3)u_t(p_4)+v_t(p_1)u_t(p_2)u_t(p_3)\overline{v_t}(p_4)\big)\big(\hat{v}(p_1+p_2)+\hat{v}(p_2-p_3)\big)\\
			& + \, \big(v_t(p_1)u_t(p_2)\overline{v_t}(p_3)u_t(p_4)+u_t(p_1)v_t(p_2)u_t(p_3)\overline{v_t}(p_4)\big)\big(\hat{v}(p_1+p_2)+\hat{v}(p_1-p_3)\big) \, .
		\end{align}
    in the expression for $\Hquart(t)$. The Bogoliubov coefficients $u$ and $v$ are related to the HFB fields via
    \begin{equation}
			u_t(p) =  \sqrt{1+\gamma_t(p)} \, , \quad v_t(p) =  \frac{\sigma_t(p)}{\sqrt{1+\gamma_t(p)}} \, .
    \end{equation}
    Moreover, we abbreviate 
	\begin{align}
		\hb(p) \, &:= \, h(p)+1 \, , \label{def-hbar}\\
        \bpb \, &:= \, (p_1,p_2,-p_3) \, ,\label{def-bpb}\\
        \bprev_3 \, &:= \, (p_3,p_2,p_1) \, .\label{def-bprev}
	\end{align} 
    Then we introduce the cubic Boltzmann operators
    \begin{align}\label{def-f-Q3}
	\MoveEqLeft Q_3[h](t,p) \, := \\
	& 2\lambda^2 \Re \int_0^t \dx{s} \int \dx{\bp_3}  \Big(\frac1{2!}\big(\delta(p_1-p)+\delta(p_2-p)-\delta(p_3-p)\big) \\
	& \bbf{1}{2}_{s_1}(\bp_3)\bbfb{1}{2}_{s}(\bp_3)e^{i\int_{s}^{t} \dx{\tau} \big(\Omega_\tau(p_1)+\Omega_\tau(p_2)-\Omega_\tau(p_3)\big)}\delta(p_1+p_2-p_3)\\
	&\big(\hb_{s}(p_1)\hb_{s}(p_2)h_{s}(p_3)-h_{s}(p_1)h_{s}(p_2)\hb_{s}(p_3)\big) \\
	& + \, \frac1{3!}\big(\delta(p_1-p)+\delta(p_2-p)+\delta(p_3-p)\big)\\
	& \bbf{0}{3}_{s}(\bp_3)\bbfb{0}{3}_{s}(\bp_3)e^{i\int_{s}^{t} \dx{\tau} \big(\Omega_\tau(p_1)+\Omega_\tau(p_2)+\Omega_\tau(p_3)\big)}\delta(p_1+p_2+p_3) \\
	&\big(\hb_{s}(p_1)\hb_{s}(p_2)\hb_{s}(p_3)-h_{s}(p_1)h_{s}(p_2)h_{s}(p_3)\big)\Big) 
    \end{align}
    and
    \begin{align}\label{def-g-Q3}
            \MoveEqLeft Q_3^{(g)}[h](t)[J] \, :=\\
            &\lambda^2\int\dx{p}J(p)\int_0^t\dx{s} \int \dx{\bp_3}\Big[\delta(p_1+p_2-p_3)\Big(\delta(p-p_3)e^{2i\int_0^{t}\dx{\tau}\Omega_\tau(p_3)}\\
        & \quad e^{i\int_{s}^{s}\dx{\tau}(\Omega_\tau(p_1)+\Omega_\tau(p_2)-\Omega_\tau(p_3))} \bbf{0}{3}_{t}(\bpb_3) \bbfb{1}{2}_{s_2}(\bp_3) -2\delta(p-p_1)e^{-2i\int_0^{t}\dx{\tau}\Omega_\tau(p_1)}\\
        & \quad e^{-i\int_{s}^{t}\dx{\tau}(\Omega_\tau(p_1)+\Omega_\tau(p_2)-\Omega_\tau(p_3))}\bbfb{1}{2}_{t}(\bpb_3)\bbf{1}{2}_{s}(\bp_3)\Big)\\
        & \quad \big(h_{s}(p_1)h_{s}(p_2)\hb_{s}(p_3)-\hb_{s}(p_1)\hb_{s}(p_2)h_{s}(p_3)\big)\\
        & + \, \delta(p-p_3)e^{2i\int_0^{t}\dx{\tau}\Omega_\tau(p_3)}\delta(p_1+p_2+p_3)e^{-i\int_{s}^{t}\dx{\tau}(\Omega_\tau(p_1)+\Omega_\tau(p_2)+\Omega_\tau(p_3))} \\
        & \quad \bbfb{1}{2}_{t}(\bpb_3)\bbf{0}{3}_{s}(\bp_3)\big(\hb_{s}(p_1)\hb_{s}(p_2)\hb_{s}(p_3)-h_{s}(p_1)h_{s}(p_2)h_{s}(p_3)\big)\Big] 
    \end{align}
    and
    \begin{align}\label{def-Phi-Q3}
		    \MoveEqLeft Q_3^{(\Phi)}[h](t) \, := \\
            &\lambda^2\int_0^t\dx{s}e^{i\int_0^{s_1}\dx{\tau} \Omega_\tau(0)}\Big[\frac12 \delta(p_1+p_2-p_3)\Big(e^{i\int_{s}^{t}\dx{\tau}(\Omega_\tau(p_1)+\Omega_\tau(p_2)-\Omega_\tau(p_3))}\\
            & \quad \bbf{1}{3}_{t}(0,\bp_3)\bbfb{1}{2}_{s}(\bp_3) \, - \, e^{-i\int_{s}^{t}\dx{\tau}(\Omega_\tau(p_1)+\Omega_\tau(p_2)-\Omega_\tau(p_3))}\bbf{2}{2}_{t}(0,\bprev_3)\bbf{1}{2}_{s}(\bprev_3)\Big)\\
            & \qquad \big(h_{s}(p_1)h_{s}(p_2)\hb_{s}(p_3)-\hb_{s}(p_1)\hb_{s}(p_2)h_{s}(p_3)\big)\\
            & \quad + \, \frac1{3!}\delta(p_1+p_2+p_3)\Big(\bbf{0}{4}_{t}(0,\bp_3)\bbfb{0}{3}_{s}(\bp_3)e^{i\int_{s}^{t}\dx{\tau}(\Omega_\tau(p_1)+\Omega_\tau(p_2)+\Omega_\tau(p_3))}\\
            & \quad - \, \bbfb{1}{3}_{s_1}(\bp_3,0)\bbf{0}{3}_{s}(\bp_3) e^{-i\int_{s}^{t}\dx{\tau}(\Omega_\tau(p_1)+\Omega_\tau(p_2)+\Omega_\tau(p_3))}\Big)\\
            & \qquad \big(h_{s}(p_1)h_{s}(p_2)h_{s}(p_3)-\hb_{s}(p_1)\hb_{s}(p_2)\hb_{s}(p_3)\big)\Big] \, .
		\end{align}
        In the evolution of $\Phi$, we obtain the additional collision term
        \begin{align}\label{def-Phi-Q33}
            Q_{3,3}^{(\Phi)}[h](t) \, &:= \, \lambda e^{i\int_0^t\dx{\tau}\Omega_\tau(0)}\int\dx{p}\Big(\bbf{1}{2}_t(0,p,p)Q_3[h](t,p) \\
                & \quad\quad + \, \bbfb{1}{2}_t(p,p,0)Q_3^{(g)}[h](t,p) \, + \, \bbf{0}{3}_t(p,p,0) \overline{Q_3^{(g)}[h](t,p)}\Big) \, .
        \end{align}

        \subsubsection{Main theorem}

        \begin{theorem}\label{thm-main}
            Impose the same assumptions as in Proposition \ref{prop-HFB-gwp}, and let $\Omt$ be the Bogoliubov dispersion defined in \eqref{def-bog-dispersion}. In addition, let $\vol\geq 1$, $\lambda>0$, and $t>0$, and $N>0$, and assume that $\fd,\vd,\nd{\gamma_0}<\infty$. Then there exist constants $C>0$ dependent on $\fd,\nd{\gamma_0},\vd,\vol$, and $K>0$ dependent on $\fd,\nd{\gamma_0}$ s.t. we have for all $J\in L^2\cap L^\infty(\lattice)$ that
            \begin{align}
                \Big|\Phi_t \, - \, \frac1{N^{\frac32}}\int_0^t\dx{s}\big(Q_3^{(\Phi)}[f](s)+Q_{3,3}^{(\Phi)}[f](s)\big)\Big| \, &\leq Ce^{K\vd\vol\lambda t} \frac1{N^2} \, , \label{eq-phi-main}\\
                \Big|\int\dx{p}J(p)\big(f_t(p)-f_0(p)-\frac1N\int_0^t\dx{s}Q_3[f](s,p)\big)\Big| &\leq Ce^{K\vd\vol\lambda t} \frac{\|J\|_\infty}{N^{\frac32}} \, ,\label{eq-f-main}\\
                \Big|\int\dx{p}J(p)\big(g_t(p)-\frac1N\int_0^t\dx{s}Q_3^{(g)}[f](s,p)\big)\Big| &\leq C e^{K\vd\vol\lambda t} \frac{\|J\|_2+\|J\|_\infty}{N^{\frac32}}
                 \, . \label{eq-g-main}
            \end{align}
        \end{theorem}

        \begin{remark}
            The error bounds in Theorem \ref{thm-main} improve those obtained in \cite{chenhott} significantly. In the latter, the upper bounds were of order $O(\frac\lambda{N^{1/2}})$ for \eqref{eq-phi-main}, and $O(\frac\lambda{N})$ for \eqref{eq-f-main}, \eqref{eq-g-main}, respectively. The $t$-dependence of the error remained the same in each case. In Theorem \ref{thm-main}, the dependence of the error terms with respect to $N$ is sharp.
        \end{remark}

        For the following statement, we introduce the mesoscopic fields
        \begin{equation}\label{def-meso-fields}
            \Psi_T \, := \, \Phi_{T/\lambda^2}, \quad F_T \, := \, f_{T/\lambda^2}, \quad G_T \, := \, g_{T/\lambda^2} \, ,
        \end{equation}
        as well as the mesoscopic Boltzmann operators
        \begin{align}\label{eq-meso-boltz}
            \cQ^{(\Psi)}_k[F](S) \, &:= \, \lambda^{-2}Q^{(\Phi)}_k[f](S/\lambda^2), \quad \cQ_3[F](S) \, := \, \lambda^{-2}Q_3[f](S/\lambda^2), \\
            \cQ_3^{(G)}[F](S) \, &:= \, \lambda^{-2} Q_3^{(g)}[f](S/\lambda^2) \, .
        \end{align}
        
        \begin{corollary}\label{cor-validity}
            Under the same assumptions of Theorem \ref{thm-main}, for any $\delta\in(0,\frac12)$, there exists a constant $C_{\delta}>0$ dependent on $\fd$, $\nd{\gamma_0}$, $\vd$ and a constant $K_\delta$ dependent on $\fd$, $\nd{\gamma_0}$, $\vd$, $\vol$, such that for $t=\lambda^{-2}T$ and $\lambda=\frac{C_\delta}{\log N}$, we have that
            \begin{align}
                \Big|\Psi_{T} \, - \, \frac1{N^{\frac32}}\int_0^T\dx{S}\big(\cQ_3^{(\Psi)}[F](S)+\cQ_{3,3}^{(\Psi)}[F](S)\big)\Big| \, &\leq \frac{K_\delta}{N^{\frac32+\delta}} \, , \\
                \Big|\int\dx{p}J(p)\big(F_T(p)-F_0(p)-\frac1N\int_0^T\dx{S}\cQ_3[F](S,p)\big)\Big| &\leq \frac{K_\delta\|J\|_\infty}{N^{1+\delta}} \, ,\\
                \Big|\int\dx{p}J(p)\big(G_T(p)-\frac1N\int_0^T\dx{S}\cQ_3^{(G)}[F](S,p)\big)\Big| &\leq \frac{K_\delta(\|J\|_2+\|J\|_\infty)}{N^{1+\delta}}
                 \, . 
            \end{align}
        \end{corollary}
        
        \begin{remark}
             Corollary \ref{cor-validity} improves our previous time window, see \cite{chenhott}, which was of order $t\sim\big(\frac{\log N}{\log \log N}\big)^2$. Here, we obtain $t\sim (\log N)^2$.
        \end{remark}

        \begin{remark}
            In \cite{chenhott}, we also studied the case $L=\lambda^{-2-}\gg \lambda^{-2}\sim t$. This is due to the reason, that for longer times/shorter system sizes, we observe superposition of waves, while for shorter times, we observe dispersion. In that case, the time window of validity was given by $t\sim \lambda^{-2}$, with $\lambda=O\big(\big(\frac{\log \log N}{\log N}\big)^{\frac27-}\big)$. For these times, we observed an \emph{elastic} QBE, for which the dispersion relation is given by the Bogoliubov dispersion $\Omega$. We do not provide such a result in the present work, due to the complicated phase structure in the quantum Boltzmann operators, involving the phases of the HFB fields $(\phi,\sigma)$ as well as the \emph{modified} Bogoliubov dispersion $\Omt$, see Remark \ref{rem-mod-bog-disp}.
        \end{remark}
        
    \subsection{Sketch of the proof}

    In Section \ref{sec-main-terms}, we first derive the expansions for the total BEC wave function and total pair correlations. We then compute the perturbation expansion of $f$, $g$ and $\Phi$. In that expansion, we collect terms corresponding to the BEC evolution up to order $\frac1{\sqrt{N}}$. \mh{In particular, they correspond to the first-order Duhamel expansion of $\Phi$ and involve a single commutator with $\Hfluc$.} These terms can be eliminated\mh{, to leading order,} by a suitable choice of the HFB field $\phi$ \mh{fixing the Weyl transform $\weyl[\sqrt{N\vol}\phi]$.} The pair-absorption terms in the evolution of $g$ up to order $\frac1N$ contain some terms that can be characterized as HFB terms. They arise in the second order Duhamel expansion, either due to a single or double commutator with $\Hfluc$. These terms, in turn, can be eliminated by a suitable choice of $(\gamma,\sigma)$, \mh{determining the Bogoliubov rotation $\bog[k_t]$, see \eqref{eq-bog-rot-raw} and \eqref{def-sigma}, \eqref{def-up}.} The remaining 'free' evolution terms can be eliminated by properly choosing the dispersion $\Omega$ in $\bprop(t)$. The \mh{entire} procedure amounts to the renormalization of $(\phi,\gamma,\sigma,\Omega)$. In Section \ref{sec-HFB-est}, we derive a priori bounds for the renormalized HFB fields $(\phi,\gamma,\sigma)$. Those allow us to control the tail and other lower-order terms in the Duhamel expansions of $(\Phi,f,g)$.

    \section{Derivation of leading order terms\label{sec-main-terms}}

    \mh{In this section, we sketch the ideas needed to separate the HFB evolution from the QBE evolution. Crucially, we will use the transformation behavior of $a^{\#}$ under conjugation with $\weyl$, $\bog$, and $\bprop$. In particular, these transformations preserve quasifreeness.}

    \subsection{Density expansion\label{sec-dens-exp}}

    \mh{In this step, we subtract the HFB dynamics, by way of the unitary transformations associated with $\weyl[\sqrt{N\vol}\phi_t]$, $\bog[k_t]$ and $\bprop(t)$. Notice that for now we are not prescribing the specific dynamics. As we will see below, the dynamics of $(\phi,\gamma,\sigma,\omega)$ will emerge from a constraint condition formulated in the leading order dynamics.}
    
    \par \eqref{eq-weyl} implies
    \begin{align}\label{eq-dens-weyl-trafo}
        & \weyld[\sqrt{N\vol}\phi_t]\ad_pa_p\weyl[\sqrt{N\vol}\phi_t]  \\
        &\quad = \, \ad_pa_p \, + \, \delta(p)\Big(\sqrt{N\vol}(\phi_t\ad_0 +\overline{\phi}_t a_0) \, + \, N\vol^2|\phi_t|^2 \Big) \, .
    \end{align}
    Analogously, \eqref{eq-bog-rot}, followed by \eqref{def-gamma}, \eqref{def-sigma}, yields
    \begin{align}\label{eq-dens-bog-trafo}
        \begin{split}
        &\bogd[k_t]\ad_pa_p\bog[k_t]\\
        &\quad = \, (u_t(p)\ad_p+\overline{v}_t(p)a_{-p})(u_t(p)a_p+v_t(p)\ad_{-p})  \\
        &\quad = \, u_t(p)^2\ad_pa_p + |v_t(p)|^2\ad_{-p}a_{-p} +(u_t(p)v_t(p)\ad_p\ad_{-p}+\mathrm{h.c.})+\vol|v_t(p)|^2\\
        &\quad = \, (1+\gamma_t(p))\ad_pa_p+\gamma_t(p)\ad_{-p}a_{-p}+(\sigma_t(p)\ad_p\ad_{-p}+\mathrm{h.c.})+\vol\gamma_t(p) \, ,
        \end{split}
    \end{align}
    and, similarly,
    \begin{align}\label{eq-BEC-dens-bog-trafo}
        \begin{split}
        \bogd[k_t](\phi_t\ad_0 +\overline{\phi}_t a_0)\bog[k_t] \, &= \, \phi_t(u_t(0)\ad_0+\overline{v}_t(0)a_0)+\overline{\phi}_t(u_t(0)a_0+v_t(0)\ad_0)\\
        &= \, \big(\phi_tu_t(0)+\overline{\phi}_tv_t(0)\big)\ad_0+\mathrm{h.c.}
        \end{split}
    \end{align}
    Next, \eqref{eq-bog-prop} implies
    \begin{align}
        \bpropd(t)a_0\bprop(t) \, &= \, e^{-i\int_0^t \dx{s} \Omega_s(0)}a_0 \, , \label{eq-BEC-bprop-trafo}\\
        \bpropd(t)\ad_p\ad_{-p}\bprop(t) \, &= \, e^{2i\int_0^t \dx{s} \Omega_s(p)}\ad_p\ad_{-p} \, , \label{eq-pair-ab-bprop-trafo}
    \end{align}
    while $[\ad_pa_p,\bprop(t)]=0$.
    
    Collecting \eqref{eq-dens-weyl-trafo}, \eqref{eq-dens-bog-trafo}, \eqref{eq-BEC-dens-bog-trafo}, \eqref{eq-BEC-bprop-trafo} and \eqref{eq-pair-ab-bprop-trafo}, and using the fact that $f_t$ is even, we can rewrite the total density in terms of the relative densities
    \begin{align}
        \ftot_t(p) \, =& \, \delta(p)\Big[N\vol|\phi_t|^2 \, + \, \Big(\sqrt{N\vol}e^{i\int_0^t \dx{s} \Omega_s(0)}\big(\phi_tu_t(0)+\overline{\phi}_tv_t(0)\big)\overline{\Phi}_t \, + \, \mathrm{h.c.}\Big)\Big] \\
        & \, + \, \gamma_t(p) \, + \, \big(1+\gamma_t(p)\big)f_t(p) \, + \, \gamma_t(p)f_t(-p) \, + \, \Big(e^{2i\int_0^t \dx{s} \Omega_s(p)}\sigma_t(p) \overline{g}_t(p)+\mathrm{h.c.}\Big) \, .
    \end{align}
    Similarly, we have that
    \begin{align}\label{eq-pair-abs-weyl-trafo}
        \weyld[\sqrt{N\vol}\phi_t]a_pa_{-p} \weyl[\sqrt{N\vol}\phi_t] \, = \, a_pa_{-p} \, + \, \delta(p)\Big(2\sqrt{N\vol}\phi_ta_0 \, + \, N\vol^2\phi_t^2\Big) \, ,
    \end{align}
    and that
    \begin{align}\label{eq-pair-abs-bog-trafo}
        \begin{split}
        \bogd[k_t]a_pa_{-p}\bog[k_t] \, =& \, (u_t(p)a_p+v_t(p)\ad_{-p})(u_t(p)a_{-p}+v_t(p)\ad_{p})\\
        =& \, \vol\sigma_t(p)  \, + \, \sigma_t(\ad_pa_p+\ad_{-p}a_{-p}) \, + \, (1+\gamma_t(p))a_pa_{-p}\\
        & \, + \, \frac{\sigma_t(p)^2}{1+\gamma_t(p)}\ad_p\ad_{-p} \, .
        \end{split}
    \end{align}
    Following analogous steps as above, we obtain
    \begin{align}
        \gtot_t(p) \, =& \, \delta(p)\Big[N\vol\phi_t^2 \, + \, 2\sqrt{N\vol}\phi_t(u_t(0)e^{-i\int_0^t \dx{s} \Omega_s(0)}\Phi_t+v_t(0)e^{i\int_0^t \dx{s} \Omega_s(0)}\overline{\Phi}_t)\Big]\\
        & \, + \, \sigma_t(p) \, + \, \sigma_t(f_t(p)+f_t(-p)) \, + \, (1+\gamma_t(p))e^{2i\int_0^t \dx{s} \Omega_s(p)}g_t(p)\\
        & \, + \, \frac{\sigma_t(p)^2e^{-2i\int_0^t \dx{s} \Omega_s(p)}\overline{g}_t(p)}{1+\gamma_t(p)} \, ,
    \end{align}
    and also
    \begin{align}
        \Phitot_t \, = \, \sqrt{N\vol}\phi_t \, + \, u_t(0)e^{-i\int_0^t\dx{\tau}\Omega_{\tau}(0)}\Phi_t \, + \,  v_t(0)e^{i\int_0^t\dx{\tau}\Omega_{\tau}(0)}\overline{\Phi} \, . 
    \end{align}
    
	\subsection{Perturbation expansion\label{sec-perturbation}}

    \mh{Notice that thus far, we have not specified the fields $(\phi,\gamma,\sigma,\Omega)$ involved in the subtracted dynamics in the previous subsection. As we will see, the dynamics will be fixed by eliminating non-QBE terms in the leading order of the fluctuation dynamics. Consequently, the remaining terms in the fluctuation dynamics will describe the QBE to leading order.}
    
    \par We are interested in the evolution of 
	\begin{align}
	f_t(p) \, = \,  \frac{\jb{\flucd(t)\ad_pa_p\fluc(t)}_0}\vol \, . 
	\end{align} 
	Note that due to conjugation with $\fluc(t)$, the phase-factor $e^{iS_t}$ in $\Hfluc(t)$ drops out. Using the Duhamel expansion, we obtain that
	\begin{align}
	\mh{f_t(p)-f_0(p)\,  =\,} &   -i \int_0^t \dx{s} \frac{\jb{[\ad_pa_p,\Hfluc(s)]}_0}\vol \\
	& - \, \int_{[0,t]^2} \dx{\bs_2} \mathds{1}_{s_1\geq s_2} \, \frac{\jb{[[\ad_pa_p,\Hfluc(s_1)],\Hfluc(s_2)]}_0}\vol \\
	& + \,  \int_{[0,t]^3}\dx{\bs_3}\frac{\jb{[[[\ad_pa_p,\Hfluc(s_1)],\Hfluc(s_2)],\Hfluc(s_3)]}_{s_3}}\vol \, .\label{eq-f-tail}
	\end{align} 
	Translation invariance and gauge invariance imply that
	\begin{align}\label{eq-f-transport-vanish}
	\jb{[\ad_pa_p,\Hfluc(s)]}_0 \, = \, \wick{[\c1 {\ad_p} \c2 a_p, \settowidth{\wdth}{$\Hfluc$}\hspace{.25\wdth}\c1{\vphantom{\Hcub}}\hspace{-.25\wdth}\c2\Hfluc (s)]}  \, \propto \, f_0(p)\fbar(p)- \fbar(p) f_0(p) \, = \, 0 \, . 
	\end{align} 	
	Next, we have, due to gauge invariance, that
	\begin{align}
	\MoveEqLeft \frac{\jb{[[\ad_pa_p,\Hfluc(s_1)],\Hfluc(s_2)]}_0}\vol\\
	=&  \frac{\jb{[[\ad_pa_p,\HBEC(s_1)+\Hcub(s_1)],\HBEC(s_2)+\Hcub(s_2)]}_0}\vol \\
	& + \, \frac{\jb{[[\ad_pa_p,\HHFB(s_1)+\Hquart(s_1)],\HHFB(s_2)+\Hquart(s_2)]}_0}\vol  \, . \label{eq-f-expansion-0}
	\end{align}

    \mh{Our goal is now to eliminate all terms in \eqref{eq-f-expansion-0} \emph{not} corresponding to the quantum Boltzmann dynamics, which we will obtain from the contraction
    \begin{align}
        \begin{split}
		\frac{\jb{\wick{[[\c1 \ad_p \c2 a_p, \settowidth{\wdth}{$\Hcub$}\hspace{.25\wdth}\c3{\vphantom{\Hcub}}\hspace{-.25\wdth}\c1 \Hcub\settowidth{\wdth}{$\Hcub$}\hspace{-.25\wdth}\c1{\vphantom{\Hcub}}\hspace{.25\wdth}(s_1)], \settowidth{\wdth}{$\Hcub$}\hspace{.25\wdth}\c1{\vphantom{\Hcub}}\hspace{-.25\wdth} \c2 \Hcub\settowidth{\wdth}{$\Hcub$}\hspace{-.25\wdth} \c3{\vphantom{\Hcub}}\hspace{.25\wdth}( s_2)]}}_0}\vol + \, \frac{\jb{\wick{[[\c2 \ad_p \c1 a_p, \settowidth{\wdth}{$\Hcub$}\hspace{.25\wdth}\c3{\vphantom{\Hcub}}\hspace{-.25\wdth}\c1 \Hcub\settowidth{\wdth}{$\Hcub$}\hspace{-.25\wdth}\c1{\vphantom{\Hcub}}\hspace{.25\wdth}(s_1)], \settowidth{\wdth}{$\Hcub$}\hspace{.25\wdth}\c1{\vphantom{\Hcub}}\hspace{-.25\wdth} \c2 \Hcub\settowidth{\wdth}{$\Hcub$}\hspace{-.25\wdth} \c3{\vphantom{\Hcub}}\hspace{.25\wdth}( s_2)]}}_0}\vol \, .
	\end{split}
    \end{align}
    We will next describe how to gradually eliminate all other terms.
    }

    \subsubsection{First order HFB renormalization} \mh{The following analysis is analogous to the approach studied in \cite{grmama} and establishes the same equations.}
    \par For general fields $(\phi,\gamma,\sigma)$, the leading order contributions in \eqref{eq-f-expansion-0} are generated by $\HBEC$ and $\HHFB$. Thus, a possible choice is to set 
	\begin{align}
		\HBEC(t) \, =& \, 0 \, , \label{eq-phi-ren-0}\\
		\HHFB(t) \, =& \, 0 \, . \label{eq-gamma-sigma-ren-0}
	\end{align} 
	Observe that it is sufficient for calculating the leading order expressions to assume
	\begin{align}
		\HBEC(t),\HHFB(t)=O(\frac1{\sqrt{N}}) \, .
	\end{align}
	Lemma \ref{lem-hfluc} implies that \eqref{eq-phi-ren-0} is equivalent to
	\begin{align}
		\MoveEqLeft u_t(0)\Big( -i\partial_t \phi_t  +  \lambda\vol|\phi_t|^2\hat{v}(0) \phi_t +  \frac{\lambda}{N}\int \dx{p} \hat{v}(p)\sigma_t(p)\overline{\phi}_t\\
		& + \, \frac{\lambda}{N}\int \dx{p} \big(\hat{v}(p)+\hat{v}(0)\big) \gamma_t(p) \phi_t\Big) \, + \, v_t(0)\Big( -\overline{i\partial_t \phi_t} +  \lambda\vol|\phi_t|^2\hat{v}(0) \overline{\phi}_t \\
		& +  \frac{\lambda}{N}\int \dx{p} \hat{v}(p)\overline{\sigma}_t(p)\phi_t + \frac{\lambda}{N}\int \dx{p} \big(\hat{v}(p)+\hat{v}(0)\big) \gamma_t(p) \overline{\phi}_t\Big) \, = \, 0 \, .
	\end{align} 
	Abbreviating
	\begin{align}
	\Gamo_t(p) :=&  \gamma_t(p) \, + \, N\vol|\phi_t|^2 \delta(p) \, , \label{def-Gamo} \\
	\Sigo_t(p) :=&  \sigma_t(p) \, + \, N\vol \phi_t^2 \delta(p) \, , \label{def-Sigo}
	\end{align} 
	this condition is satisfied if 
	\begin{align}\label{eq-phi-HFB-0}
	    \begin{split}
	i\partial_t \phio_t \, =& \, \frac{\lambda}{N} \Big( \big(\Gamo_t*(\hat{v}+\hat{v}(0))\big) (0) \phio_t \, + \, (\Sigo_t*\hat{v}) (0) \phiob_t\Big) \\
	& - \, 2\lambda\vol\hat{v}(0) |\phio_t|^2 \phio_t \, . 
	\end{split}
    \end{align}
	The superscript '$^{(\stepo)}$' accounts for renormalization to first order, as detailed below. For now, they do not play a specific role. 
	\par Since $\HHFBd(t)$ is a diagonal quadratic operator, we can absorb it into phase factors in $\HBEC(t)$, $\Hcub(t)$, and $\Hquart(t)$. For that purpose, we set 
	\begin{align} 
	\HHFBd(t) \, = \, 0 \, , \label{eq-HFBd=0}
	\end{align} 
	or equivalently,
	\begin{align}\label{def-bog-dispersion}
    \begin{split}
	\MoveEqLeft\Omo_t(p) \, = \,\\
	&\Big(E(p)+\frac{\lambda}N\big((\gamo_t+N\vol|\phio_t|^2\delta)*(\hat{v}+\hat{v}(0))\big)(p)\Big)\big(1+2\gamo_t(p)\big) \\
	& + \, \frac{2\lambda}N\Re\Big(\big((\sigob_t+N\vol(\phiob_t)^2\delta)*\hat{v}\big)(p)\sigo_t(p)\Big) \, -  \frac{\Re\big(\sigob_t(p)i\partial_t\sigo_t(p)\big)}{1+\gamo_t(p)} \, . 
	\end{split}
    \end{align}
	see Lemma \ref{lem-bog-dispersion}. 
	\par Thus \eqref{eq-gamma-sigma-ren-0} is satisfied if
	\begin{align}
	\HHFBod(t) \, = \, 0 \, , \label{eq-HHFBod=0}
	\end{align} 
	which has been elaborated on in \cite{GM1,GM2,GM3}. Lemma \ref{lem-hfluc} implies that \eqref{eq-HHFBod=0} is satisfied if 
	\begin{align}\label{eq-HFB-tot-0}
    \begin{split}
	\MoveEqLeft\frac{i\partial_t \sigo_t(p)}2 \, - \, \frac{\sigo_t(p)i\partial_t \gamo_t(p)}{2(1+\gamo_t(p))} \, = \,\\
	&  \Big(E(p)+\frac{\lambda}{N}\Gamo_t*\big(\hat{v}+\hat{v}(0)\big)(p)\Big)\sigo_t(p)  \\
	& + \, \frac{\lambda}{2N} \Big(\big(\Sigo_t* \hat{v}\big)(p) (1+\gamo_t(p)) \, + \, \big(\Sigob_t*\hat{v}\big)(p) \frac{\sigo_t(p)^2}{1+\gamo_t(p)}\Big) \, . 
	\end{split}
    \end{align}
	We show in Lemma \ref{lem-HFB-equations} that \eqref{eq-HFB-tot-0} is equivalent to
	\begin{align}
	i\partial_t \gamo_t \, &=\, \frac{\lambda}N \big[\big(\Sigo_t*\hat{v} \big) \sigob_t \, - \,  \big(\Sigob_t*\hat{v} \big) \,\sigo_t\big] \, , \label{eq-gamma-ren-2}\\
	i\partial_t \sigo_t \, &= \, 2\big(E+\frac{\lambda}{N}\Gamo_t*(\hat{v}+\hat{v}(0))\big)\sigo_t \, + \, \frac{\lambda}{N} \big( \Sigo_t* \hat{v}\big) (1+2\gamo_t) \, .  \label{eq-sigma-ren-2}
	\end{align}
    Together with \eqref{eq-phi-HFB-0}, we thus have shown that
	\begin{align}\label{eq-HFB-ren-1}
    \begin{split}
	i\partial_t \phio_t  \, =& \, \frac{\lambda}{N} \Big( \big(\Gamo_t*(\hat{v}+\hat{v}(0))\big) (0) \phio_t \, + \, (\Sigo_t*\hat{v}) (0) \phiob_t\Big) \\
	& - \, 2\lambda\vol\hat{v}(0) |\phio_t|^2 \phio_t \, . \\
	\partial_t \gamo_t  \, =& \, \frac{2\lambda}N \Im\big((\Sigo_t*\hat{v} ) \sigob_t\big) \, , \\
	i\partial_t \sigo_t \, =& \, 2\big(E+\frac{\lambda}{N}\Gamo_t*(\hat{v}+\hat{v}(0))\big)\sigo_t \, + \, \frac{\lambda}{N} \big( \Sigo_t* \hat{v}\big) (1+2\gamo_t) \, . 
	\end{split}
    \end{align}
	These are the well-known HFB equations in the translation invariant case, see, e.g., \cite{BBCFS22,GM1}. Moreover, Lemma \ref{lem-bog-dispersion} implies that, if $\sigo$ satisfies \eqref{eq-HFB-ren-1}, then the Bogoliubov dispersion, see \eqref{def-bog-dispersion}, satisfies
	\begin{align}\label{eq-Omo-explicit}
	\Omo_t  \, =\, E+\frac{\lambda}{N}\Gamo_t*\big(\hat{v}+\hat{v}(0)\big) \, + \, \frac{\lambda}{N}\frac{\Re\big((\Sigo_t*\hat{v}) \sigob_t \big)}{1+\gamo_t} \, .
	\end{align}

	\subsubsection{Second order HFB renormalization}
    \mh{We now perform a recursion that extends the previous construction in a natural sense. Namely, eliminating terms only in the first-order Duhamel expansion does not ensure that corrections to the quantum Boltzmann evolution in the second-order Duhamel expansion vanish. Accordingly, we impose corrections to the equations for the governing the HFB fields, to eliminate those corrections.}
    \par In order to determine the cubic Boltzmann operator, we follow \cite{chenhott} and compute the second order Duhamel expansion.
 
 Observe that for a self-adjoint operator $A$, operators $B$, $C$ and any state $\nu$ we have that
	\begin{align}
	\nu([[ A,B+B^\dagger], C+C^\dagger]) \, = \, 2\Re\big(\nu([[ A,B], C]) + \nu([[ A,B], C^\dagger])\big) \, . \label{eq-ada-dc-0}
	\end{align} 
	\begin{remark}[Commutator rule]\label{rem-com-rule}
		We note that, due to the commutators, every right argument in a commutator needs to be \mh{contracted with} at least one argument to the left of it. We refer to this fact as the \emph{commutator rule}.
	\end{remark}

    \begin{remark}\label{rem-int-abs-sq}
		Observe that we have that
		\begin{align}
		2\Re \int_{[0,t]^2} \dx{\bs_2} \mathds{1}_{s_1\geq s_2} A(s_1)\overline{A}(s_2) \, = \, \Big| \int_0^t \dx{s} A(s)\Big|^2 \, .
		\end{align} 
	\end{remark}

 \par Lemma \ref{lem-hfluc} implies that
    \begin{align}\label{eq-hcub-con}
            \begin{split}
                \wick{\settowidth{\wdth}{$\Hcub$}\hspace{.33\wdth}\c1{\vphantom{\Hcub}}\hspace{.33\wdth}\c1{\vphantom{\Hcub}}\hspace{-.66\wdth}\Hcub(t)} \, &= \, \frac{\lambda\sqrt{\vol}}{\sqrt{N}} e^{i\int_0^t\dx{\tau}\Omega_\tau(0)}\ad_0 \int \dx{k} \Big[u_t(0)\Big(\big(1+2\gamma_t(k)\big)f_0(k)\big(\hat{v}(k)+\hat{v}(0)\big)\phi_t\\
                & \qquad + 2f_0(k)\sigma_t(k)\hat{v}(k)\overline{\phi}_t\Big) \, + \, v_t(0)\Big(\big(1+2\gamma_t(k)\big)f_0(k)\big(\hat{v}(k)+\hat{v}(0)\big)\overline{\phi}_t \\
                & \qquad + 2f_0(k)\overline{\sigma}_t(k)\hat{v}(k)\phi_t\Big)\Big] \, + \, \mathrm{h.c.} \, .
            \end{split}
        \end{align}
	With that, we obtain 
	\begin{align}
	\MoveEqLeft -\int_{[0,t]^2} \dx{\bs_2} \mathds{1}_{s_1\geq s_2}\frac{\wick{[[\c1 {\ad_p} \c2 a_p, \settowidth{\wdth}{$\Hcub$}\hspace{.25\wdth}\c1{\vphantom{\Hcub}}\hspace{.25\wdth}\c1{\vphantom{\Hcub}}\hspace{.25\wdth}\c1{\vphantom{\Hcub}}\hspace{-.75\wdth}\Hcub(s_1) ], \hspace{.25\wdth}\c2{\vphantom{\Hcub}}\hspace{.25\wdth}\c1{\vphantom{\Hcub}}\hspace{.25\wdth}\c1{\vphantom{\Hcub}}\hspace{-.75\wdth}\Hcub(s_2)]} \, + \, \wick{[[\c2 {\ad_p} \c1 a_p, \settowidth{\wdth}{$\Hcub$}\hspace{.25\wdth}\c1{\vphantom{\Hcub}}\hspace{.25\wdth}\c1{\vphantom{\Hcub}}\hspace{.25\wdth}\c1{\vphantom{\Hcub}}\hspace{-.75\wdth}\Hcub(s_1) ], \hspace{.25\wdth}\c2{\vphantom{\Hcub}}\hspace{.25\wdth}\c1{\vphantom{\Hcub}}\hspace{.25\wdth}\c1{\vphantom{\Hcub}}\hspace{-.75\wdth}\Hcub(s_2)]}}\vol\\
	= & - \, 2\Re\Big( \int_{[0,t]^2} \dx{\bs_2} \mathds{1}_{s_1\geq s_2} \frac{\wick{[[\c1 {\ad_p} \c2 a_p, \settowidth{\wdth}{$\Hcub$}\hspace{.25\wdth}\c1{\vphantom{\Hcub}}\hspace{.25\wdth}\c1{\vphantom{\Hcub}}\hspace{.25\wdth}\c1{\vphantom{\Hcub}}\hspace{-.75\wdth}\Hcub(s_1) ], \hspace{.25\wdth}\c2{\vphantom{\Hcub}}\hspace{.25\wdth}\c1{\vphantom{\Hcub}}\hspace{.25\wdth}\c1{\vphantom{\Hcub}}\hspace{-.75\wdth}\Hcub(s_2)]}}\vol\Big) \\
	=&  - \, \delta(p)\int_{[0,t]^2} \dx{\bs_2} \mathds{1}_{s_1\geq s_2}\Big(\frac{\jb{[ \ad_0 , \Hcub(s_1)]}_0}\vol\frac{\jb{[ a_0 , \Hcub(s_2)]}_0}\vol \, + \, \big(s_1\leftrightarrow s_2\big)\Big) \\
	=&  \delta(p)\Big|-i\int_0^t \dx{s} \frac{\jb{[a_0,\Hcub(s)]}_0}\vol \Big|^2 \, ,\label{eq-f-hcub-con-0}
	\end{align} 
    see Remark \ref{rem-int-abs-sq}. As we show in \cite{chenhott}, we have that this condensate contribution of size $\frac{\lambda^2t^2}{N}$, and it dominates the cubic Boltzmann collision operator coming from
	\begin{align}
	\frac1N\int_0^t\dx{s}Q_3[f_0](s,p) = & - \, \int_{[0,t]^2} \dx{\bs_2}  \mathds{1}_{s_1\geq s_2} \Big(\frac{\jb{\wick{[[\c1 \ad_p \c2 a_p, \settowidth{\wdth}{$\Hcub$}\hspace{.25\wdth}\c3{\vphantom{\Hcub}}\hspace{-.25\wdth}\c1 \Hcub\settowidth{\wdth}{$\Hcub$}\hspace{-.25\wdth}\c1{\vphantom{\Hcub}}\hspace{.25\wdth}(s_1)], \settowidth{\wdth}{$\Hcub$}\hspace{.25\wdth}\c1{\vphantom{\Hcub}}\hspace{-.25\wdth} \c2 \Hcub\settowidth{\wdth}{$\Hcub$}\hspace{-.25\wdth} \c3{\vphantom{\Hcub}}\hspace{.25\wdth}( s_2)]}}_0}\vol \\
	& + \, \frac{\jb{\wick{[[\c2 \ad_p \c1 a_p, \settowidth{\wdth}{$\Hcub$}\hspace{.25\wdth}\c3{\vphantom{\Hcub}}\hspace{-.25\wdth}\c1 \Hcub\settowidth{\wdth}{$\Hcub$}\hspace{-.25\wdth}\c1{\vphantom{\Hcub}}\hspace{.25\wdth}(s_1)], \settowidth{\wdth}{$\Hcub$}\hspace{.25\wdth}\c1{\vphantom{\Hcub}}\hspace{-.25\wdth} \c2 \Hcub\settowidth{\wdth}{$\Hcub$}\hspace{-.25\wdth} \c3{\vphantom{\Hcub}}\hspace{.25\wdth}( s_2)]}}_0}\vol\Big) \, , \label{eq-cub-boltz-first-intro-0}
	\end{align} 
	at least in the continuous approximation, when it is of size $\frac{\lambda^2t}{N}$. Since \eqref{eq-f-hcub-con-0} is proportional to $\delta(p)$, we can absorb it into the condensate contribution. With similar steps as in \eqref{eq-f-hcub-con-0} and employing \eqref{eq-f-transport-vanish}, we obtain
	\begin{align}\label{eq-f-condensate-0}
        \begin{split}
	\MoveEqLeft- \, \int_{[0,t]^2} \dx{\bs_2} \mathds{1}_{s_1\geq s_2}\Big(\frac{\jb{[[\ad_pa_p,\HBEC(s_1)],\HBEC(s_2)]}_0}\vol\\
	&+ \, \frac{\jb{[[\ad_pa_p,\HBEC(s_1)],\Hcub(s_2)]}_0}\vol\\
	&+ \, \frac{\jb{[[\ad_pa_p,\Hcub(s_1)],\HBEC(s_2)]}_0}\vol\\
	& + \, 2\Re \frac{\wick{[[\c1 {\ad_p} \c2 a_p, \settowidth{\wdth}{$\Hcub$}\hspace{.25\wdth}\c1{\vphantom{\Hcub}}\hspace{.25\wdth}\c1{\vphantom{\Hcub}}\hspace{.25\wdth}\c1{\vphantom{\Hcub}}\hspace{-.75\wdth}\Hcub(s_1) ], \hspace{.25\wdth}\c2{\vphantom{\Hcub}}\hspace{.25\wdth}\c1{\vphantom{\Hcub}}\hspace{.25\wdth}\c1{\vphantom{\Hcub}}\hspace{-.75\wdth}\Hcub(s_2)]}}\vol\Big)\\
	&= \, \delta(p)\Big|-i\int_0^t \dx{s} \frac{\jb{[a_0,\HBEC(s) \, + \, \Hcub(s)]}_0}\vol \Big|^2 \, . 
	\end{split}
    \end{align}
	In order to eliminate this contribution, we choose 
	\begin{align}
		\jb{[a_0,\HBEC(t)+\Hcub(t)]}_0 \, = \, 0 \, . \label{eq-condensate-cond-0}
	\end{align} 
    \mh{This condition is equivalent to eliminating all first-order terms in the second-order Duhamel expansion of $\Phi$.}
	Lemma \ref{lem-hfluc} implies that
	\begin{align}\label{eq-phi-HBEC-0}
    \begin{split}
	\MoveEqLeft \frac{\jb{[a_0,\HBEC(t)]}_0}\vol \, = \\
	&\sqrt{N\vol} \Big[u_t(0)\Big( -i\partial_t \phi_t  +  \lambda\vol|\phi_t|^2\hat{v}(0) \phi_t +  \frac{\lambda}{N}\int \dx{p}\hat{v}(p)\sigma_t(p)\overline{\phi}_t\\
	& + \, \frac{\lambda}{N}\int \dx{p} \big(\hat{v}(p)+\hat{v}(0)\big) \gamma_t(p) \phi_t\Big) \, + \, v_t(0)\Big( -\overline{i\partial_t \phi_t} +  \lambda\vol|\phi_t|^2\hat{v}(0) \overline{\phi}_t \\
	& +  \frac{\lambda}{N}\int \dx{p} \hat{v}(p)\overline{\sigma}_t(p)\phi_t + \frac{\lambda}{N}\int \dx{p} \big(\hat{v}(p)+\hat{v}(0)\big) \gamma_t(p) \overline{\phi}_t\Big)\Big] e^{i\int_0^t \dx{s} \Omega_{s}(0)} \, , 
	\end{split}
    \end{align}
	and \eqref{eq-hcub-con} yields
	\begin{align}\label{eq-phi-cub-0}
	    \begin{split}
	\MoveEqLeft \frac{\jb{[a_0,\Hcub(t)]}_0}\vol \, = \, \frac{\jb{[a_0,\wick{\settowidth{\wdth}{$\Hcub$}\hspace{.33\wdth}\c1{\vphantom{\Hcub}}\hspace{.33\wdth}\c1{\vphantom{\Hcub}}\hspace{-.66\wdth}\Hcub(t)}]}_0}\vol \, =\\
	& \frac{\lambda\sqrt{\vol}}{\sqrt{N}} e^{i\int_0^t \dx{s} \Omega_s(0)}\Big[u_t(0)\int \dx{p} \Big(\big(1+2\gamma_t(p)\big)\phi_t\big(\hat{v}(p)+\hat{v}(0)\big)+ \\
	& 2 \sigma_t(p)\overline{\phi}_t \hat{v}(p)\Big)f_0(p) + v_t(0) \int \dx{p} \Big(  \big(1+2\gamma_t(p)\big)\overline{\phi}_t\big(\hat{v}(p)+\hat{v}(0)\big)+ \\
	& 2 \overline{\sigma}_t(p)\phi_t \hat{v}(p)\Big)f_0(p) \, . 
	\end{split}
    \end{align}
	In order to satisfy \eqref{eq-condensate-cond-0}, it suffices to equate the sum of the coefficients of $u_t(0)$ in \eqref{eq-phi-HBEC-0} and \eqref{eq-phi-cub-0} to zero. This condition is equivalent to
	\begin{align} \label{eq-phi-ren-1}
    \begin{split}
	i\partial_t \phi_t  \, =& \, \lambda\vol|\phi_t|^2\hat{v}(0) \phi_t +  \frac{\lambda}{N}\int \dx{p} \hat{v}(p)\sigma_t(p)\overline{\phi}_t \, + \, \frac{\lambda}{N}\int \dx{p} \big(\hat{v}(p)+\hat{v}(0)\big) \gamma_t(p) \phi_t\\
	& + \, \frac{\lambda}{N}\int \dx{p} \Big(\big(1+2\gamma_t(p)\big)\phi_t\big(\hat{v}(p)+\hat{v}(0)\big)+2 \sigma_t(p)\overline{\phi}_t \hat{v}(p)\Big)f_0(p) \, . 
	\end{split}
    \end{align}
	Observe that all integrands except for $f_0$ are even \mh{in $p$}. In anticipation of the evolution of $\gamma$ and $\sigma$, recall from \eqref{def-fplus}
	\begin{align}
	\fplus(p) \, = \, \frac12\big(f_0(p)+f_0(-p)\big) \, ,
	\end{align} 
	the even symmetrization of $f_0$.  In analogy to \eqref{def-Gamo}, \eqref{def-Sigo}, we introduce the second order renormalized shifted expectations
	\begin{align}
	\Gamt :=&  (1+2\fplus)\gamma \, + \, \fplus \, + \, N\vol|\phi|^2\delta \, , \label{def-Gamt} \\
	\Sigt :=&  (1+2\fplus)\sigma \, + \, N\vol\phi^2\delta \, . \label{def-Sigt}
	\end{align} 
	With these, and writing $\phi=\phit$, we can simplify \eqref{eq-phi-ren-1} into
	\begin{align}
    \begin{split}
	i\partial_t \phit_t \, =& \,  \frac{\lambda}{N} \Big( \big(\Gamt_t*(\hat{v}+\hat{v}(0))\big) (0) \phit_t \, + \, (\Sigt_t*\hat{v}) (0) \phitb_t\Big) \\
	& - \, 2\lambda\vol\hat{v}(0) |\phit_t|^2 \phit_t \, . \label{eq-phi-HFB-1}
	\end{split}
    \end{align}
	We recognize that \eqref{eq-phi-HFB-1} is a renormalization of \eqref{eq-phi-HFB-0}, where we substituted $(\Gamo,\Sigo)$ by the renormalized fields $(\Gamt,\Sigt)$. 
	\par Next, observe that
	\begin{align}
		[\ad_pa_p,\HHFBod(t)] \, = \, A_t(p)a_pa_{-p} \, + \mathrm{h.c.} \label{eq-f-HFB-com-0}
	\end{align} 
	for some coefficient $A_t$.	In particular, \eqref{eq-ada-dc-0} yields
	\begin{align}
		\MoveEqLeft \jb{[[\ad_pa_p,\HHFBod(s_1)],\HHFBod(s_2)+\Hquart(s_2)]}_0\\
		=&  2\Re\big(A_{s_1}(p)\jb{[a_pa_{-p},\HHFBod(s_2)+\Hquart(s_2)]}_0\big) \, . \label{eq-aa-cond-part-1}
	\end{align} 
	In addition, we have that 
	\begin{align}
		\wick{[\c1 {\ad_p} \c2 a_p, \c1 \Hquart\settowidth{\wdth}{$\Hquart$}\hspace{-.25\wdth}\c2{\vphantom{\Hquart}}\hspace{.25\wdth}( t)]} \, = \, \wick{[\c1 a_p \c2 {\ad_p}, \c1 \Hquart\settowidth{\wdth}{$\Hquart$}\hspace{-.25\wdth}\c2{\vphantom{\Hquart}}\hspace{.25\wdth}(t)]}\, \propto \,  f_0(p)\fbar(p)-\fbar(p)f_0(p) \, = \, 0 \, . \label{eq-aa-cond-part-2}
	\end{align} 
    Lemma \ref{lem-con-vert} yields
    \begin{align}\label{eq-hquart-con}
            \begin{split}
             \wick{\settowidth{\wdth}{$\Hquart$}\hspace{.33\wdth}\c1{\vphantom{\Hquart}}\hspace{.33\wdth}\c1{\vphantom{\Hquart}}\hspace{-.66\wdth}\Hquart(t)} \, =& \, \frac{\lambda}{N} \int \dx{p}\Big[\Big(\big(\fplus \sigma_t\big)*\hat{v}(p)(1+\gamma_t(p)) \, + \, \frac{\big(\fplus \overline{\sigma}_t\big)*\hat{v}(p)\sigma_t(p)}{1+\gamma_t(p)}\\
             & \quad + \, \big((1+2\gamma_t)\fplus\big)*\big(\hat{v}+\hat{v}(0)\big)(p)\sigma_t(p)\Big)e^{2i\int_0^t \dx{\tau}\Omega_\tau(p)}\ad_{p}\ad_{-p}\,+ \, \mathrm{h.c.}\\
             & \quad + \, \Big(\big((1+2\gamma_t)\fplus)*\big(\hat{v}+\hat{v}(0)\big)(p)(1+2\gamma_t(p)) \\
             & \quad + \, 4\Re\big(\big(\fplus\sigma_t\big)*\hat{v}(p)\overline{\sigma}_t(p)\big)\Big) \ad_pa_p\Big] \, .
            \end{split}
    \end{align}
	Then the CCRs \eqref{eq-CCR} imply
	\begin{align}\label{eq-aa-cond-part-3}
        \begin{split}
		[\ad_p a_p , \wick{\settowidth{\wdth}{$\Hquart$}\hspace{.33\wdth}\c1{\vphantom{\Hquart}}\hspace{.33\wdth}\c1{\vphantom{\Hquart}}\hspace{-.66\wdth}\Hquart(t)}] \, & = \, B(t,p) a_pa_{-p} \, + \, \mathrm{h.c.} 
	\end{split}
    \end{align}
	for some coefficient $B$. Using \eqref{eq-ada-dc-0}, \eqref{eq-f-HFB-com-0}, \eqref{eq-aa-cond-part-2}, and \eqref{eq-aa-cond-part-3}, we thus have
	\begin{align}\label{eq-aa-cond-part-4}
        \begin{split}
		\MoveEqLeft \jb{[[\wick{\ad_p a_p , \settowidth{\wdth}{$\Hquart$}\hspace{.25\wdth}\c1{\vphantom{\Hquart}}\hspace{-.25\wdth} \Hquart\settowidth{\wdth}{$\Hquart$}\hspace{-.25\wdth}\c1{\vphantom{\Hquart}}\hspace{.25\wdth}(s_1)}],\HHFBod(s_2)+\Hquart(s_2)]}_0\\
		& =  2\Re\big(B(s_1,p)\jb{[a_pa_{-p},\HHFBod(s_2)+\Hquart(s_2)]}_0\big) \, . 
	\end{split}
    \end{align}
	In order to eliminate the contributions coming from \eqref{eq-aa-cond-part-1} and \eqref{eq-aa-cond-part-4}, we choose
	\begin{align}
		\jb{[a_pa_{-p},\HHFBod(s_2)+\Hquart(s_2)]}_0 \, = \, 0. \label{eq-aa-cond-0}
	\end{align} 
	This condition relates to the {\it pair-absorption rate}. Observe that we have
	\begin{align}
	\wick{[\c1 a_p \c2 a_{-p},\c1 \ad \c2 \ad ]} \, \propto \, 1 \, + \, 2\fplus(p) \, ,
	\end{align} 
	where $\fplus(p)=f_0(p)+f_0(-p)$. 
	\par We start by calculating
	\begin{align}
	\MoveEqLeft \frac{\jb{[a_pa_{-p},\HHFBod(t)]}_0}\vol \, = \,  2 \big(1+2\fplus(p)\big)e^{2i\int_0^t \dx{s} \Omega_s(p)}\\
	&  \Big[-\frac{i\partial_t \sigma_t(p)}2 \, + \, \frac{\sigma_t(p)i\partial_t \gamma_t(p)}{2(1+\gamma_t(p))} +  \Big(E(p)+\frac{\lambda}{N}\Gamo_t*\big(\hat{v}+\hat{v}(0)\big)(p)\Big)\sigma_t(p)  \\
	& + \, \frac{\lambda}{2N} \Big( \big(\Sigo_t* \hat{v}\big)(p) (1+\gamma_t(p)) \, + \, \big(\Sigob_t*\hat{v}\big)(p) \frac{\sigma_t(p)^2}{1+\gamma_t(p)}\Big)  \Big]  \, , \label{eq-g-hcor-0}
	\end{align} 
	see Lemma \ref{lem-hfluc} for the expression for $\HHFBod$. We used the fact that all functions appearing here, except for $f_0$, are even, and we replaced $\gamma$, $\sigma$ by their respective shifts $\Gamo$, $\Sigo$, see \eqref{def-Gamo}, \eqref{def-Sigo}. Similarly, \eqref{eq-hquart-con} yields
	\begin{align}\label{eq-g-hquart-0}
    \begin{split}
	\MoveEqLeft \frac{\jb{[a_pa_{-p},\Hquart(t)]}_0}\vol \, = \, \frac{\jb{[a_pa_{-p},\wick{\settowidth{\wdth}{$\Hquart$}\hspace{.33\wdth}\c1{\vphantom{\Hquart}}\hspace{.33\wdth}\c1{\vphantom{\Hquart}}\hspace{-.66\wdth}\Hquart(t)}]}_0}\vol \, =\\
	&  2 \frac{\lambda}{N} \big(1+2\fplus(p)\big)e^{2i\int_0^t \dx{s} \Omega_s(p)} \Big(\big((\sigma\fplus)*\hat{v}\big)(p)\big(1+\gamma_t(p)\big)\\
	&  + \, \big((\overline{\sigma}_t\fplus)*\hat{v}\big)(p)\frac{\sigma_t(p)^2}{1+\gamma_t(p)} \, + \, \big((1+2\gamma_t)\fplus\big)*\big(\hat{v}+\hat{v}(0)\big)(p)\sigma_t(p) \Big)  \, . 
    \end{split}
	\end{align} 	
	Hence, substituting \eqref{eq-g-hcor-0} and \eqref{eq-g-hquart-0} into \eqref{eq-aa-cond-0} yields
	\begin{align}\label{eq-HFB-cor-raw-1}
    \begin{split}
	\MoveEqLeft \frac{i\partial_t \sigma_t(p)}2-\frac{\sigma_t(p)i\partial_t \gamma_t(p)}{2(1+\gamma_t(p))} \, = \\
	& \Big(E(p)+\frac{\lambda}{N}\Gamo_t*\big(\hat{v}+\hat{v}(0)\big)(p)\Big)\sigma_t(p) \\
	& + \, \frac{\lambda}{2N} \Big( \big(\Sigo_t* \hat{v}\big)(p) (1+\gamma_t(p)) \, + \, \big(\Sigob_t*\hat{v}\big)(p) \frac{\sigma_t(p)^2}{1+\gamma_t(p)}\Big)\\
	& \, + \, \frac{\lambda}{N} \Big(\big((1+2\gamma_t)\fplus\big)*\big(\hat{v}+\hat{v}(0)\big)(p)\sigma_t(p) \, + \, \big((\sigma\fplus)*\hat{v}\big)(p)\big(1+\gamma_t(p)\big) \\
	&  + \, \big((\overline{\sigma}_t\fplus)*\hat{v}\big)(p)\frac{\sigma_t(p)^2}{1+\gamma_t(p)}\Big) \, . 
    \end{split}
	\end{align} 	 
	Using the renormalized shifted fields $\Gamt$ and $\Sigt$, see \eqref{def-Gamt} and \eqref{def-Sigt}, we can rewrite \eqref{eq-HFB-cor-raw-1} as
	\begin{align}
	\MoveEqLeft \frac{i\partial_t \sigma_t(p)}2-\frac{\sigma_t(p)i\partial_t \gamma_t(p)}{2(1+\gamma_t(p))} \, = \\
	& \Big(E(p)+\frac{\lambda}{N}\Gamt_t*\big(\hat{v}+\hat{v}(0)\big)(p)\Big)\sigma_t(p) \\
	& + \, \frac{\lambda}{2N} \Big( \big(\Sigt_t* \hat{v}\big)(p) (1+\gamma_t(p)) \, + \, \big(\Sigtb_t*\hat{v}\big)(p) \frac{\sigma_t(p)^2}{1+\gamma_t(p)}\Big) \, . \label{eq-HFB-tot-1}
	\end{align} 
	We recognize that \eqref{eq-HFB-tot-1} is a renormalization of \eqref{eq-HFB-tot-0}, where $(\Gamo,\Sigo)$ is replaced by the renormalized fields $(\Gamt,\Sigt)$. In particular, the proof of Lemma \ref{lem-HFB-equations} implies that \eqref{eq-HFB-tot-1} is equivalent to
	\begin{align}
    \begin{split}
	i\partial_t \gamma_t \, = &\, \frac{\lambda}N \big[\big(\Sigt_t*\hat{v} \big) \,\overline{\sigma}_t \, - \,  \big(\Sigtb_t*\hat{v} \big) \,\sigma_t\big] \, , \\
	i\partial_t \sigma_t \, =& \, 2\big(E+\frac{\lambda}{N}\Gamt_t*(\hat{v}+\hat{v}(0))\big)\sigma_t \, + \, \frac{\lambda}{N} \big( \Sigt_t* \hat{v}\big) (1+2\gamma_t) \, . 
	\end{split}
    \end{align}
	Recalling \eqref{def-Gamt}, \eqref{def-Sigt} and \eqref{eq-phi-HFB-1}, we thus have shown that
	\begin{align}\label{eq-HFB-ren-2}
    \begin{split}
	i\partial_t \phit_t  \, = & \,  \frac{\lambda}{N} \Big( \big(\Gamt_t*(\hat{v}+\hat{v}(0))\big) (0) \phit_t \, + \, (\Sigt_t*\hat{v}) (0) \phitb_t\Big) \\
	& - \, 2\lambda\vol\hat{v}(0) |\phit_t|^2 \phit_t \, .  \\
	\partial_t \gamt_t \, =& \,\frac{2\lambda}N \Im\big((\Sigt_t*\hat{v} ) \sigtb_t\big) \, , \\
	i\partial_t \sigt_t \, =& \, 2\big(E+\frac{\lambda}{N}\Gamt_t*(\hat{v}+\hat{v}(0))\big)\sigt_t \, + \, \frac{\lambda}{N} \big( \Sigt_t* \hat{v}\big) (1+2\gamt_t) \, .
	\end{split}
    \end{align}
	\eqref{eq-HFB-ren-2} corresponds to the second order renormalization of \eqref{eq-HFB-ren-1} with the corresponding renormalized fields.

    \par To complete the renormalization of the HFB fields, we need to also renormalize the Bogoliubov dispersion $\Omega$. For that purpose, we recall the diagonal part of $\Hquartcon(t)$ from \eqref{eq-hquart-con}, see also Lemma \ref{lem-hfluc},
    \begin{align}
        \begin{split}
        \Hquartcond{t} \, :=& \, \frac{\lambda}{N} \int \dx{p}\Big(\big((1+2\gamma_t)\fplus)*\big(\hat{v}+\hat{v}(0)\big)(p)(1+2\gamma_t(p)) \\
        & \quad + \, 4\Re\big(\big(\fplus\sigma_t\big)*\hat{v}(p)\overline{\sigma}_t(p)\big)\Big) \ad_pa_p \, .
        \end{split}
    \end{align}
    Then we choose $\Omt$ such that
    \begin{align}
        \HHFBd(t) \, + \, \Hquartcond{t} \, = \, 0 \, . \label{eq-omt-ren-cond}
    \end{align}
    Employing Lemma \ref{lem-hfluc}, see also \eqref{def-bog-dispersion}, we obtain
    \begin{align}
        \begin{split}
        \MoveEqLeft\Omt_t(p) \, = \,\\
	       &\Big(E(p)+\frac{\lambda}N\big((\gamt_t+N\vol|\phit_t|^2\delta)*(\hat{v}+\hat{v}(0))\big)(p)\Big)\big(1+2\gamt_t(p)\big) \\
	       & + \, \frac{2\lambda}N\Re\Big(\big((\sigtb_t+N\vol(\phitb_t)^2\delta)*\hat{v}\big)(p)\sigt_t(p)\Big) \, -  \frac{\Re\big(\sigtb_t(p)i\partial_t\sigt_t(p)\big)}{1+\gamt_t(p)}\\
          & + \, \frac{\lambda}N \big((1+2\gamma_t)\fplus)*\big(\hat{v}+\hat{v}(0)\big)(p)(1+2\gamma_t(p)) \\
          & + \, \frac{4\lambda}{N}\Re\Big(\big(\fplus\sigtb_t\big)*\hat{v}(p)\sigt_t(p)\Big) \,.
    \end{split}
    \end{align}
    Recalling the total fields \eqref{def-Gamt}, \eqref{def-Sigt}, and employing \eqref{eq-HFB-ren-2}, we thus conclude that
    \begin{align}\label{eq-Omt-explicit}
        \begin{split}
          \Omt_t \, = \, E+\frac{\lambda}N\big(\Gamt_t*(\hat{v}+\hat{v}(0))\big) \, + \, \frac{\lambda}N\frac{\Re\big(\Sigtb*\hat{v})\sigt_t\big)}{1+\gamt_t} \, .
        \end{split}
    \end{align}
    \eqref{eq-Omt-explicit} corresponds to \eqref{eq-Omo-explicit} with renormalized fields.
    \par We conclude this section by equivalently reformulating the collected renormalization conditions \eqref{eq-condensate-cond-0}, \eqref{eq-aa-cond-0}, \eqref{eq-omt-ren-cond}:
    \begin{equation}\label{eq-ren-cond-2}
        \begin{cases}
            \HBEC(t)+\Hcubcon(t) \, &= \, 0 \, , \\
            \HHFB(t)+\Hquartcon(t) \, &= \, 0 \, , 
        \end{cases}
    \end{equation}
    compared with the first-order renormalization conditions \eqref{eq-phi-ren-0}, \eqref{eq-gamma-sigma-ren-0}.
        
	\subsection{Boltzmann collision terms} 
 The discussion in section \eqref{sec-perturbation} yields
	\begin{align}\label{eq-Boltzmann+error-exp}
		f_t(p) =&  f_0(p) \, + \, \frac1N\int_0^t\dx{s}Q_3[f_0](s,p) \\
            & \, + \,  \int_{[0,t]^3}\dx{\bs_3}\frac{\jb{[[[\ad_pa_p,\Hfluc(s_1)],\Hfluc(s_2)],\Hfluc(s_3)]}_{s_3}}\vol \, ,
	\end{align} 
	where
	\begin{align}\label{def-cub-boltzmann-0}
        \begin{split}
		\frac1N\int_0^t\dx{s}Q_3[f_0](s,p) \, =& \,  - \, \int_{[0,t]^2} \dx{\bs_2}  \mathds{1}_{s_1\geq s_2} \Big(\frac{\jb{\wick{[[\c1 \ad_p \c2 a_p, \settowidth{\wdth}{$\Hcub$}\hspace{.25\wdth}\c3{\vphantom{\Hcub}}\hspace{-.25\wdth}\c1 \Hcub\settowidth{\wdth}{$\Hcub$}\hspace{-.25\wdth}\c1{\vphantom{\Hcub}}\hspace{.25\wdth}(s_1)], \settowidth{\wdth}{$\Hcub$}\hspace{.25\wdth}\c1{\vphantom{\Hcub}}\hspace{-.25\wdth} \c2 \Hcub\settowidth{\wdth}{$\Hcub$}\hspace{-.25\wdth} \c3{\vphantom{\Hcub}}\hspace{.25\wdth}( s_2)]}}_0}\vol \\
		& + \, \frac{\jb{\wick{[[\c2 \ad_p \c1 a_p, \settowidth{\wdth}{$\Hcub$}\hspace{.25\wdth}\c3{\vphantom{\Hcub}}\hspace{-.25\wdth}\c1 \Hcub\settowidth{\wdth}{$\Hcub$}\hspace{-.25\wdth}\c1{\vphantom{\Hcub}}\hspace{.25\wdth}(s_1)], \settowidth{\wdth}{$\Hcub$}\hspace{.25\wdth}\c1{\vphantom{\Hcub}}\hspace{-.25\wdth} \c2 \Hcub\settowidth{\wdth}{$\Hcub$}\hspace{-.25\wdth} \c3{\vphantom{\Hcub}}\hspace{.25\wdth}( s_2)]}}_0}\vol\Big) \, , 
	\end{split}
    \end{align}
	see \eqref{eq-cub-boltz-first-intro-0}, denotes the cubic Boltzmann collision operator.

	\par We show in Lemma \ref{lem-cub-bol} that
	\begin{align}\label{eq-f-main-1}
        \begin{split}
	\MoveEqLeft \int_0^t\dx{s} Q_3[f_0](s,p) \, = \\
	& 2\lambda^2 \Re \int_{[0,t]^2} \dx{\bs_2}  \mathds{1}_{s_1\geq s_2}\int \dx{\bp_3}  \Big(\frac1{2!}\big(\delta(p_1-p)+\delta(p_2-p)-\delta(p_3-p)\big) \\
	& \bbf{1}{2}_{s_1}(\bp_3)\bbfb{1}{2}_{s_2}(\bp_3)e^{i\int_{s_2}^{s_1} \dx{\tau} \big(\Omega_\tau(p_1)+\Omega_\tau(p_2)-\Omega_\tau(p_3)\big)}\delta(p_1+p_2-p_3)\\
	&\big(\fbar(p_1)\fbar(p_2)f_0(p_3)-f_0(p_1)f_0(p_2)\fbar(p_3)\big) \\
	& + \, \frac1{3!}\big(\delta(p_1-p)+\delta(p_2-p)+\delta(p_3-p)\big)\\
	& \bbf{0}{3}_{s_1}(\bp_3)\bbfb{0}{3}_{s_2}(\bp_3)e^{i\int_{s_2}^{s_1} \dx{\tau} \big(\Omega_\tau(p_1)+\Omega_\tau(p_2)+\Omega_\tau(p_3)\big)}\delta(p_1+p_2+p_3) \\
	&\big(\fbar(p_1)\fbar(p_2)\fbar(p_3)-f_0(p_1)f_0(p_2)f_0(p_3)\big)\Big) \, .
	\end{split}
    \end{align}
    \mh{In particular, this provides an explicit expression for the Boltzmann operator in the evolution of $f$.} In Lemma \ref{lem-cub-bol}, we also give an expression for the quartic Boltzmann term $\frac1{N^2}\int_0^t\dx{s} Q_4[f_0](s,p)$.

    \subsection{Error terms}

    We abbreviate
    \begin{equation}\label{def-f[J]}
        f[J] \, := \, \int \dx{p}J(p)\ad_pa_p \, .
    \end{equation}
    Then we define the error
    \begin{align}
        \MoveEqLeft\Rem[f](t)[J] \, :=\\
        \begin{split}\label{eq-f-rem-main}
        & \, -\int \dx{p} J(p) \int_0^t\dx{s}\Big(\frac1N Q_3[f_0](s,p)+\frac1{N^2} Q_4[f_0](s,p)\Big) \\
        & + \, \int_0^t \dx{s} \jb{[f[J],\Hfluc(s)]}_0 \\
        & + \, \int_{[0,t]^2}\dx{\bs_2}\jb{[[f[J],\Hfluc(s_1)],\Hfluc(s_2)]}_0
        \end{split}\\
        & \, - \, \frac1{N^2}\int_0^t\dx{s} Q_4[f_0](s,p) \label{eq-f-rem-quart-Boltz}\\ 
        & \, + \, \frac1N\int_0^t\dx{s} \big(Q_3[f](s,p)-Q_3[f_0](s,p)\big)\label{eq-f-rem-center}\\
        & \, + \,  \int_{[0,t]^3}\dx{\bs_3}\jb{[[[f[J],\Hfluc(s_1)],\Hfluc(s_2)],\Hfluc(s_3)]}_{s_3}  \, . \label{eq-f-rem-tail}
    \end{align}
    Using the fact that odd moments of $a^{\#}$ w.r.t. $\jb{\cdot}_0$ vanish and rearranging the terms in the first three lines of \eqref{eq-f-rem-main}, yields
    \begin{align}
        \begin{split}\label{eq-main-error-terms}
            &\int_0^t \dx{s} \jb{[f[J],\HBEC(s)+\Hcub(s)]}_0 \, - \, \int \dx{p} J(p) \frac1N\int_0^t\dx{s}Q_3[f_0](s,p)\\
            & \quad + \, \int_{[0,t]^2} \dx{\bs_2} \jb{[[f[J],\HBEC(s_1)+\Hcub(s_1)],\HBEC(s_2)+\Hcub(s_2)]}_0
        \\
            & + \, \int_0^t \dx{s} \jb{[f[J],\HHFBod(s)+\Hquartconod{s}]}_0 \\
            & \quad + \, \int_{[0,t]^2} \dx{\bs_2} \jb{[[f[J],\HHFBod(s_1)+\Hquartconod{s_1}],\HHFB(s_2)+\Hquart(s_2)]}_0
        \\
        & + \, \int_0^t \dx{s} \jb{[f[J],\HHFBd(s)+\Hquartcond{s}]}_0 \\
            & \quad + \, \int_{[0,t]^2} \dx{\bs_2} \jb{[[f[J],\HHFBd(s_1)+\Hquartcond{s_1}],\HHFB(s_2)+\Hquart(s_2)]}_0
        \\
            & + \, \int_{[0,t]^2} \dx{\bs_2} \jb{[[f[J],\Hquart(s_1)-\Hquartcon(s_1)],\Hquart(s_2)]}_0 \\
            & \quad - \, \int \dx{p} J(p) \frac1{N^2}\int_0^t\dx{s} Q_4[f_0](s,p) \, . 
        \end{split} 
    \end{align}
    Due to our choice of HFB fields $(\phit,\gamt,\sigt)$ satisfying \eqref{eq-HFB-ren-2} and the Bogoliubov dispersion $\Omt$ satisfying \eqref{eq-Omt-explicit}, we have that \eqref{eq-main-error-terms} vanishes.
    \begin{enumerate}
        \item In order to estimate the terms coming from the tail \eqref{eq-f-rem-tail}, we compute $\HBEC$, $\HHFBod$, and $\HHFBd$ for the HFB fields $(\phit,\gamt,\sigt,\Omt)$ in Lemma \ref{lem-HBEC-HHFB-with-HFB-eq}. Then we use the following ideas: 
        \begin{enumerate}
            \item $\int \dx{p}a_p^{\#}\lesssim \nb^{\frac12}$ and that $a_p^{\#}\lesssim (\vol\nb)^{\frac12}$, see Lemmata \ref{lem-prod}, \ref{lem-fg-bound}, and \ref{lem-wick-ord-bd}.
            \item Proposition \ref{prop-fluct-dynam-bound}, followed by Lemma \ref{lem-num-mom}
            implies
            \begin{align}
                \jb{(\nb+\vol)^{\frac{\ell}2}}_t \, & \lesssim_{\fd,\nd{\gamma_0},\ell} \, e^{K_{\ell,\fd,\nd{\gamma_0}}\vd\vol\lambda t}\\
                & \quad \quad \jb{(\nb+\vol)^{\frac{\ell}2}\Big(1+\frac{\nb}{N\vol}\Big)}_0 \\
                & \lesssim_{\fd,\nd{\gamma_0},\vol,\ell} \, e^{K_{\ell,\fd,\nd{\gamma_0}}\vd\vol\lambda t} \vol^{\frac\ell2} \, .
            \end{align}
        \end{enumerate}
        With these steps, we obtain
        \begin{align}
            \MoveEqLeft |\int_{[0,t]^3}\dx{\bs_3}\jb{[[[f[J],\Hfluc(s_1)],\Hfluc(s_2)],\Hfluc(s_3)]}_{s_3} \Big) |\\
            &\lesssim_{\fd,\nd{\gamma_0},\vd,\vol}e^{K_{\fd,\nd{\gamma_0}}\vd\vol\lambda t} \frac{1}{N^{\frac32}}\|J\|_\infty 
        \end{align}
        \item Due to the dependence of $\frac1{N^2}\int_0^t\dx{s} Q_4[f_0](s,p)$ on the HFB fields $(\phi_t,\gamma_t,\sigma_t)$, we will employ a priori estimates established in Corollary \ref{cor-bog-coeff}. With that, we can control \eqref{eq-f-rem-quart-Boltz} by
        \begin{align}
            \frac1{N^2}|\int \dx{p}J(p)\int_0^t\dx{s} Q_4[f_0](s,p)|  \, \lesssim_{\fd,\vd} \, \frac{\lambda^2t^2}{N^2}e^{C_{\fd,\nd{\gamma_0}}\lambda\vd t}\|J\|_\infty \, .
        \end{align}
        \item In order to estimate \eqref{eq-f-rem-center}, we write $Q_3[h]=:Q_3[h,h,h]$, in order to emphasize the dependence on three arguments, each evaluated at different momenta. Then we have that
        \begin{align}
            Q_3[f,f,f]-Q_3[f_0,f_0,f_0] \, =& \, (Q_3[f,f,f]-Q_3[f_0,f,f]) \\
            & \, + \, (Q_3[f_0,f,f] \, - \, Q_3[f_0,f_0,f])\\
            & \, + \, (Q_3[f_0,f_0,f] \,- \, Q_3[f_0,f_0,f_0]) \, .
        \end{align}
        Each of the differences contains a factor $f-f_0\propto \frac1N$. Arguing as in \cite[Chapter 5.3]{chenhott}, we then obtain
        \begin{align}
            \frac1N\Big|\int_0^t\dx{s}\big(Q_3[f](s)[J]-Q_3[f_0](s)[J]\big)\Big| \, \lesssim_{\fd,\nd{\gamma_0},\vd,\vol}e^{K_{\fd,\nd{\gamma_0}}\vd\vol\lambda t} \frac{\|J\|_\infty}{N^{\frac32}} \, .
        \end{align}
    \end{enumerate}
    With these estimates, we obtain 
    \begin{align}
        |\Rem[f](t)[J]| \, \lesssim_{\fd,\nd{\gamma_0},\vd,\vol}e^{K_{\fd,\nd{\gamma_0}}\vd\vol\lambda t} \frac{\|J\|_\infty}{N^{\frac32}} \, .
    \end{align}

    \subsection{Evolution of \texorpdfstring{$g$}{g}}
    With analogous calculations as for $f$, one can show that the leading order term in the evolution of $g$ is given by
    \begin{align}
        -\frac{1}{\vol}\int\dx{p}J(p)\int_{[0,t]^2}\dx{\bs_2} \mathds{1}_{s_1\geq s_2} \jb{[[a_pa_{-p},\wick{
        \settowidth{\wdth}{$\Hcub$}\hspace{.25\wdth}\c1{\vphantom{\Hcub}}\hspace{.25\wdth}\c2{\vphantom{\Hcub}}\hspace{-.5\wdth}\Hcub(s_1)],\settowidth{\wdth}{$\Hcub$}\hspace{.25\wdth}\c1{\vphantom{\Hcub}}\hspace{.25\wdth}\c2{\vphantom{\Hcub}}\hspace{-.5\wdth}\Hcub
        }
        (s_2)]}_0 \, .
    \end{align}
    We show in Lemma \ref{lem-cub-bol-phi-g} that 
        \begin{align}
            \MoveEqLeft -\frac{1}{\vol}\int\dx{p}J(p)\int_{[0,t]^2}\dx{\bs_2} \mathds{1}_{s_1\geq s_2} \jb{[[a_pa_{-p},\wick{
        \settowidth{\wdth}{$\Hcub$}\hspace{.25\wdth}\c1{\vphantom{\Hcub}}\hspace{.25\wdth}\c2{\vphantom{\Hcub}}\hspace{-.5\wdth}\Hcub(s_1)],\settowidth{\wdth}{$\Hcub$}\hspace{.25\wdth}\c1{\vphantom{\Hcub}}\hspace{.25\wdth}\c2{\vphantom{\Hcub}}\hspace{-.5\wdth}\Hcub
        }
        (s_2)]}_0 \\
        & = \, \frac1N\int_0^t \dx{s} Q_3^{(g)}[f_0](s)[J] \, ,
    \end{align}
    where $Q_3^{(g)}$ is given in \eqref{def-g-Q3}.
	\par With analogous calculations as in the case of $f$, and using the fact that $g_0=0$, we obtain that
    \begin{align}
        \int \dx{p} J(p)g_t(p) \, &= \, \frac1N\int_0^t \dx{s} Q_3^{(g)}[f](s)[J] \, + \, \Rem[g](t)[J] \, ,
    \end{align}
    where
    \begin{align}
        |\Rem[g](t)[J]|\, \lesssim_{\fd,\nd{\gamma_0},\vd,\vol}e^{K_{\fd,\nd{\gamma_0}}\vd\vol\lambda t} \frac{\nd{J}}{N^{\frac32}} \, .
    \end{align}

    \subsection{Evolution of \texorpdfstring{$\Phi$}{F}}

    Recalling \eqref{def-Phi-Q3}, we show in Lemma \ref{lem-cub-bol-phi-g} that
    \begin{align}
        \MoveEqLeft \frac1{N^{\frac32}}\int_0^t \dx{s} Q_3^{(\Phi)}[f_0](s) \, = \, -\frac{1}{\vol}\int_{[0,t]^2}\dx{\bs_2} \mathds{1}_{s_1\geq s_2} \jb{[[a_0,\wick{
        \settowidth{\wdth}{$\Hquart$}\hspace{.25\wdth}\c1{\vphantom{\Hquart}}\hspace{.25\wdth}\c2{\vphantom{\Hquart}}\hspace{.25\wdth}\c3{\vphantom{\Hquart}}\hspace{-.75\wdth}\Hquart(s_1)],\settowidth{\wdth}{$\Hcub$}\hspace{.25\wdth}\c1{\vphantom{\Hcub}}\hspace{.25\wdth}\c2{\vphantom{\Hcub}}\hspace{.25\wdth}\c3{\vphantom{\Hcub}}\hspace{-.75\wdth}\Hcub
        }
        (s_2)]}_0 
    \end{align}
    We need to collect all terms involving a factor $N^{-\frac32}$. Consequently, we need to consider the third-order Duhamel expansion of $\Phi$. 

        \begin{enumerate}
            \item Due to $[a_0,\HBEC(s_1)]$ being a scalar, we have that \begin{equation}
                [[[a_0,\HBEC(s_1)],\Hfluc(s_2)],\Hfluc(s_3)] \, = \,  0 \, .
            \end{equation}
            \item \begin{align}
                \MoveEqLeft\jb{[[[a_0,\HHFB(s_1)],\Hfluc(s_2)],\Hfluc(s_3)]}_0 \, =\\
                &  \frac\lambda{N}\jb{[[a_{2,1}(s_1)a_0+a_{2,2}\ad_0],\Hfluc(s_2)],\Hfluc(s_3)]}_0\, \propto \, \frac{\lambda^3}{N^2} \, ,
            \end{align} 
            for $\Hfluc(t)\sim\frac{\lambda}{\sqrt{N}}$, see Proposition \ref{prop-hfluc-est}.
            \item Similarly to the previous step, we have that \begin{align}
                \MoveEqLeft\jb{[[[a_0,\Hquart(s_1)],\Hfluc(s_2)],\Hfluc(s_3)]}_0 \, \propto \, \frac{\lambda^3}{N^2} \, .
            \end{align}
            \item We are left with computing
            \begin{align}
                \MoveEqLeft\jb{[[[a_0,\Hcub(s_1)],\Hfluc(s_2)],\Hfluc(s_3)]}_0 \, =\\
                &  \frac\lambda{\sqrt{N}}\jb{[[f[a_{3,1}(s_1)]+g[a_{3,2}(s_1)]+g^\dagger[a_{3,3}(s_1)]],\Hfluc(s_2)],\Hfluc(s_3)]}_0 \, ,
            \end{align}
            where we abbreviated
            \begin{align}
                g[J] \, := \, \int\dx{p}J(p)a_pa_{-p} \, , \quad g^\dagger[J] \, := \, \int\dx{p}J(p)\ad_p\ad_{-p}  \, .\label{def-gJ}
            \end{align}
            Employing Lemma \ref{lem-hfluc}, we find that
            \begin{align}
                a_{3,1}(t,p) \, &= \, e^{i\int_0^t\dx{\tau}\Omega_\tau(0)}\bbf{1}{2}_t(0,p,p) \, ,\\
                a_{3,2}(t,p) \, &= \, \frac12e^{i\int_0^t\dx{\tau}\Omega_\tau(0)}\bbfb{1}{2}_t(p,p,0) \, ,\\
                a_{3,3}(t,p) \, &= \, \frac12e^{i\int_0^t\dx{\tau}\Omega_\tau(0)}\bbf{0}{3}_t(p,p,0) \, .
            \end{align}
            We recognize the evolutions of $f$ and $g$, resulting in
            \begin{align}
                \MoveEqLeft\int_{[0,t]^3}\dx{\bs_3}\mathds{1}_{s_1\geq s_2\geq s_3} \jb{[[[a_0,\Hcub(s_1)],\Hfluc(s_2)],\Hfluc(s_3)]}_0 \, =\\
                &  \frac\lambda{N^{\frac32}}\int_0^t\dx{s} e^{i\int_0^s\dx{\tau}\Omega_\tau(0)}\int\dx{p}\Big(\bbf{1}{2}_s(0,p,p)Q_3[f_0](s,p) \\
                & \quad + \, \bbfb{1}{2}_s(p,p,0)Q_3^{(g)}[f_0](s,p) \, + \, \bbf{0}{3}_s(p,p,0) \overline{Q_3^{(g)}[f_0](s,p)}\Big) \, .
            \end{align}
            With analogous steps as above and recalling \eqref{def-Phi-Q33}, we then obtain that
            \begin{align}
                \Phi_t \, &= \, \frac1{N^{\frac32}}\int_0^t \dx{s} \big(Q_3^{(\Phi)}[f](s)+Q_{3,3}^{(\Phi)}[f](s)\big)\, + \, \Rem[\Phi](t) \, , 
            \end{align}
            where
            \begin{align}
                |\Rem[\Phi](t)|\, \lesssim_{\fd,\nd{\gamma_0},\vd,\vol}e^{K_{\fd,\nd{\gamma_0}}\vd\vol\lambda t} \frac1{N^2} \, .
            \end{align}
        \end{enumerate}

    \section{Estimates on the HFB evolution\label{sec-HFB-est}}
	
	In order to proceed as in \cite{chenhott}, we need to determine a priori bounds for $u$,$v$, and $\phi$. Observe that it suffices to establish bounds for the HFB fields $\phi$ and $\gamma$, for $|\sigma|^2=(1+\gamma)\gamma$. 
    \par We start by rewriting the renormalized HFB equations \eqref{eq-HFB-ren-1} and \eqref{eq-HFB-ren-2}. In both cases, we get a closed system for the total fields $\phi^{(j)}$ and
    \begin{align}
        \begin{pmatrix}
         \Sigma^{(j)} \\
         \Gamma^{(j)}
        \end{pmatrix} \, = \, \big((1+\delta_{j,2}\fplus)\begin{pmatrix}
            \sigma^{(j)}\\
            \gamma^{(j)}
        \end{pmatrix} \, + \, \delta_{j,2}\fplus \, \begin{pmatrix}
            0\\1
        \end{pmatrix} \, + \, N\vol\delta\begin{pmatrix}
            (\phi^{(j)})^2\\
            |\phi_t^{(j)}|^2
        \end{pmatrix} \, ,
    \end{align} 
    without having to explicitly refer to the reduced fields $(\phi^{(j)},\sigma^{(j)},\gamma^{(j)})$.
    Then the HFB equations \eqref{eq-HFB-ren-1} and \eqref{eq-HFB-ren-2} can be rewritten as
	
	\begin{align}\label{eq-HFB-system}
	\begin{cases}
		i\partial_t \phi_t  \, &= \, \frac{\lambda}{N} \Big( \big(\Gamma_t*(\hat{v}+\hat{v}(0))\big) (0) \phi_t + (\Sigma_t*\hat{v}) (0) \overline{\phi}_t\Big)- 2\lambda\vol\hat{v}(0) |\phi_t|^2 \phi_t \, ,  \\
	\partial_t \Gamma_t \, &= \,  \frac{2\lambda}N \Im\big((\Sigma_t*\hat{v} ) \overline{\Sigma}_t\big) \, , \\
	i\partial_t \Sigma_t \, &= \, 2\big(E+\frac{\lambda}{N}\Gamma_t*(\hat{v}+\hat{v}(0))\big)\Sigma_t \, + \, \frac{\lambda}{N} \big( \Sigma_t* \hat{v}\big) (1+2\Gamma_t) \\
    &\quad \quad - \, 4N\lambda\vol^2\hat{v}(0) |\phi_t|^2 \phi_t^2\delta \, . 	
	\end{cases}	
 	\end{align}

	\subsection{Well-posedness theory for the HFB system}

	\begin{lemma}[Local well-posedness] \label{lem-HFB-lwp}
		Let $(\phi_0,\Gamma_0,\Sigma_0)\in \hfbspace^1$. Assume that $\hat{v}\in L^1_{\sqrt{1+E}}\cap L^\infty(\lattice)$. Then there exists a maximal existence time $0<T\leq \infty$ and a unique maximal mild solution $(\phi,\Gamma,\Sigma)$ of \eqref{eq-HFB-system} on $[0,T)$. If $T<\infty$, then $\lim_{t\to T-} \|(\phi_t,\Gamma_t,\Sigma_t)\|_{\hfbspace^1}=\infty$. The solution depends continuously on the initial datum.
	\end{lemma}
	
	\prf
		We start by abbreviating the nonlinearity 
		\begin{align}
			\cJ_1(\phi,\Gamma,\Sigma) :=&  - \, i\Big[\frac{\lambda}{N} \Big( \big(\Gamma*(\hat{v}+\hat{v}(0))\big) (0) \phi \, + \, (\Sigma*\hat{v}) (0) \overline{\phi}\Big) \, , \label{def-cj1}\\
			& - \,  2\lambda\vol\hat{v}(0) |\phi|^2 \phi\Big] \\
			\cJ_2(\phi,\Gamma,\Sigma) :=&  - \, \frac{2\lambda}N \Im\Big(\big(\overline{\Sigma}*\hat{v} \big) \, \Sigma\Big) \, , \label{def-cj2}\\
			\cJ_3(\phi,\Gamma,\Sigma) :=&  - \, i\Big[2\big(\frac{\lambda}{N}\Gamma*(\hat{v}+\hat{v}(0))\big)\Sigma \, + \, \frac{\lambda}{N} \big( \Sigma* \hat{v}\big) (1+2\Gamma)\\
			& - \, 4N\lambda\vol^2\hat{v}(0)|\phi|^2\phi^2 \, \delta\Big] \, . \label{def-cj3}
		\end{align} 
		Let $\vec{\cJ}:=(\cJ_1,\cJ_2,\cJ_3)$, and $\diag(a,b,c)$ denote the diagonal $3\times 3$ matrix with entries $a$, $b$, $c$, and define
		\begin{align}
			\MoveEqLeft \cL(\widetilde{\phi},\widetilde{\Gamma},\widetilde{\Sigma}) \, := \,\\
			&(\phi_0,\Gamma_0,e^{-2iEt}\Sigma_0) + \int_0^t \dx{s} \vec{\cJ} (\widetilde{\phi}_s,\widetilde{\Gamma}_s,\widetilde{\Sigma}_s)\diag(1,1,e^{-2iE(t-s)}) \, .
		\end{align} 
		A standard calculation, explained in Lemma \ref{lem-HFB-NL-contraction}, yields
		\begin{align}
			\MoveEqLeft \|\vec{\cJ}(\phi_1,\Gamma_1,\Sigma_1) \, - \, \vec{\cJ}(\phi_2,\Gamma_2,\Sigma_2)\|_{\hfbspace^1}\\
			\leq &\,  C \Big(\frac{\lambda}N(\|\hat{v}\|_{L^1_{\sqrt{1+E}}}+\|\hat{v}\|_\infty)\big(\| (\phi_1,\Gamma_1,\Sigma_1)\|_{\hfbspace^1} + \| (\phi_2,\Gamma_2,\Sigma_2)\|_{\hfbspace^1}\big) \\
			& \|(\phi_1-\phi_2,\Gamma_1-\Gamma_2,\Sigma_1-\Sigma_2)\|_{\hfbspace^1} \, + \, \lambda\vol\hat{v}(0)(|\phi_1|^2+|\phi_2|^2)\\
			& \big(1+N\vol^{\frac32}(|\phi_1|+|\phi_2|)\big)|\phi_1-\phi_2| \, .
		\end{align} 
		With the previous estimates, we can apply the standard contraction mapping arguments to $\cL$ to show the existence of a unique solution $(\phi,\Gamma,\Sigma)\in C^0_t\big([0,T),\hfbspace^1\big)\cap C^1_t\big([0,T),\hfbspace^{-1}\big)$ as long as $T>0$ is small enough. In addition, in the usual way, one can show a blow-up alternative, i.e., if the maximal time $T>0$ of existence is finite, we have that $\lim_{t\to T-} \|(\phi_t,\Gamma_t,\Sigma_t)\|_{\hfbspace^1}=\infty$. With standard arguments, we also show the continuous dependence on initial data.
	\endprf 
	
	In order to obtain an expression for the HFB energy functional, we look at the phase factor 
	\begin{align}
		\MoveEqLeft \partial_t S_t \, = \, \\
		& \vol\int \dx{p} \Big[E(p)\gamma_t(p) \, + \, \frac{\lambda}{2N}\Gamo_t*(\hat{v}+\hat{v}(0))(p) \Gamo_t(p) \\
		& + \, \frac{\lambda}{2N}\Sigob_t*\hat{v}(p) \Sigo_t(p) \, - \, N\vol^2 \lambda|\phi_t|^4\hat{v}(p)\delta(p) \\
		&  - \, \frac{N\vol\delta(p)\Re\big(\overline{\phi}_ti\partial_t \phi_t\big)}2 \, - \,  \frac{\Re\big(\overline{\sigma}_t(p)i\partial_t\sigma_t(p)\big)}{2\big(1+\gamma_t(p)\big)}\Big] \, . \label{eq-dS-1}
	\end{align} 
	obtained in Lemma \ref{lem-hfluc}. Observe that the terms not involving any time derivatives correspond to the sum over all complete contractions of $\cH_N$ with expectations 
	\begin{align}
		\jb{a_p} =&  \phi \delta(p) \, , \\
		\jb{\ad_p a_p} =&  \vol\Gamo (p) \, , \\
		\jb{a_p a_{-p}} =&  \vol \Sigo (p) \, ,
	\end{align} 
	see, e.g., \cite{BBCFS22}. Moreover, we can replace $E(p)\gamma_t(p)$ in \eqref{eq-dS-1} by $E(p)\Gamo_t(p)$. In particular, \eqref{eq-dS-1} gives rise to the {\it HFB energy (density) functional}
	\begin{align}
		\ehfb(\phi,\Gamma,\Sigma) :=&  \int \dx{p} \Big(E(p)\Gamma(p) \, + \, \frac{\lambda}{2N}\Gamma*(\hat{v}+\hat{v}(0))(p) \Gamma(p) \\
		& + \, \frac{\lambda}{2N}\overline{\Sigma}*\hat{v}(p) \Sigma(p) \, - \, N\vol^2 \lambda|\phi|^4\hat{v}(p)\delta(p)\Big) \, . \label{def-ehfb}
	\end{align} 
	In addition, we define the {\it (total) HFB density}
	\begin{align}
		\|\Gamma\|_1 \, := \, \int \dx{p} \Gamma(p) \, . \label{def-rhfb}
	\end{align} 
	
	\begin{lemma}[Conservation laws] \label{lem-conservation}
		Assume that $\hat{v}\in L^1_{\sqrt{1+E}}\cap L^\infty(\lattice)$. Let $(\phi_0,\Gamma_0,\Sigma_0)\in \hfbspace^1$, and $(\phi,\Gamma,\Sigma)\in C^0_t\big([0,T),\hfbspace^1\big)\cap C^1_t\big([0,T),\hfbspace^{-1}\big)$ denote the maximal mild solution of \eqref{eq-HFB-system}. Then the HFB density $\|\Gamma_t\|_1$ and HFB energy $\ehfb(\phi_t,\Gamma_t,\Sigma_t)$ are differentiable in time and conserved.
	\end{lemma}
	
	Since the arguments needed for this result are standard, we omit the details at this point. For the interested reader, we elaborate on these details after the proof of Lemma \ref{lem-HFB-NL-contraction}.
	\par In order to close this section on the well-posedness of the HFB equations, recall that 
	\begin{align}
		\gamma_t \, = \, |v_t|^2 \geq 0 \, , \quad |\sigma_t|^2 \, = \, u_t^2|v_t|^2 \, = \, (1+\gamma_t)\gamma_t \, .
	\end{align} 		
	
	Then we have that
		\begin{align}
			\Gamma \, \geq \, \gamma \, \geq \, 0 \, , \label{eq-gamma>=0}
		\end{align} 
		and 
		\begin{align}
			|\Sigma|^2 =& |\sigma+N\vol\phi^2\delta|^2 \\
			\leq & 2\big(|\sigma|^2 + (N\vol|\phi|^2\delta)^2\big)\\
			= & 2\big((\gamma+1)\gamma + (N\vol|\phi|^2\delta)^2\big) \\
			     \leq & 2(\Gamma+1)\Gamma \, , \label{eq-gamma-sigma-rel-1}
		\end{align} 
		where we applied Cauchy-Schwarz in the second line. As a consequence, Young's inequality implies
		\begin{align}
			|\Sigma| \leq & \sqrt{2(\Gamma+1)\Gamma}\\
			\leq & \sqrt{2}(\Gamma+1) \, . \label{eq-sigma-gamma-sqrt-rel-0}
		\end{align} 

    For the next statement, we recall from \eqref{def-gtr-str} the truncated fields
    \begin{align}
        \gtr \, = \, \Gamma-N\vol|\phi|^2\delta \, , \quad \str \, = \, \Sigma-N\vol \phi^2\delta \, .
    \end{align} 
 
	\begin{lemma}[Conditional global well-posedness] \label{prop-HFB-cond-gwp}
		Assume that $\hat{v}\in L^1_{\sqrt{1+E}}\cap L^\infty(\lattice)$, and that $v\geq0$. Let $(\phi_0,\Gamma_0,\Sigma_0)\in\hfbspace^1$. Let $(\phi,\Gamma,\Sigma)$ be the associated unique maximal mild solution of \eqref{eq-HFB-system} with existence time $T_0>0$, and assume that the truncated expectations satisfy $\gtr\geq0$, $|\str|^2\leq(\gtr+1)\gtr$. Then $T_0=\infty$.
	\end{lemma}
	\prf
		Mass-conservation and $\gtr\geq0$ imply that
		\begin{align}
			N\vol|\phi_t|^2 \, \leq & \, \|\Gamma_t\|_1 \\
			= & \, \|\Gamma_0\|_1 \, ,  
		\end{align} 
		which is why 
		\begin{align} 
			|\phi_t| \, \leq\, C_{\|\Gamma_0\|_1} \, .
		\end{align} 
		Next,
		\begin{align}
			\scp{\Sigma_s*\hat{v}}{\Sigma_s} \, = \, \int dx \, |\widecheck{\Sigma}_s(x)|^2 v(x) \, , 
		\end{align} 
		see \eqref{eq-e-sigma>=0}, and positivity of $v$ yield 
		\begin{align}
			\ehfb(\phi_t,\Sigma_t,\Gamma_t) \geq& \int \dx{p} \Big(E(p)\Gamma_t(p) \, + \, \frac{\lambda}{2N}\gtr_t*(\hat{v}+\hat{v}(0))(p) \gtr_t(p) \\
		& + \, \frac{\lambda}{2N}\overline{\Sigma}_t*\hat{v}(p) \Sigma_t(p) \\
	    \geq& \int \dx{p} E(p)\Gamma_t(p)\, . \label{eq-egamma-bd-1}
		\end{align} 
		Because of energy and mass conservation, we thus have that 
		\begin{align}
			\|\Gamma_t\|_{L^1_{1+E}} \, \leq \, C_{\|\Gamma_0\|_1,\ehfb(\phi_0,\Gamma_0,\Sigma_0)} \, . \label{eq-gamma-l1e-bd-1}
		\end{align}  
		Using $\Gamma\geq0$ and \eqref{eq-sigma-gamma-sqrt-rel-0}, we find that
		\begin{align}
			\partial_t \Gamma_t(p) \, \leq\, \frac{4\lambda}N\|(\Gamma_t(p)+1)*\hat{v}\|_\infty (\Gamma_t(p)+1) \, .
		\end{align} 
		In particular, mass conservation implies
		\begin{align}
			\partial_t \Gamma_t(p) \, \leq \, \frac{4\lambda}N(\|\Gamma_0\|_1\|\hat{v}\|_\infty+\|\hat{v}\|_1)(\Gamma_t(p)+1) \label{eq-dt-gamma-1}
		\end{align} 
		Gronwall's inequality then yields the point-wise bound
		\begin{align}
			\Gamma_t(p) \, \leq \, e^{\frac{4\lambda}N(\|\Gamma_0\|_1\|\hat{v}\|_\infty+\|\hat{v}\|_1)t}(\Gamma_t(p)+1) \, , 
		\end{align} 
		which, in turn, gives the uniform bound
		\begin{align}
			\|\Gamma_t\|_\infty \, \leq \, e^{\frac{4\lambda}N(\|\Gamma_0\|_1\|\hat{v}\|_\infty+\|\hat{v}\|_1)t}(\|\Gamma_0\|_\infty+1) \, . \label{eq-gamma-linfty-bd-1}
		\end{align}
		Finally, employing \eqref{eq-gamma-sigma-rel-0}, we find that
		\begin{align}
			\|\Sigma_t\|_{L^2_{1+E}}^2 \leq& 2\int \dx{p} (\Gamma_t(p)+1)\Gamma_t(p) (1+E(p)) \, .
		\end{align} 
		Collecting \eqref{eq-gamma-l1e-bd-1} and \eqref{eq-gamma-linfty-bd-1}, we thus obtain
		\begin{align}
			\|\Sigma_t\|_{L^2_{1+E}}^2 \leq& C_{\|\Gamma_0\|_1,\ehfb(\phi_0,\Gamma_0,\Sigma_0),\|\Gamma_0\|_\infty} e^{\frac{4\lambda}N(\|\Gamma_0\|_1\|\hat{v}\|_\infty+\|\hat{v}\|_1)t}\, .
		\end{align}
		\eqref{eq-sigma-gamma-sqrt-rel-0} and \eqref{eq-gamma-linfty-bd-1} imply
		\begin{align}
			\|\Sigma_t\|_\infty \, \leq \, 2e^{\frac{4\lambda}N(\|\Gamma_0\|_1\|\hat{v}\|_\infty+\|\hat{v}\|_1)t}(\|\Gamma_0\|_\infty+1) \, .
		\end{align} 
		In particular, we have that $\|(\phi_t,\Gamma_t,\Sigma_t)\|_{\hfbspace^1}<\infty$ for all $t\geq0$, which, by the blow-up alternative in Lemma \ref{lem-HFB-lwp}, implies $T_0=\infty$.
 	\endprf 
	
	\subsection{Symplectic description of HFB equations\label{sec-symplectic}}
	In this section, we follow closely the ideas in \cite{BBCFS22}. \mh{We invoke a symplectic description as a convenient path to establishing global well-posedness.} Let
	\begin{align}
		\symp_t \, := \, \begin{pmatrix}
			\gtr_t & \str_t \\ \strb_t & 1+ \gtr_t 
		\end{pmatrix} \, ,
	\end{align} 
	where we recall the definitions \eqref{def-gtr-str} of $\gtr$ and $\str$. Observe that we have that
	\begin{align}
		\symp \, \geq \, 0 \quad \Leftrightarrow \quad \gtr\geq0 \, \wedge |\str|^2 \, \leq \, (\gtr+1)\gtr \, . \label{eq-symp-pos-equiv-0}
	\end{align} 
	The goal of this section is to prove that $\symp_0\geq0$ implies $\symp_t\geq0$ along the flow induced by the HFB equations \eqref{eq-HFB-system}. More precisely, let
	\begin{align}
		\hgam := & E \, + \, \frac{\lambda}N \Gamma*(\hat{v}+\hat{v}(0)) \label{def-hgam} \, , \\
		\hsig := & \frac{\lambda}N\Sigma*\hat{v} \, , \label{def-hsig} \\
		\hsymp :=&  \begin{pmatrix} \hgam & \hsig \\ \overline{\hsig} & \hgam \end{pmatrix} \, , \label{def-hsymp} \\
		\ssymp :=&  \begin{pmatrix} 1 & 0 \\ 0 & -1 \end{pmatrix} \, . \label{def-ssymp}
	\end{align} 
	A straightforward calculation yields
	\begin{align}
		i\partial_t \symp_t = \ssymp \hsympt \symp_t -\symp_t \hsympt \ssymp \, .
	\end{align} 
	Let
	\begin{align} \label{def-sympbog}
		\begin{cases}
			i\partial_t \sympbogd_t &= \hsympt \ssymp \sympbogd_t \, , \\ 
			\sympbogd_0 &= \1 \, .
		\end{cases}
	\end{align} 
	We introduce the function space
	\begin{align}
		\sympspace^j \, := \, \Big\{ \begin{pmatrix}
			a & b \\ \overline{b} & a
		\end{pmatrix} \mid a\in L^1_{(1+E)^j}(\lattice) \, , \, b\in L^2_{(1+E)^j}(\lattice)\Big\} \, .
	\end{align} 
	
	\begin{lemma}[Well-posedness of \eqref{def-sympbog}] \label{lem-sympbog-wp}
		Assume that $\hat{v}\in L^1_{\sqrt{1+E}}\cap L^\infty(\lattice)$. Let $(\phi,\Gamma,\Sigma)\in C^0_t\big([0,T),\hfbspace^1\big)$. Then there is a unique solution $\sympbogd \in C^1_t\big([0,T),\hfbspace^1\big) $ of \eqref{def-sympbog}. 
	\end{lemma}
	\prf
		We have that $E\ssymp$ is the generator of a continuous semi-group on $\sympspace^1$. In addition, $(\phi,\Gamma,\Sigma)\in C^0_t\big([0,T),\hfbspace^1\big)$ and the estimates in the proof of Lemma \ref{lem-HFB-NL-contraction} imply the continuity of  $t\mapsto\hsympt-E\ssymp \in \sympspace^1$. Finally, we apply standard functional results, e.g., as found in \cite{kato70}, in order to show the well-posedness and regularity of $\sympbogd_t$. 
	\endprf 
	Now let $(\phi,\Gamma,\Sigma)\in C^0_t\big([0,T),\hfbspace^1\big)\cap C^1_t\big([0,T),\hfbspace^{-1}\big)$ be the solution of \eqref{eq-HFB-system}. Using Lemma \ref{lem-sympbog-wp}, we have that
	\begin{align}
		i\partial_t (\sympbog_t \symp_t \sympbogd_t) \, =& \, -\sympbog_t \ssymp\hsympt\symp_t \sympbogd_t \, + \, \sympbog_t \symp_t \hsympt\ssymp\sympbogd_t\\
		& + \, \sympbog_t\ssymp \hsympt \symp_t \sympbogd_t \, - \, \sympbog_t\symp_t \hsympt \ssymp \sympbogd_t \\
		=& \, 0 \, .
	\end{align} 
	In particular, provided $\symp_0\geq0$, we have that
	\begin{align}
		\symp_t \, = \, \sympbogd_t\symp_0 \sympbog_t \, \geq \, 0 \, . \label{eq-symp>=0}
	\end{align} 
	Together with \eqref{eq-symp-pos-equiv-0}, we have thus proved Proposition \ref{prop-HFB-gwp}.

	\subsection{Bounds on the HFB fields} 
 
    After establishing global well-posedness of \eqref{eq-HFB-system}, as well as energy and mass conservation, we are ready to establish bounds on $\gtr$ and $\phi$. As described at the beginning of section, these, in turn, will yield a bound on $\gamma$, and thus on $u$ and $v$. For that, we establish the next result. Recall from \eqref{def-gtr-str} that 
    $\gtr=\Gamma-N\vol|\phi|^2\delta$, $\str=\Sigma-N\vol\phi^2\delta$.
	
	\begin{lemma}[Mass transfer bound]\label{lem-mass-transfer}
		Assume that $\hat{v}\in L^1_{\sqrt{1+E}}\cap L^\infty(\lattice)$, and that $v\geq0$. Let $(\phi_0,\Gamma_0,\Sigma_0)\in \hfbspace^1$  with $\gtr_0\geq0$, $|\str_0|^2\leq\sqrt{(\gtr_0+1)\gtr_0}$, and $(\phi,\Gamma,\Sigma)\in C^0_t\big([0,T),\hfbspace^1\big)\cap C^1_t\big([0,T),\hfbspace^{-1}\big)$ denote the global solution of \eqref{eq-HFB-system}. Then we have that
		\begin{align}
			|\phi_t|^2 \, = & \, |\phi_0|^2 \, + \, \frac{\|\gtr_0\|_1-\|\gtr_t\|_1}{N\vol} \, , \\
			\gtr_t \, \leq & \, e^{2\lambda \vd(\frac{\|\gtr_0\|_1}N+2)t}(\gtr_0+1) \, , \\
			\|\gtr_t\|_1 \, \leq & \, e^{2\lambda \vd(\frac{\|\gtr_0\|_1}N+1)t}(\|\gtr_0\|_1+1) \, , \\
			\int \dx{p} E(p)\Gamma_t(p) \, \leq& \, \ehfb(\phi_0,\Gamma_0,\Sigma_0) \, ,
		\end{align}
		and $|\str_t| \, \leq \, \sqrt{(\gtr_t+1)\gtr_t}$.
	\end{lemma}
	\prf
		Observe that we have $\|\Gamma\|_1=\int \dx{p}\Gamma(p)$ for $\Gamma\geq0$. Mass conservation implies
		\begin{align}
		|\phi_t|^2 \, = \, \frac{\|\Gamma_0\|_1-\|\gtr_t\|_1}{N\vol} \, = \, |\phi_0|^2 \, + \, \frac{\|\gtr_0\|_1-\|\gtr_t\|_1}{N\vol} \, . \label{eq-phi-growth-1}
		\end{align} 
        Notice that Lemma \ref{prop-HFB-gwp} yields
		\begin{align}
			|\str_t| \, \leq \, \sqrt{(\gtr_t+1)\gtr_t} \, \leq \, \gtr_t+1 , \label{eq-str-gtr-bd-1}
		\end{align} 
		due to Young's inequality. \eqref{eq-HFB-system} implies 
		\begin{align}
			\partial_t \gtr_t \, = \, \frac{2\lambda}N\Im\big((\Sigma_t*\hat{v})\str_t\big) \, . \label{eq-dt-gtr-0}
		\end{align} 
		Using \eqref{eq-str-gtr-bd-1}, Young's inequality yields
		\begin{align}
			\partial_t \gtr_t \, \leq & \,\frac{2\lambda}N (\|\hat{v}\|_1+\|\hat{v}\|_\infty)(\|\Gamma_t\|_1+1)(\gtr_t+1) \\
			=& \, \frac{2\lambda}N \vd(\|\Gamma_0\|_1+1)(\gtr_t+1) \, ,
		\end{align} 
		where we employed mass conservation. We thus obtain, via Gronwall's inequality and using the fact $\|\Gamma_0\|_1=\|\gtr_0\|_1+N$, that
		\begin{align}
			\gtr_t \, \leq\, e^{2\lambda \vd\big(\frac{\|\gtr_0\|_1}N+2\big)t}(\gtr_0+1) \, . \label{eq-gtr-pw-bd-1}
		\end{align} 
		Similarly, \eqref{eq-dt-gtr-0}, together with \eqref{eq-str-gtr-bd-1} and \eqref{eq-phi-growth-1}, yields
		\begin{align}
		 	\partial_t \|\gtr_t\|_1 \, = & \, 2\lambda\vol\Im\big(\int \dx{p} (\Sigma_t*\hat{v})(p)\strb_t(p)\big) \\
		 	= & \, 2\lambda\vol\Im\big((\strb_t*\hat{v})(0)\phi_t^2\big) \\
		 	\leq& \, 2\lambda \vd\frac{\|\Gamma_0\|_1-\|\gtr_t\|_1}{N}(\|\gtr_t\|_1+1) \, ,
		 \end{align} 
		 where, analogously to \eqref{eq-e-sigma>=0}, we used that $\int \dx{p} (\str_t*\hat{v})(p)\strb_t(p)\in\R$. Then Gronwall's inequality implies
		 \begin{align}
		 	\|\gtr_t\|_1 \, \leq\, e^{2\lambda \vd(\frac{\|\gtr_0\|_1}N+1)t}(\|\gtr_0\|_1+1) \, . \label{eq-gtr-l1-growth-1}
		 \end{align} 
		 \eqref{eq-egamma-bd-1} together with energy conservation implies
		 \begin{align}
		 	\int \dx{p} E(p)\Gamma_t(p) \, \leq \, \ehfb(\phi_t,\Gamma_t,\Sigma_t) \, = \, \ehfb(\phi_0,\Gamma_0,\Sigma_0) \, .
		 \end{align} 
		 This concludes the proof. 
	\endprf 
	As a direct consequence, we obtain the following result.
	\begin{corollary}[Bounds on Bogoliubov coefficients] \label{cor-bog-coeff}
		Assume that $\hat{v}\in L^1_{\sqrt{1+E}}\cap L^\infty(\lattice)$. Let $(\phi_0,\Gamma_0,\Sigma_0)\in \hfbspace^1$ with $\gtr_0\geq0$, $|\str_0|^2\leq\sqrt{(\gtr_0+1)\gtr_0}$, and $(\phi,\Gamma,\Sigma)\in C^0_t\big([0,T),\hfbspace^1\big)\cap C^1_t\big([0,T),\hfbspace^{-1}\big)$ denote the global solution of \eqref{eq-HFB-system}. Let $\gamma_t,\sigma_t: \lattice\to \C$ with $0\leq\gamma_t\leq \gtr_t$ and $|\sigma_t|^2=(\gamma_t+1)\gamma_t$ be given, and define 
		\begin{align}
			u_t(p) \, &:= \, \sqrt{1+\gamma_t(p)} \, , \\
			v_t(p) \, &:= \, \frac{\sigma_t(p)}{\sqrt{1+\gamma_t(p)}} \, , 
		\end{align}
		i.e., $|v_t|^2=\gamma_t$. Then we have that
		\begin{align}
			1 \, + \, \|\gamma_t\|_\infty \, = \, 1+\|v_t\|_\infty^2 \, = \,\|u_t\|_\infty^2 \, &\leq \, 1 \, + \, e^{2\lambda \vd(\frac{\|\gtr_0\|_1}N+2)t}(\|\gtr_0\|_\infty+1) \, , \\
			\|v_t\|_2^2 \, = \, \|\gamma_t\|_1 \, &\leq  \, e^{2\lambda \vd(\frac{\|\gtr_0\|_1}N+1)t}\|\gtr_0\|_1 \, , \\
			\|v_t\|_{L^2_E}^2 \, = \, \|\gamma_t\|_{L^1_E} \, &\leq \, \ehfb(\phi_0,\Gamma_0,\Sigma_0)  \, .
		\end{align} 
	\end{corollary}

    \section{Tail estimates}

    In this section, we bound the tail in the perturbation expansion\mh{, i.e., \eqref{eq-f-tail} and the corresponding terms in the evolutions of $\Phi$ and $g$.} Recall from \eqref{def-Gamt} and \eqref{def-Sigt} that 
    \begin{align}
        \Gamt \, & = \, (1+2\fplus)\gamt+\fplus+N\vol|\phit|^2\delta \, , \\
        \Sigt \, & = \, (1+2\fplus)\sigt+ N\vol(\phit)^2\delta \, .
    \end{align}

    \begin{proposition}[Bound on Fluctuation dynamics]\label{prop-fluct-dynam-bound}
        Assume that $\hat{v}\in L^1_{\sqrt{1+E}}\cap L^\infty(\lattice)$, and that $\vol\geq1$. Let $(\phit,\Gamt,\Sigt)\in C^0_t\big([0,T),\hfbspace^1\big)\cap C^1_t\big([0,T),\hfbspace^{-1}\big)$ denote the global solution of \eqref{eq-HFB-system} with initial datum 
        \begin{equation}
            (\phit_0,\Gamt_0,\Sigt_0) \, = \, (\phi_0,(1+2\fplus)\gamma_0+\fplus+N\vol|\phi_0|^2\delta,(1+2\fplus)\sigma_0+N\vol\phi_0^2\delta) \in \hfbspace^1 \, .
        \end{equation}
        Then for any $\ell\in\N$, there exist constants $C_\ell,K_\ell>0$ such that for all $t>0$, we have that
        \begin{align}
            \MoveEqLeft\Big\|(\nb+\vol)^{\frac\ell2} \fluc(t)(\nb+\vol)^{-\frac\ell2}\big(1+\frac{\nb}{N\vol}\big)^{-\frac12}\Big\| \\
            & \leq \, C_\ell \Big(1+\frac{(1+\nd{f_0})^2(1+\nd{\gamma_0})^2}N\Big)^{\frac\ell2} (1+\nd{f_0})^{\frac\ell2}(1+\nd{\gamma_0})^\ell \\
            & \qquad e^{K_\ell\vd\lambda\vol(1+\frac{(1+\nd{f_0})(1+\nd{\gamma_0})}N)t} \, .
        \end{align}
    \end{proposition}
    \begin{proof}
        Lemma \ref{lem-bog-dynamics} implies
        \begin{align}
            \big\|(\nb+\vol)^{\frac\ell2} \bog[k_0](\nb+\vol)^{-\frac\ell2}\big\| \, \lesssim_\ell & \, (1+\|\gamma_0\|_1+\|\gamma_0\|_\infty)^{\frac\ell2} \, ,\label{eq-nb-bog-flucdyn-bd}\\
            \big\|(\nb+\vol)^{\frac\ell2} \bog[k_t](\nb+\vol)^{-\frac\ell2}\big\| \, \lesssim_\ell & \, (1+\|\gamma_t\|_1+\|\gamma_t\|_\infty)^{\frac\ell2} \, .
        \end{align}
        By definition \eqref{def-gtr-str}, we have that $(\Gamt)^T=(1+2\fplus)\gamt+\fplus$. Corollary \ref{cor-bog-coeff} then implies
        \begin{align}\label{eq-nb-bogt-HFB-bd}
            \begin{split}
            \MoveEqLeft \big\|(\nb+\vol)^{\frac\ell2} \bog[k_t](\nb+\vol)^{-\frac\ell2}\big\| \\
            &\lesssim_\ell \, e^{\ell\lambda \vd(\frac{\|(\Gamt_0)^T\|_1}N+1)t}(1+\|(\Gamt_0)^T\|_1+\|(\Gamt_0)^T\|_\infty)^{\frac\ell2}\\
            &\lesssim_\ell \, e^{\ell\lambda \vd(\frac{(1+2\|f_0\|_\infty)\|\gamma_0\|_1+\|f_0\|_1}N+1)t}(1+\nd{f_0})^{\frac\ell2}(1+\nd{\gamma_0})^{\frac\ell2} \, .
            \end{split}
        \end{align}
        Lemma \ref{lem-weyl-fluct-dynamics} implies that for $\whfluc(t):=\weyld[\sqrt{N\vol}\phi_t]e^{-it\cH_N}\weyl[\sqrt{N\vol}\phi_0]$, we have that
        \begin{align}\label{eq-BEC-fluc-dyn}
        \begin{split}
		\MoveEqLeft\Big\|(\nb+\vol)^{\frac\ell2} \whfluc(t)(\nb+\vol)^{-\frac\ell2}\big(1+\frac{\nb}{N\vol}\big)^{-\frac12}\Big\| \\
        & \leq \, C_\ell \Big(1+\frac{(1+\nd{f_0})^2(1+\nd{\gamma_0})^2}N\Big)^{\frac\ell2} e^{K_\ell\vd\lambda\vol(1+\frac{(1+\nd{f_0})(1+\nd{\gamma_0})}N)t} \, .
        \end{split}
		\end{align}
        Collecting \eqref{eq-nb-bog-flucdyn-bd}, \eqref{eq-nb-bogt-HFB-bd}, and \eqref{eq-BEC-fluc-dyn}, we obtain
        \begin{align}
            \MoveEqLeft\Big\|(\nb+\vol)^{\frac\ell2} \fluc(t)(\nb+\vol)^{-\frac\ell2}\big(1+\frac{\nb}{N\vol}\big)^{-\frac12}\Big\| \\
            & \leq \, C_\ell \Big(1+\frac{(1+\nd{f_0})^2(1+\nd{\gamma_0})^2}N\Big)^{\frac\ell2} (1+\nd{f_0})^{\frac\ell2}(1+\nd{\gamma_0})^\ell \\
            & \qquad e^{K_\ell\vd\lambda\vol(1+\frac{(1+\nd{f_0})(1+\nd{\gamma_0})}N)t} \, ,
        \end{align}
        which concludes the proof.
    \end{proof}

    \begin{proposition}\label{prop-hfluc-est}
        Under the same assumptions as in Proposition \ref{prop-fluct-dynam-bound}, we have that
        \begin{align}
            \MoveEqLeft \frac{\|\HBEC(t) P_n\|}{\sqrt{n+\vol}}, \, \sqrt{N}\frac{\|\HHFBod(t) P_n\|}{n+\vol}, \, \sqrt{N}\frac{\|\HHFBd(t) P_n\|}{n}, \, \frac{\|\Hcub(t) P_n\|}{(n+\vol)^{\frac32}}, \, \sqrt{N}\frac{\|\Hquart(t) P_n\|}{(n+\vol)^2} \\
            & \leq \, C\frac{\lambda}{\sqrt{N}}e^{C\lambda\vd(1+\frac{(1+\fd)(1+\nd{\gamma_0})}N)t}\vd\fd (1+\fd)^2 (1+\nd{\gamma_0})^2 \, . 
        \end{align}
    \end{proposition}
    \begin{proof}
        Recalling from \eqref{eq-gamma-sigma-rel-0} that $|\sigma|^2=(1+\gamma)\gamma$, Lemmata \ref{lem-HBEC-HHFB-with-HFB-eq} and \ref{lem-wick-ord-bd} imply 
        \begin{align}
            \|\HBEC(t) P_n\| \, \leq& \, C\frac{\lambda\vol}{\sqrt{N}}(u_t(0)+|v_t(0)|)|\phi_t|\\
            & \, (|(\fplus\sigma_t)*\hat{v}(0)|+|\big((1+2\gamma_t)\fplus\big)*(\hat{v}+\hat{v}(0))(0)|)\sqrt{n+1} \\
            \leq& \, C\frac{\lambda\vol}{\sqrt{N}}\vd (u_t(0)+|v_t(0)|)|\phi_t| \fd(1+\nd{\gamma_t})\sqrt{n+1} \, .
        \end{align}
        Now we employ Lemma \ref{lem-mass-transfer} and Corollary \ref{cor-bog-coeff}, and obtain
        \begin{align}
            \frac{\|\HBEC(t) P_n\|}{\sqrt{n+1}} \, \leq& \, C\frac{\lambda }{\sqrt{N}} e^{C\lambda\vd(1+\frac{\|\gtr_0\|_1}N)t} \vd\fd  \Big(1+\frac{\|\gtr_0\|_1}N\Big)^{\frac12}(1+\nd{\gtr_0})^{\frac32} \\
            \leq& \, C\frac{\lambda }{\sqrt{N}} e^{C\lambda\vd(1+\frac{(1+\fd)(1+\nd{\gamma_0})}N)t}\vd\fd (1+\fd)^2 (1+\nd{\gamma_0})^2 \ . 
        \end{align}
        The bounds on $\|\HHFBod(t) P_n\|$ and $\|\HHFBd(t) P_n\|$, see Lemma \ref{lem-HBEC-HHFB-with-HFB-eq}, and on $\|\Hcub(t) P_n\|$ and $\|\Hquart(t) P_n\|$, see Lemma \ref{lem-hfluc}, can be obtained in an analogous fashion, using Lemmata \ref{lem-fg-bound} and \ref{lem-wick-ord-bd}, respectively. This concludes the proof.
    \end{proof}

    \begin{corollary}
        Under the same assumptions as in Proposition \ref{prop-fluct-dynam-bound}, and for $\vol\leq N$, we have that
        \begin{align}
            \MoveEqLeft|\int_{[0,t]^3}\dx{\bs_3}\jb{[[[f[J],\Hfluc(s_1)],\Hfluc(s_2)],\Hfluc(s_3)]}_{s_3}|\\
            \leq& \, C_{\fd}\|J\|_2 \frac1{N^{\frac32}} e^{C\lambda\vol\vd(1+\frac{(1+\fd)(1+\nd{\gamma_0})}N)t}(1+\nd{\gamma_0})^6\vol^{\frac{11}2} \, .
        \end{align}
    \end{corollary}
    
    \begin{proof}
        Proposition \ref{prop-hfluc-est} and Lemma \ref{lem-prod} imply
        \begin{align}
            \MoveEqLeft|\int_{[0,t]^3}\dx{\bs_3}\jb{[[[f[J],\Hfluc(s_1)],\Hfluc(s_2)],\Hfluc(s_3)]}_{s_3}|\\
            \leq& \, C\|J\|_2 \frac{\lambda^3t^3}{N^{\frac32}}e^{C\lambda\vd(1+\frac{(1+\fd)(1+\nd{\gamma_0})}N)t}\vd^3\fd^3 \\
            & \, (1+\fd)^6 (1+\nd{\gamma_0})^6\sup_{s\in[0,t]}\jb{(\nb+\vol)^{1+\frac92}\Big(1+\frac{(\nb+\vol)^{\frac12}}{\sqrt{N}}\Big)^3}_s \, .
        \end{align}
        Proposition \ref{prop-fluct-dynam-bound} then yields
        \begin{align}
            \MoveEqLeft|\int_{[0,t]^3}\dx{\bs_3}\jb{[[[f[J],\Hfluc(s_1)],\Hfluc(s_2)],\Hfluc(s_3)]}_{s_3}|\\
            \leq& \, C\|J\|_2 \frac{\lambda^3t^3}{N^{\frac32}}e^{C\lambda\vol\vd(1+\frac{(1+\fd)(1+\nd{\gamma_0})}N)t}\vd^3\fd^3 \\
            & \, (1+\fd)^6 (1+\nd{\gamma_0})^6\jb{(\nb+\vol)^{\frac{11}2}\Big(1+\frac{(\nb+\vol)^{\frac12}}{\sqrt{N}}\Big)^3\Big(1+\frac{\nb}{N\vol}\Big)}_0 \, .
        \end{align}
        Finally, we employ Lemma \ref{lem-num-mom} and obtain
        \begin{align}
            \MoveEqLeft|\int_{[0,t]^3}\dx{\bs_3}\jb{[[[f[J],\Hfluc(s_1)],\Hfluc(s_2)],\Hfluc(s_3)]}_{s_3}|\\
            \leq& \, C_{\fd}\|J\|_2 \frac{\lambda^3t^3}{N^{\frac32}}e^{C\lambda\vol\vd(1+\frac{(1+\fd)(1+\nd{\gamma_0})}N)t}\vd^3\\
            & \, (1+\nd{\gamma_0})^6\vol^{\frac{11}2}\Big(1+\frac{\vol^{\frac12}}{N^{\frac12}}\Big)^3 \\
            \leq& \, C_{\fd}\|J\|_2 \frac1{N^{\frac32}} e^{C\lambda\vol\vd(1+\frac{(1+\fd)(1+\nd{\gamma_0})}N)t}(1+\nd{\gamma_0})^6\vol^{\frac{11}2} \, .
        \end{align}
    \end{proof}

\appendix 

\section{Calculation of the fluctuation dynamics}
	
	\begin{lemma}\label{lem-bogoliubov}
		We have that 
		\begin{align}
			\bogd[k]a_p\bog[k] =&  \cosh\big(|k(p)|\big)a_p \, + \, \sinh\big(|k(p)|\big)\frac{k(p)}{|k(p)|}\ad_{-p} \, .
		\end{align} 	
	\end{lemma}
	\prf
		Observe that
		\begin{align}
		\partial_\tau\begin{pmatrix}\bogd[\tau k]a_p\bog[\tau k]\\ \bogd[\tau k]\ad_{-p}\bog[\tau k]\end{pmatrix} = &
		\begin{pmatrix}
			 \frac12\int dq \, k(q)\bogd[\tau k][a_p,\ad_q\ad_{-q}]\bog[\tau k] \\
			 -\frac12\int dq \, k(q)\bogd[\tau k][\ad_{-p},a_qa_{-q}]\bog[\tau k]
		\end{pmatrix}\\
		=&  \begin{pmatrix}
			0 & k(p) \\
			\overline{k}(p) & 0
		\end{pmatrix}		
		\begin{pmatrix}
			\bogd[\tau k]a_p\bog[\tau k]\\ \bogd[\tau k]\ad_{-p}\bog[\tau k]
		\end{pmatrix} \, ,
		\end{align} 
		where we used the fact that $k$ is even. Employing the fact that $\bog[0]=\1$ and $\begin{pmatrix}
			0 & k(p) \\
			\overline{k}(p) & 0
		\end{pmatrix}^2=|k(p)|^2\1$, we find that
		\begin{align}
			\begin{pmatrix}\bogd[k]a_p\bog[k]\\ \bogd[ k]\ad_{-p}\bog[k]\end{pmatrix} =&  \exp\Big(\begin{pmatrix}
				0 & k(p) \\
				\overline{k}(p) & 0
			\end{pmatrix}\Big) \begin{pmatrix}a_p\\ \ad_{-p}\end{pmatrix}\\
			=&  \Big[\cosh(|k(p)|)\1+\frac{\sinh(|k(p)|)}{|k(p)|}\begin{pmatrix}
				0 & k(p) \\
				\overline{k}(p) & 0
			\end{pmatrix}\Big]\begin{pmatrix}a_p\\ \ad_{-p}\end{pmatrix} \, .
		\end{align} 
		This finishes the proof.
	\endprf

	\begin{lemma}\label{lem-hfluc}
		Let $\Hfluc(t)$ be defined by \eqref{def-fluc} and \eqref{eq-fluc-dyn}. Recall from \eqref{def-bbf03}--\eqref{def-bbf22} the definitions of $\bbf{j}{j}$. By choosing
		\begin{align}
			S_t \, = \, & \vol \int_0^t \dx{s} \int \dx{p} \Big[E(p)\gamma_s(p) \, + \, \frac{\lambda}{2N}\Gamo_s*(\hat{v}+\hat{v}(0))(p) \Gamo_s(p) \\
			& + \, \frac{\lambda}{2N}\Sigob_s*\hat{v}(p) \Sigo_s(p) \, - \, N\vol^2 \lambda|\phi_s|^4\hat{v}(0) \\
			&  - \, \frac{N\vol\delta(p)\Re\big(\overline{\phi}_si\partial_s \phi_s\big)}2 \, - \,  \frac{\Re\big(\overline{\sigma}_s(p)i\partial_s\sigma_s(p)\big)}{2\big(1+\gamma_s (p)\big)}\Big] \, , 
		\end{align} 
		we then have that 
		\begin{align}
		\Hfluc(t) \, = \, \HBEC(t)  \, + \, \HHFB(t) \, + \, \Hcub(t) \, + \, \Hquart(t) \, ,
		\end{align} 
		 where 
		 \begin{align}
		 	\HBEC(t) \, = & \, \sqrt{N\vol}\Big[u_t(0)\Big( -i\partial_t \phi_t  +  \lambda\vol\hat{v}(0)|\phi_t|^2 \phi_t +  \frac{\lambda}{N}(\sigma_t* \hat{v})(0)\overline{\phi}_t\\
		 	& + \, \frac{\lambda}{N} \big(\gamma_t*(\hat{v}+\hat{v}(0))\big) (0) \phi_t\Big) \, + \, v_t(0)\Big( -\overline{i\partial_t \phi_t} +  \lambda\vol\hat{v}(0) |\phi_t|^2\overline{\phi}_t \\
		 	& +  \frac{\lambda}{N}(\overline{\sigma}_t* \hat{v})(0)\phi_t + \frac{\lambda}{N} \big(\gamma_t*(\hat{v}+\hat{v}(0))\big) (0) \overline{\phi}_t\Big)\Big] e^{i\int_0^t \dx{s} \Omega_s(0)}\ad_0 \, + \, \mathrm{h.c.} \label{def-HBEC}
		 \end{align}  
		 is the BEC-Hamiltonian,
		 \begin{align}
		 \HHFB(t) \, = \, \HHFBd(t) \, + \, \HHFBod(t) 
		 \end{align}
		 is the HFB-Hamiltonian, 
		 \begin{align}
		 \HHFBd(t) \, = & \, \int \dx{p} \Big[ -\Omega_t(p)  -  \frac{\Re\big(\overline{\sigma}_t(p)i\partial_t\sigma_t(p)\big)}{1+\gamma_t(p)}\\
		 & +\,  \Big(E(p)+\frac{\lambda}N\big((\gamma_t+N\vol|\phi_t|^2\delta)*(\hat{v}+\hat{v}(0))\big)(p)\Big)\big(1+2\gamma_t(p)\big) \\
		 & + \, \frac{2\lambda}N\Re\Big(\big((\overline{\sigma}_t+N\vol\overline{\phi}_t^2\delta)*\hat{v}\big)(p)\sigma_t(p)\Big)\Big]\ad_p a_p  \label{def-HHFBd}
		 \end{align}
		 refers to its diagonal part, 
		 \begin{align}
		 \HHFBod(t) \, = \, & \int \dx{p} \Big[-\frac{i\partial_t \sigma_t(p)}2+\frac{\sigma_t(p)i\partial_t \gamma_t(p)}{2(1+\gamma_t(p))}\\
		 & + \,   \Big(E(p)+\frac{\lambda}N\big((\gamma_t+N\vol|\phi_t|^2\delta)*(\hat{v}+\hat{v}(0))\big)(p)\Big)\sigma_t(p)  \\
		 & + \, \frac{\lambda}{2N}\Big( \big((\sigma_t+N\vol\phi_t^2\delta)*\hat{v}\big)(p)(1+\gamma_t(p)) \\
		 & + \, \big((\overline{\sigma}_t+N\vol\overline{\phi}_t^2\delta)*\hat{v}\big)(p)\frac{\sigma_t(p)^2}{1+\gamma_t(p)}\Big) \Big] e^{2i\int_0^t \dx{\tau} \Omega_\tau(p)}\ad_p\ad_{-p} \, + \, \mathrm{h.c.} \label{def-HHFBod}
		 \end{align} 
		 to its off-diagonal part,
		 \begin{align}
		 	\Hcub(t) \, = \, & \frac{\lambda}{\sqrt{N}} \int \dx{\bp_3} \Big(\delta(p_1+p_2+p_3)e^{i\int_0^t \dx{\tau} \big(\Omega_\tau(p_1)+\Omega_\tau(p_2)+\Omega_\tau(p_3)\big)}\\
		 	&\frac1{3!} \bbf{0}{3}_t(\bp_3)\ad_{p_1}\ad_{p_2}\ad_{p_3}\\
		 	& + \,  \delta(p_1+p_2-p_3)e^{i\int_0^t \dx{\tau} \big(\Omega_\tau(p_1)+\Omega_\tau(p_2)-\Omega_\tau(p_3)\big)}\\
		 	& \frac1{2!}\bbf{1}{2}_t(\bp_3)\ad_{p_1}\ad_{p_2}a_{p_3}\\
		 	& + \, h.c.\Big) \label{def-Hcub}
		 \end{align} 
		 describes the cubic processes, and
		 \begin{align}
			 \Hquart(t) \, = \, & \frac{\lambda}{N} \int \dx{\bp_4}  \Big( \delta(p_1+p_2+p_3+p_4)e^{i\int_0^t \dx{\tau} \big(\Omega_\tau(p_1)+\Omega_\tau(p_2)+\Omega_\tau(p_3)+\Omega_\tau(p_4)\big)}\\
			 & \frac1{4!}\bbf{0}{4}_t(\bp_4) \ad_{p_1}\ad_{p_2}\ad_{p_3}\ad_{p_4}\\
			 & + \, \delta(p_1+p_2+p_3-p_4)e^{i\int_0^t \dx{\tau} \big(\Omega_\tau(p_1)+\Omega_\tau(p_2)+\Omega_\tau(p_3)-\Omega_\tau(p_4)\big)}  \\
			& \frac1{3!}\bbf{1}{3}_t(\bp_4)\ad_{p_1}\ad_{p_2}\ad_{p_3}a_{p_4}\\
			& + \, h.c. \\
			& + \,  \delta(p_1+p_2-p_3-p_4)e^{i\int_0^t \dx{\tau} \big(\Omega_\tau(p_1)+\Omega_\tau(p_2)-\Omega_\tau(p_3)-\Omega_\tau(p_4)\big)}  \\
			& \frac1{(2!)^2}\bbf{2}{2}_t(\bp_4)\ad_{p_1}\ad_{p_2}a_{p_3}a_{p_4}\Big) \label{def-Hquart} 
		 \end{align} 
		 describes quartic interactions.
	\end{lemma}
	\prf
		By definition, we have that
		\begin{align}
			\Hfluc(t) =&   -\partial_t S_t \, - \, \int \dx{p} \Omega_t(p) \ad_pa_p \, + \, \bpropd(t) [i\partial_t \bogd[k_t]]\bog[k_t]\bprop(t) \\
			& + \, \bpropd(t) \bogd[k_t][i\partial_t \weyld[\sqrt{N\vol}\phi_t]]\weyl[\sqrt{N\vol}\phi_t]\bog[k_t]\bprop(t) \\
			& + \, \bpropd(t) \bogd[k_t] \weyld[\sqrt{N\vol}\phi_t]\cH_N\weyl[\sqrt{N\vol}\phi_t]\bog[k_t]\bprop(t) \, . \label{eq-hfluc-raw-0}
		\end{align} 
		In ascending order, we define $\partial_t S_t$ till $\Hquart(t)$ to be the normal-ordered zeroth till fourth order polynomials in $\Hfluc(t)$. 
		\par In order to calculate $[i\partial_t \bogd(t)]\bog(t)$, we use the fact that
		\begin{align} 
			\big[\partial_te^{-A(t)}\big]e^{A(t)} =&  \lim_{h\to 0} \frac1h \int_0^1 \dx{\tau} \partial_\tau\Big(e^{-A(t+h)\tau}e^{A(t)\tau}\Big)\\
			=&  \int_0^1 \dx{\tau} e^{-A(t)\tau}[-\partial_tA(t)]e^{A(t)\tau} \, , \label{eq-partialt-uni}
		\end{align} 
		see, e.g., the proof of Proposition 3.2 in \cite{BPS2}. Accordingly, we have that
		\begin{align}
			[\partial_t \bogd[k_t]]\bog[k_t] =&  \frac12\int_0^1 \dx{\tau} \bogd[\tau k_t]\int \dx{p} \big(\partial_t\overline{k_t}(p)a_pa_{-p} \, - \, \partial_tk_t(p)\ad_p\ad_{-p}\big) \bog[\tau k_t]\\
			=&  \frac12\int \dx{p} \int_0^1 \dx{\tau} \Big(\big(\cosh(\tau|k_t(p)|)^2\partial_t\overline{k_t}(p)\\
			&- \, \sinh(\tau|k_t(p)|)^2\frac{\overline{k_t}(p)^2}{|k_t(p)|^2}\partial_tk_t(p)\big)a_pa_{-p} \, - \, h.c.\\
			& + \, \sinh(\tau|k_t(p)|)\cosh(\tau|k_t(p)|)\frac{k_t(p)\partial_t\overline{k_t}(p)-\partial_tk_t(p)\overline{k_t}(p)}{|k_t(p)|}\\
			&\big(\ad_p a_p \, + \, a_p \ad_p\big) \Big)\, .
		\end{align} 
		Using the trigonometric identities for $\sinh$ and $\cosh$, we can simplify the expression to be
		\begin{align}
			[\partial_t \bogd[k_t]]\bog[k_t] =&  \frac14\int \dx{p} \Big[ \Big(\frac{\sinh(2|k_t(p)|)+1}{2|k_t(p)|}\partial_t\overline{k_t}(p)\\
			& - \, \frac{\sinh(2|k_t(p)|)-1}{2|k_t(p)|}\frac{\overline{k_t}(p)^2}{|k_t(p)|^2}\partial_tk_t(p)\Big)a_pa_{-p} \, - \, h.c. \\
			&+ \, 4\frac{\cosh(2|k_t(p)|)-1}{2|k_t(p)|}\frac{i\Im\big(k_t(p)\partial_t\overline{k_t}(p)\big)}{|k_t(p)|}\big(\ad_p a_p \, + \, \frac\vol2\big)\Big] \\
			=&  \frac12 \int \dx{p} \Big[\Big( u_t(p)\overline{v_t}(p) \frac{i\Im\big(k_t(p)\partial_t\overline{k_t}(p)\big)}{|k_t(p)|^2}\\
			& + \, \frac{\overline{k_t}(p)\Re\big(k_t(p)\partial_t\overline{k_t}(p)\big)}{|k_t(p)|^2}\Big)a_pa_{-p} \, - \, h.c. \\
			&+ \, 2 |v_t(p)|^2\frac{i\Im\big(k_t(p)\partial_t\overline{k_t}(p)\big)}{|k_t(p)|^2}\big(\ad_p a_p \, + \, \frac\vol2\big)\Big] \, . \label{eq-bog-intermed-0}
		\end{align} 
		A straightforward calculation yields
		\begin{align}
			\partial_t u_t(p) =&  v_t(p) \frac{\overline{k_t}(p)\Re\big(k_t(p)\partial_t\overline{k_t}(p)\big)}{|k_t(p)|^2} \, , \\
			\partial_t \overline{v_t}(p) =&  u_t(p)\frac{\overline{k_t}(p)\Re\big(k_t(p)\partial_t\overline{k_t}(p)\big)}{|k_t(p)|^2} \, + \, \overline{v_t}(p)\frac{i\Im\big(k_t(p)\partial_t\overline{k_t}(p)\big)}{|k_t(p)|^2} \, . 
		\end{align} 
		In particular, we have that 
		\begin{align}
			\MoveEqLeft \frac{\overline{k_t}(p)\Re\big(k_t(p)\partial_t\overline{k_t}(p)\big)}{|k_t(p)|^2} \, + \, u_t(p)\overline{v_t}(p)\frac{i\Im\big(k_t(p)\partial_t\overline{k_t}(p)\big)}{|k_t(p)|^2}\\
			=&  u_t(p)	\partial_t \overline{v_t}(p) \, - \, \overline{v_t}(p)\partial_t u_t(p) \\
			=&  u_t(p)^2\partial_t\big(\frac{\overline{v_t}(p)}{u_t(p)}\big) \, . \label{eq-bog-aa-coeff-0}
		\end{align} 
		Because $u_t(p)^2=1+|v_t(p)|^2$, we can invoke that 
		\begin{align}
		u_t(p)\partial_t u_t(p) \, = \, \frac12\partial_t u_t(p)^2 \,=\, \frac12\partial_t |v_t(p)|^2 \, = \, \Re\big(v_t(p)\partial_t \overline{v_t}(p)\big) \, ,
		\end{align}
		in order to obtain that
		\begin{align}
			\MoveEqLeft |v_t(p)|^2\frac{i\Im\big(k_t(p)\partial_t\overline{k_t}(p)\big)}{|k_t(p)|^2}\\
			=&  v_t(p)\partial_t \overline{v_t}(p)-u_t(p)\partial_t u_t(p)\\
			=&  i\Im\big(v_t(p)\partial_t \overline{v_t}(p)\big) \\
			=&  v_t(p)^2\partial_t \big(\frac{\overline{v_t}(p)}{v_t(p)}\big) \, . \label{eq-bog-ada-coeff-0}
		\end{align} 
		Using the notation $\gamma_t(p)=|v_t(p)|^2$, $\sigma_t(p)=u_t(p)v_t(p)$, and employing \eqref{eq-bog-aa-coeff-0}, \eqref{eq-bog-ada-coeff-0}, we can hence simplify \eqref{eq-bog-intermed-0} as
		\begin{align}
			[i\partial_t \bogd[k_t]]\bog[k_t] =&  \frac{i}2 \int \dx{p} \Big( (1+\gamma_t(p))\partial_t\big(\frac{\overline{\sigma}_t(p)}{1+\gamma_t(p)}\big)a_pa_{-p} \, - \, h.c. \\
			& + \, 2 \frac{\sigma_t(p)^2}{1+\gamma_t(p)}\partial_t\big(\frac{\overline{\sigma}_t(p)}{\sigma_t(p)}\big)\big(\ad_p a_p \, + \, \frac\vol2\big)\Big) \\
			=&  -\int \dx{p} \Big[\Big(\frac{\overline{i\partial_t \sigma_t}(p)}2+\frac{\overline{\sigma}_t(p)i\partial_t \gamma_t(p)}{2(1+\gamma_t(p))}\Big) a_pa_{-p} \, + \, \mathrm{h.c.} \\
			& - \,  \frac{\Im\big(\sigma_t(p)\partial_t\overline{\sigma}_t(p)\big)}{1+\gamma_t(p)}\Big(\ad_p a_p \, + \, \frac\vol2\Big)\Big] \label{eq-bog-intermed-1}
		\end{align} 
		Observing that $\Im(z)=\Re(i\overline{z})$, and using \eqref{eq-bog-prop}, \eqref{eq-bog-intermed-1} implies
		\begin{align}
			\MoveEqLeft \bpropd(t) [i\partial_t \bogd[k_t]]\bog[k_t]\bprop(t)\\
			 =&  -\int \dx{p} \Big[\Big(\frac{\overline{i\partial_t \sigma_t}(p)}2+\frac{\overline{\sigma}_t(p)i\partial_t \gamma_t(p)}{2(1+\gamma_t(p))}\Big)e^{-2i\int_0^t \dx{s} \Omega_s(p)}a_pa_{-p} \, + \, \mathrm{h.c.} \\
			&  + \,  \frac{\Re\big(\overline{\sigma}_t(p)i\partial_t\sigma_t(p)\big)}{1+\gamma_t(p)}\Big(\ad_p a_p \, + \, \frac\vol2\Big)\Big] \label{eq-bog-dyn-0}
		\end{align}
		Next, using \eqref{eq-partialt-uni} again, we have that 
		\begin{align}\label{eq-har-dyn-0}
            \begin{split}
			& [i\partial_t \weyld[\sqrt{N\vol}\phi_t] ] \weyl[\sqrt{N\vol}\phi_t] \\
			&= \, -\sqrt{N\vol}\int_0^1 \dx{\tau} \weyld[\tau\sqrt{N\vol}\phi_t]\big(i\partial_t \phi_t\ad_0 +\overline{i\partial_t \phi_t} a_0\big)\weyl[\tau\sqrt{N\vol}\phi_t] \\
			&= \, -\sqrt{N\vol}\big(i\partial_t \phi_t\ad_0 +\overline{i\partial_t \phi_t} a_0\big) -\frac{N\vol^2}2\Re\big(\overline{\phi}_ti\partial_t \phi_t\big) \, , 
		\end{split}
        \end{align}
		where we also employed that $\delta(0)=\vol$. 
		\par Finally, we calculate
		\begin{align}
				\MoveEqLeft \weyld[\sqrt{N\vol}\phi_t]\cH_N\weyl[\sqrt{N\vol}\phi_t]\\
				=&  \frac{N\vol^3 \lambda|\phi_t|^4\hat{v}(0)}2 \\
				&+ \, \sqrt{N\vol^3}\lambda|\phi_t|^2\hat{v}(0)\big(\phi_t \ad_0 + \overline{\phi}_ta_0\big) \\
				& + \, \int \dx{p} \big(E(p)+\lambda\vol(\hat{v}(p)+\hat{v}(0))|\phi_t|^2\big)\ad_pa_p \\
				&+ \, \frac{\lambda\vol}2 \int \dx{p} \hat{v}(p) \big(\phi_t^2\ad_p\ad_{-p}+\overline{\phi}_t^2a_pa_{-p}\big)\\
				&+ \, \frac{\lambda\sqrt{\vol}}{\sqrt{N}} \int d\bp_2 \, \hat{v}(p_2)\big(\ad_{p_1}\ad_{p_2}a_{p_{12}} \phi_t \, + \, \ad_{p_{12}}a_{p_2}a_{p_1}\overline{\phi}_t\big) \\
				&+ \, \frac{\lambda}{2N} \int \dx{\bp_4}  \delta(p_1+p_2-p_3-p_4)\hat{v}(p_1-p_3)\ad_{p_1}\ad_{p_2}a_{p_3}a_{p_4} \, . \label{eq-weyl-HN-0}
		\end{align} 	
		In order to calculate 
		\begin{align}
			\bogd[k_t]\weyld[\sqrt{N\vol}\phi_t]\cH_N\weyl[\sqrt{N\vol}\phi_t]\bog[k_t] \, , 
		\end{align} 
		we calculate the transformations for the scalar terms, the terms linear in $a$, $\ad$, etc. in $\weyld[\sqrt{N\vol}\phi_t]\cH_N\weyl[\sqrt{N\vol}\phi_t]$. We have that
		\begin{align}
			\MoveEqLeft \sqrt{N\vol^3}\lambda|\phi_t|^2\hat{v}(0)\bogd[k_t]\big(\phi_t \ad_0 + \overline{\phi}_ta_0\big)\bog[k_t]\\
			=&  \sqrt{N\vol^3}\lambda|\phi_t|^2\hat{v}(0) \Big(\big(\phi_tu_t(0)+\overline{\phi}_tv_t(0)\big)\ad_0 \, + \, \big(\overline{\phi}_tu_t(0)+\phi_t\overline{v_t}(0)\big)a_0\Big) \, . \label{eq-twhwt-1}
		\end{align} 
		For the quadratic terms, we have that
		\begin{align}
			\MoveEqLeft \int \dx{p} \big(E(p)+\lambda\vol(\hat{v}(p)+\hat{v}(0))|\phi_t|^2\big)\bogd[k_t]\ad_pa_p\bog[k_t]\\
			& + \, \frac{\lambda\vol}2 \int \dx{p} \hat{v}(p) \bogd[k_t]\big(\phi_t^2\ad_p\ad_{-p}+\overline{\phi}_t^2a_pa_{-p}\big)\bog[k_t] \\
			=&  \int \dx{p} \Big[\big(E(p)+\lambda\vol(\hat{v}(p)+\hat{v}(0))|\phi_t|^2\big)(u_t(p)^2\ad_pa_p+|v_t(p)|^2a_p\ad_p)\\
			& + \, \frac{\lambda\vol}{2}\hat{v}(p)\big(u_t(p)\overline{v_t}(p)\phi_t^2   + u_t(p)v_t(p)\overline{\phi}_t^2 \big)\big(\ad_pa_p+a_p\ad_p\big)\\
			&+ \, \big(E(p)+\lambda\vol(\hat{v}(p)+\hat{v}(0))|\phi_t|^2\big)u_t(p)\big(v_t(p)\ad_p\ad_{-p} + \overline{v_t}(p)a_pa_{-p}\big)\\
			&+ \, \frac{\lambda\vol}2\hat{v}(p)\Big(\big(u_t(p)^2\phi_t^2+v_t(p)^2\overline{\phi}_t^2\big)\ad_p\ad_{-p}\\
			& + \, \big(\overline{v_t}(p)^2\phi_t^2 + u_t(p)^2\overline{\phi}_t^2\big)a_pa_{-p}\Big)\Big]\\
			=&  \int \dx{p} \Big[\Big(\big(E(p)+\lambda\vol(\hat{v}(p)+\hat{v}(0))|\phi_t|^2\big)\gamma_t(p) \\
			& + \frac{\lambda\vol}2 \hat{v}(p) \big( \overline{\sigma}_t(p)  \phi_t^2 +  \sigma_t(p)\overline{\phi}_t^2\big)\Big)\vol \\
			& + \, \Big( \big(E(p)+\lambda\vol(\hat{v}(p)+\hat{v}(0))|\phi_t|^2\big)\big(1+2\gamma_t(p)\big) \\
			& + \, \lambda\vol \hat{v}(p)\big(\overline{\sigma}_t(p)\phi_t^2+\sigma_t(p)\overline{\phi}_t^2\big)\Big) \ad_pa_p \\
			& + \, \Big( \big(E(p)+\lambda\vol(\hat{v}(p)+\hat{v}(0))|\phi_t|^2\big)\sigma_t(p)  \\
			& + \, \frac{\lambda\vol}2 \hat{v}(p)\big( (1+\gamma_t(p))\phi_t^2 +\frac{\sigma_t(p)^2}{1+\gamma_t(p)} \overline{\phi}_t^2\big)\Big)\ad_p\ad_{-p} \, + \, \mathrm{h.c.}\Big] \, , \label{eq-quad-rot-0}
		\end{align} 
		where the hermitian conjugate 'h.c.' only refers to the terms proportional to $\ad_p\ad_{-p}$. Bogoliubov-rotating the cubic terms yields
		\begin{align}
			\MoveEqLeft \frac{\lambda\sqrt{\vol}}{\sqrt{N}} \int d\bp_2 \, \hat{v}(p_2)\bogd[k_t]\big(\ad_{p_1}\ad_{p_2}a_{p_{12}} \phi_t \, + \, \ad_{p_{12}}a_{p_2}a_{p_1}\overline{\phi}_t\big)\bog[k_t]\\
			=&  \frac{\lambda\sqrt{\vol}}{\sqrt{N}} \int d\bp_2 \, \frac{\hat{v}(p_1)+\hat{v}(p_2)}{2}\Big(\big(u_t(p_1)\ad_{p_1}+\overline{v_t}(p_1)a_{-p_1}\big)\\
			&\big(u_t(p_2)\ad_{p_2}+\overline{v_t}(p_2)a_{-p_2}\big)\big(u_t(p_{12})a_{p_{12}}+v_t(p_{12})\ad_{-p_{12}}\big) \phi_t  \, + \, \mathrm{h.c.}\Big) \, . \label{eq-cub-rot-0}
		\end{align} 
		After normal-ordering, the expressions linear in $a$, $\ad$ in \eqref{eq-cub-rot-0} are given by
		\begin{align}
			\MoveEqLeft \frac{\lambda\sqrt{\vol}}{\sqrt{N}} \int \dx{p}  \Big[\Big(\hat{v}(p)\big(u_t(p)\overline{v_t}(p)v_t(0)\phi_t\, + \, u_t(p)v_t(p)u_t(0)\overline{\phi}_t\big)\\
			& + |v_t(p)|^2\big(\hat{v}(p)+\hat{v}(0)\big)\big(u_t(0)\phi_t+v_t(0)\overline{\phi}_t \big)\Big)\ad_0 \,+ \, h.c.\Big] \\
			=&  \frac{\lambda\sqrt{\vol}}{\sqrt{N}} \int \dx{p} \Big[u_t(0)\Big(\hat{v}(p)\sigma_t(p)\overline{\phi}_t+  \big(\hat{v}(p)+\hat{v}(0)\big) \gamma_t(p) \phi_t\Big) \\
			& + \, v_t(0)\Big(\hat{v}(p)\overline{\sigma}_t(p)\phi_t+ \big(\hat{v}(p)+\hat{v}(0)\big) \gamma_t(p) \overline{\phi}_t\Big) \Big] \ad_0 \, + \,  h.c. \, . \label{eq-cub-rot-NO-lin-0}
		\end{align} 
		Similarly, the cubic terms obtained from normal-ordering \eqref{eq-cub-rot-0} are given by
		\begin{align}
			\MoveEqLeft \frac{\lambda\sqrt{\vol}}{2\sqrt{N}} \int \dx{\bp_3} \big(\hat{v}(p_1)+\hat{v}(p_2)\big)\delta(p_1+p_2-p_3)\\
			&\Big(\big(u_t(p_1)u_t(p_2)u_t(p_3)\phi_t+v_t(p_1)v_t(p_2)\overline{v_t}(p_3)\overline{\phi}_t\big)\ad_{p_1}\ad_{p_2}a_{p_3}\\
			& + \, \big( u_t(p_1)u_t(p_2)v_t(p_3)\phi_t + v_t(p_1)v_t(p_2)u_t(p_3)\overline{\phi}_t\big)\ad_{p_1}\ad_{p_2}\ad_{-p_3}\\
			& + \, 2\big(\overline{v_t}(p_1)u_t(p_2)v_t(p_3)\phi_t+u_t(p_1)v_t(p_2)u_t(p_3)\overline{\phi}_t\big)\ad_{-p_3}\ad_{p_2}a_{-p_1}\\
			& + \, h.c.\Big) \, . \label{eq-cub-rot-NO-cub-0}
		\end{align} 
		Here, we already exploited the symmetry w.r.t. $p_1\leftrightarrow p_2$. Relabeling momenta yields
		\begin{align}
			\MoveEqLeft \frac{\lambda\sqrt{\vol}}{2\sqrt{N}} \int \dx{\bp_3} \Big[\Big(\big(u_t(p_1)u_t(p_2)u_t(p_3)\phi_t+v_t(p_1)v_t(p_2)\overline{v_t}(p_3)\overline{\phi}_t\big)\big(\hat{v}(p_1)+\hat{v}(p_2)\big)\\
			& + 2\big(v_t(p_1)u_t(p_2)\overline{v_t}(p_3)\phi_t+u_t(p_1)v_t(p_2)u_t(p_3)\overline{\phi}_t\big)\big(\hat{v}(p_2)+\hat{v}(p_3)\big)\Big)\\
			& \delta(p_1+p_2-p_3)\ad_{p_1}\ad_{p_2}a_{p_3}\\
			& + \, \big( u_t(p_1)u_t(p_2)v_t(p_3)\phi_t + v_t(p_1)v_t(p_2)u_t(p_3)\overline{\phi}_t\big)\big(\hat{v}(p_1)+\hat{v}(p_2)\big)\\
			& \delta(p_1+p_2+p_3)\ad_{p_1}\ad_{p_2}\ad_{p_3}\\
			& + \, h.c.\Big) \, . 
		\end{align} 
		To conclude the calculation of these cubic terms, we symmetrize the coefficient of $\ad_{p_1}\ad_{p_2}a_{p_3}$ w.r.t. $p_1\leftrightarrow p_2$, and the coefficient of $\ad_{p_1}\ad_{p_2}\ad_{p_3}$ w.r.t. permutations of $(p_1,p_2,p_3)$ to obtain
		\begin{align}
			\MoveEqLeft \frac{\lambda\sqrt{\vol}}{\sqrt{N}} \int \dx{\bp_3} \Big[ \frac1{3!}\delta(p_1+p_2+p_3)\ad_{p_1}\ad_{p_2}\ad_{p_3}\\
			& \Big(\big( u_t(p_1)u_t(p_2)v_t(p_3)\phi_t + v_t(p_1)v_t(p_2)u_t(p_3)\overline{\phi}_t\big)\big(\hat{v}(p_1)+\hat{v}(p_2)\big)\\
			& + \, \big( v_t(p_1)u_t(p_2)u_t(p_3)\phi_t + u_t(p_1)v_t(p_2)v_t(p_3)\overline{\phi}_t\big)\big(\hat{v}(p_2)+\hat{v}(p_3)\big)\\
			& + \, \big( u_t(p_1)v_t(p_2)u_t(p_3)\phi_t + v_t(p_1)u_t(p_2)v_t(p_3)\overline{\phi}_t\big)\big(\hat{v}(p_1)+\hat{v}(p_3)\big)\Big) \\
			& + \, \frac1{2!}\delta(p_1+p_2-p_3)\ad_{p_1}\ad_{p_2}a_{p_3}\\
			& \Big(\big(u_t(p_1)u_t(p_2)u_t(p_3)\phi_t+v_t(p_1)v_t(p_2)\overline{v_t}(p_3)\overline{\phi}_t\big)\big(\hat{v}(p_1)+\hat{v}(p_2)\big)\\
			& + \big(v_t(p_1)u_t(p_2)\overline{v_t}(p_3)\phi_t+u_t(p_1)v_t(p_2)u_t(p_3)\overline{\phi}_t\big)\big(\hat{v}(p_2)+\hat{v}(p_3)\big)\\
			& + \big(u_t(p_1)v_t(p_2)\overline{v_t}(p_3)\phi_t+v_t(p_1)u_t(p_2)u_t(p_3)\overline{\phi}_t\big)\big(\hat{v}(p_1)+\hat{v}(p_3)\big)\Big)\\
			& + \, h.c.\Big) \, . \label{eq-cub-rot-NO-cub-1}
		\end{align} 
		We are left with calculating
		\begin{align}
			\MoveEqLeft \frac{\lambda}{2N} \int \dx{\bp_4}  \delta(p_1+p_2-p_3-p_4)\hat{v}(p_1-p_3)\bogd[k_t]\ad_{p_1}\ad_{p_2}a_{p_3}a_{p_4}\bog[k_t]\\
			=&  \frac{\lambda}{2N} \int \dx{\bp_4}  \delta(p_1+p_2-p_3-p_4)\hat{v}(p_1-p_3)\big(u_t(p_1)\ad_{p_1}+\overline{v_t}(p_1)a_{-p_1}\big)\\
			& \big(u_t(p_2)\ad_{p_2}+\overline{v_t}(p_2)a_{-p_2}\big)\big(u_t(p_3)a_{p_3}+v_t(p_3)\ad_{-p_3}\big)\\
			&\big(u_t(p_4)a_{p_4}+v_t(p_4)\ad_{-p_4}\big) \, . \label{eq-quart-rot-0}
		\end{align} 
		After normal-ordering, the scalar terms come from pairs of annihilation operators $a$ left to a pair of creation operators $\ad$ in \eqref{eq-quart-rot-0}, and they are given by
		\begin{align}
			\MoveEqLeft \frac{\lambda}{2N} \int \dx{\bp_4}  \delta(p_1+p_2-p_3-p_4)\hat{v}(p_1-p_3) \\
			&\Big( \delta(p_1+p_2)\delta(p_3+p_4)\overline{v_t}(p_1)u_t(p_2)u_t(p_3)v_t(p_4)\\
			& + \, \big(\delta(p_1-p_4)\delta(p_2-p_3)+\delta(p_1-p_3)\delta(p_2-p_4)\big)\\
			& \overline{v_t}(p_1)\overline{v_t}(p_2)v_t(p_3)v_t(p_4)\Big) \\
			=&  \frac{\lambda\vol}{2N} \int \dx{p} dq \, \Big(\overline{\sigma}_t(p) \hat{v}(p-q)\sigma_t(q) \, + \, \gamma_t(p)(\hat{v}(p-q)+\hat{v}(0)) \gamma_t(q)\Big) \, . \label{eq-quart-rot-NO-sca-0}
		\end{align} 
		Next, we determine all terms proportional to $aa$ in the normal-ordering of \eqref{eq-quart-rot-0}, that come from having a creation operator $a$ left of an annihilation operator $\ad$, are given by
		\begin{align}
			\MoveEqLeft \frac{\lambda}{2N} \int \dx{\bp_4}  \delta(p_1+p_2-p_3-p_4)\hat{v}(p_1-p_3)\\
			& \Big(\overline{v_t}(p_1)u_t(p_2)u_t(p_3)u_t(p_4)\delta(p_1+p_2)a_{p_3}a_{p_4} \, + \, \overline{v_t}(p_1)\overline{v_t}(p_2)u_t(p_3)v_t(p_4)\\
			&\big(\delta(p_3+p_4)a_{p_1}a_{p_2}+\delta(p_2-p_4)a_{p_1}a_{p_3}+\delta(p_1-p_4)a_{p_2}a_{p_3}\big)\\
			& + \, \overline{v_t}(p_1)\overline{v_t}(p_2)v_t(p_3)u_t(p_4)\big(\delta(p_1-p_3)a_{p_2}a_{p_4}+\delta(p_2-p_3)a_{p_1}a_{p_4}\big)\Big)\\
			=&  \frac{\lambda}{2N} \int \dx{p}dq \, \Big( \overline{\sigma}_t(q) \hat{v}(q-p) (1+\gamma_t(p)) \, + \, \sigma_t(q)\hat{v}(q-p) \frac{\overline{\sigma}_t(p)^2}{1+\gamma_t(p)}\\
			& + \, 2\gamma_t(q)\big(\hat{v}(q-p)+\hat{v}(0)\big)\overline{\sigma}_t(p) \Big) a_pa_{-p} \, . \label{eq-quart-rot-NO-aa-0}
		\end{align} 
		In a similar fashion, we determine all terms involving $\ad a$ in \eqref{eq-quart-rot-0} after normal-ordering to be
		\begin{align}
			\MoveEqLeft \frac{\lambda}{2N} \int \dx{\bp_4}  \delta(p_1+p_2-p_3-p_4)\hat{v}(p_1-p_3)\\
			&\Big(u_t(p_1)\overline{v_t}(p_2)v_t(p_3)u_t(p_4)\delta(p_2-p_3)\ad_{p_1}a_{p_4} \, + \, u_t(p_1)\overline{v_t}(p_2)u_t(p_3)v_t(p_4)\\ 
			& \big(\delta(p_3+p_4)\ad_{p_1}a_{-p_2}+\delta(p_2-p_4)\ad_{p_1}a_{p_3}\big)\, + \, \overline{v_t}(p_1)u_t(p_2)v_t(p_3)u_t(p_4)\\
			&\big(\delta(p_1+p_2)\ad_{-p_3}a_{p_4}+\delta(p_1-p_3)\ad_{p_2}a_{p_4}\big) \, + \, \overline{v_t}(p_1)u_t(p_2)u_t(p_3)v_t(p_4)\\
			& \big(\delta(p_1+p_2)\ad_{-p_4}a_{p_3} +\delta(p_1-p_4)\ad_{p_2}a_{p_3} + \delta(p_3+p_4)\ad_{-p_2}a_{p_1}\big) \, \\
			& + \, \overline{v_t}(p_1)\overline{v_t}(p_2)v_t(p_3)v_t(p_4)\big(\delta(p_1-p_4)\ad_{-p_3}a_{-p_2}+ \delta(p_1-p_3)\ad_{-p_4}a_{-p_2}\\
			&+\delta(p_2-p_3)\ad_{-p_4}a_{-p_1}+\delta(p_2-p_4)\ad_{-p_3}a_{-p_2}\big) \\
			=&  \frac{\lambda}{N} \int \dx{p} dq \, \Big( \gamma_t(q)\big(\hat{v}(q-p)+\hat{v}(0)\big)(1+2\gamma_t(p)) \\
			& + \,  \sigma_t(q)\hat{v}(q-p)\overline{\sigma}_t(p) \, + \,  \overline{\sigma}_t(q)\hat{v}(q-p)\sigma_t(p)\Big) \ad_p a_p \, . \label{eq-quart-rot-NO-ada-0}
		\end{align}
		Here recall that we need to contract each $a$ with each $\ad$ to the right of it. Finally, the quartic, normal-ordered terms coming from \eqref{eq-quart-rot-0} are given by
		\begin{align}
			\MoveEqLeft \frac{\lambda}{4N} \int \dx{\bp_4}  \delta(p_1+p_2-p_3-p_4)\big(\hat{v}(p_1-p_3)+\hat{v}(p_1-p_4)\big)\\
			&\Big( u_t(p_1)u_t(p_2)v_t(p_3)v_t(p_4)\ad_{p_1}\ad_{p_2}\ad_{-p_3}\ad_{-p_4}\\
			& + \, 2u_t(p_1)u_t(p_2)v_t(p_3)u_t(p_4) \ad_{p_1}\ad_{p_2}\ad_{-p_3}a_{p_4}\\
			& + \, 2u_t(p_1)\overline{v_t}(p_2)v_t(p_3)v_t(p_4)\ad_{p_1}\ad_{-p_3}\ad_{-p_4} a_{-p_2}\\
			& + \, h.c. \\
			& + \, \big(u_t(p_1)u_t(p_2)u_t(p_3)u_t(p_4) + v_t(p_1)v_t(p_2)\overline{v_t}(p_3)\overline{v_t}(p_3)\big) \ad_{p_1}\ad_{p_2}a_{p_3}a_{p_4}\\
			& + \, 4u_t(p_1)\overline{v_t}(p_2)v_t(p_3)u_t(p_4)\ad_{p_1}\ad_{-p_3}a_{-p_2}a_{p_4}\Big) \, , 
		\end{align} 
		where the hermitian conjugate refers to the preceding terms. Analogously to \eqref{eq-cub-rot-NO-cub-0}, we already symmetrized the expression w.r.t. $p_1\leftrightarrow p_2$ and $p_3\leftrightarrow p_4$. Relabeling momenta again, we have
		\begin{align}
			\MoveEqLeft \frac{\lambda}{4N} \int \dx{\bp_4}  \Big[ \delta(p_1+p_2+p_3+p_4)\ad_{p_1}\ad_{p_2}\ad_{p_3}\ad_{p_4}\\
			& u_t(p_1)u_t(p_2)v_t(p_3)v_t(p_4)\big(\hat{v}(p_1+p_3)+\hat{v}(p_2+p_3)\big)\\
			& + \, 2\delta(p_1+p_2+p_3-p_4) \ad_{p_1}\ad_{p_2}\ad_{p_3}a_{p_4}\\
			& \Big(u_t(p_1)u_t(p_2)v_t(p_3)u_t(p_4)\big(\hat{v}(p_1+p_3)+\hat{v}(p_2+p_3)\big) \\
			& + \, u_t(p_1)v_t(p_2)v_t(p_3)\overline{v_t}(p_4)\big(\hat{v}(p_1+p_2)+\hat{v}(p_1+p_3)\big) \Big) \\
			& + \, h.c. \\
			& + \,  \delta(p_1+p_2-p_3-p_4)\ad_{p_1}\ad_{p_2}a_{p_3}a_{p_4}\\
			& \Big(\big(u_t(p_1)u_t(p_2)u_t(p_3)u_t(p_4) \, + \, v_t(p_1)v_t(p_2)\overline{v_t}(p_3)\overline{v_t}(p_3)\big)\big(\hat{v}(p_1-p_3)+\hat{v}(p_2-p_3)\big)\\
			& + \, 4u_t(p_1)v_t(p_2)\overline{v_t}(p_3)u_t(p_4)\big(\hat{v}(p_1+p_2)+\hat{v}(p_2-p_3)\big)\Big) \, . \label{eq-quart-rot-NO-quart-0}
		\end{align}
		Let 
		\begin{align}
			a^{(\sigma)}(\bq_\ell) \, := \, \prod_{j=1}^\ell a^{(\sigma)}(q_j) \, , \label{eq-a-vector-0}
		\end{align} 
		see also \eqref{eq-a-set-0}. Symmetrizing the coefficients of $\ad(\bp_n)a(\bk_m)$ in \eqref{eq-quart-rot-NO-quart-0} w.r.t. $\bp_n$ and $\bk_m$, we find that
		\begin{align}
			\MoveEqLeft \frac{\lambda}{N} \int \dx{\bp_4}  \Big[ \frac1{4!}\delta(p_1+p_2+p_3+p_4)\ad_{p_1}\ad_{p_2}\ad_{p_3}\ad_{p_4}\\
			& \Big(\big(u_t(p_1)u_t(p_2)v_t(p_3)v_t(p_4)  +  v_t(p_1)v_t(p_2)u_t(p_3)u_t(p_4)\big)\big(\hat{v}(p_1+p_3)+\hat{v}(p_2+p_3)\big)\\
			& + \, \big(u_t(p_1)v_t(p_2)u_t(p_3)v_t(p_4)+ v_t(p_1)u_t(p_2)v_t(p_3)u_t(p_4)\big)\big(\hat{v}(p_1+p_2)+\hat{v}(p_2+p_3)\big)\\
			& + \, \big(u_t(p_1)v_t(p_2)v_t(p_3)u_t(p_4)+ v_t(p_1)u_t(p_2)u_t(p_3)v_t(p_4)\big)\big(\hat{v}(p_1+p_2)+\hat{v}(p_1+p_3)\big)\Big)\\
			& + \, \frac1{3!}\delta(p_1+p_2+p_3-p_4) \ad_{p_1}\ad_{p_2}\ad_{p_3}a_{p_4}\\
			& \Big(\big(u_t(p_1)u_t(p_2)v_t(p_3)u_t(p_4)+ v_t(p_1)v_t(p_2)u_t(p_3)\overline{v_t}(p_4)\big)\big(\hat{v}(p_1+p_3)+\hat{v}(p_2+p_3)\big) \\
			& + \big(u_t(p_1)v_t(p_2)u_t(p_3)u_t(p_4)+ v_t(p_1)u_t(p_2)v_t(p_3)\overline{v_t}(p_4)\big)\big(\hat{v}(p_1+p_2)+\hat{v}(p_2+p_3)\big) \\
			& + \, \big(v_t(p_1)u_t(p_2)u_t(p_3)u_t(p_4)+ u_t(p_1)v_t(p_2)v_t(p_3)\overline{v_t}(p_4)\big)\big(\hat{v}(p_1+p_2)+\hat{v}(p_1+p_3)\big) \\
			& + \, h.c. \\
			& + \,  \frac1{(2!)^2}\delta(p_1+p_2-p_3-p_4)\ad_{p_1}\ad_{p_2}a_{p_3}a_{p_4} \label{eq-quart-rot-NO-quart-1}\\
			& \Big(\big(u_t(p_1)u_t(p_2)u_t(p_3)u_t(p_4) \, + \, v_t(p_1)v_t(p_2)\overline{v_t}(p_3)\overline{v_t}(p_3)\big)\big(\hat{v}(p_1-p_3)+\hat{v}(p_2-p_3)\big)\\
			& + \, \big(u_t(p_1)v_t(p_2)\overline{v_t}(p_3)u_t(p_4)+v_t(p_1)u_t(p_2)u_t(p_3)\overline{v_t}(p_4)\big)\big(\hat{v}(p_1+p_2)+\hat{v}(p_2-p_3)\big)\\
			& + \, \big(v_t(p_1)u_t(p_2)\overline{v_t}(p_3)u_t(p_4)+u_t(p_1)v_t(p_2)u_t(p_3)\overline{v_t}(p_4)\big)\big(\hat{v}(p_1+p_2)+\hat{v}(p_1-p_3)\big)\Big) \, . 
		\end{align} 
		\par With these calculations, we are ready to identify each of the terms in $\Hfluc$, in ascending order of involved number of annihilation and creation operators. Collecting \eqref{eq-hfluc-raw-0}, \eqref{eq-bog-dyn-0}, \eqref{eq-har-dyn-0}, \eqref{eq-weyl-HN-0}, \eqref{eq-quad-rot-0}, and \eqref{eq-quart-rot-NO-sca-0}, we obtain that we need to choose
		\begin{align}
			\MoveEqLeft \partial_t S_t \, = \, \\
			&  \vol\Big[\frac{N\vol \lambda|\phi_t|^4\hat{v}(0)}2 \, - \, \frac{N\vol}2\Re\big(\overline{\phi}_ti\partial_t \phi_t\big) \, - \, \frac12 \int \dx{p} \frac{\Re\big(\overline{\sigma}_t(p)i\partial_t\sigma_t(p)\big)}{1+\gamma_t(p)}\\
			& + \,  \int \dx{p}  \Big(\frac{\lambda}{2N}\big(\overline{\sigma}_t* \hat{v}(p)\sigma_t(p) \, + \, \gamma_t*(\hat{v}+\hat{v}(0))(p) \gamma_t(p)\big) \\
			& \big(E(p)+\lambda\vol(\hat{v}(p)+\hat{v}(0))|\phi_t|^2\big)\gamma_t(p) \, + \, \frac{\lambda\vol}2 \hat{v}(p) \big( \overline{\sigma}_t(p)  \phi_t^2 +  \sigma_t(p)\overline{\phi}_t^2\big)\Big) 
		\end{align} 
		in order to eliminate the scalar contributions in $\Hfluc(t)$.
		\par Using the shifted expectations $\Gamo$ and $\Sigo$, see \eqref{def-Gamo}, \eqref{def-Sigo}, we can rewrite $\partial_t S_t$ as
		\begin{align}
			\MoveEqLeft \partial_t S_t \, = \, \\
			& \vol\int \dx{p} \Big[E(p)\gamma_t(p) \, + \, \frac{\lambda}{2N}\Gamo_t*(\hat{v}+\hat{v}(0))(p) \Gamo_t(p) \\
			& + \, \frac{\lambda}{2N}\Sigob_t*\hat{v}(p) \Sigo_t(p) \, - \, N\vol \lambda|\phi_t|^4\hat{v}(0) \\
			&  - \, \frac{N\vol\delta(p)\Re\big(\overline{\phi}_ti\partial_t \phi_t\big)}2 \, - \,  \frac{\Re\big(\overline{\sigma}_t(p)i\partial_t\sigma_t(p)\big)}{2\big(1+\gamma_t(p)\big)}\Big] \, . 
		\end{align} 
		Similarly, collecting \eqref{eq-bog-prop}, \eqref{eq-har-dyn-0}, \eqref{eq-twhwt-1}, \eqref{eq-cub-rot-NO-lin-0}, we have, after conjugating with $\bprop(t)$, that 
		\begin{align}
		\MoveEqLeft \HBEC(t) \, = \, \\
		& \sqrt{N\vol}\Big[-\big(i\partial_t \phi_tu_t(0)+\overline{i\partial_t \phi_t}v_t(0)\big) \\
		& + \, \lambda\vol|\phi_t|^2\hat{v}(0) \big(\phi_tu_t(0)+\overline{\phi}_tv_t(0)\big)\\
		&+ \, \frac{\lambda}{N}  u_t(0)\Big(\int \dx{p}\hat{v}(p)\sigma_t(p)\overline{\phi}_t+  \int \dx{p}\big(\hat{v}(p)+\hat{v}(0)\big) \gamma_t(p) \phi_t\Big) \\
		& + \, v_t(0)\Big(\int \dx{p}\hat{v}(p)\overline{\sigma}_t(p)\phi_t+ \int \dx{p}\big(\hat{v}(p)+\hat{v}(0)\big) \gamma_t(p) \overline{\phi}_t\Big)   \Big]e^{i\int_0^t \dx{s} \Omega_\tau(0)}\ad_0 \\
		&+ \, h.c. \\
		=&  \sqrt{N\vol}\Big[u_t(0)\Big( -i\partial_t \phi_t  +  \lambda\vol\hat{v}(0)|\phi_t|^2 \phi_t +  \frac{\lambda}{N}(\sigma_t* \hat{v})(0)\overline{\phi}_t\\
		& + \, \frac{\lambda}{N} \big(\gamma_t*(\hat{v}+\hat{v}(0)\big) (0) \phi_t\Big) \, + \, v_t(0)\Big( -\overline{i\partial_t \phi_t} +  \lambda\vol\hat{v}(0) |\phi_t|^2\overline{\phi}_t \\
		& +  \frac{\lambda}{N}(\overline{\sigma}_t* \hat{v})(0)\phi_t + \frac{\lambda}{N} \big(\gamma_t*(\hat{v}+\hat{v}(0)\big) (0) \overline{\phi}_t\Big)\Big] e^{i\int_0^t \dx{\tau} \Omega_\tau(0)}\ad_0 \, + \, \mathrm{h.c.}
		\end{align} 
		In order to give an expression for $\HHFB(t)$, we decompose it into 
		\begin{align}
			\HHFB(t) \, = \, \HHFBd(t) \, + \, \HHFBod(t) \, ,
		\end{align} 
		where $\HHFBd(t)$ refers to the diagonal part, involving terms proportional to $\ad a$, and $\HHFBod(t)$ to the off-diagonal part, involving terms propotional to $aa$ and $\ad \ad$. Then \eqref{eq-hfluc-raw-0}, \eqref{eq-bog-dyn-0}, \eqref{eq-quad-rot-0}, \eqref{eq-quart-rot-NO-ada-0}, after conjugation with $\bprop(t)$ using \eqref{eq-bog-prop}, imply
		\begin{align}
			\MoveEqLeft \HHFBd(t) \, = \, \\
			& \int \dx{p} \Big[ -\Omega_t(p)  -  \frac{\Re\big(\overline{\sigma}_t(p)i\partial_t\sigma_t(p)\big)}{1+\gamma_t(p)}\\
			& +\,  \Big(E(p)+\frac{\lambda}N\big((\gamma_t+N\vol|\phi_t|^2\delta)*(\hat{v}+\hat{v}(0))\big)(p)\Big)\big(1+2\gamma_t(p)\big) \\
			& + \, \frac{2\lambda}N\Re\Big(\big((\overline{\sigma}_t+N\vol\overline{\phi}_t^2\delta)*\hat{v}\big)(p)\sigma_t(p)\Big)\Big]\ad_p a_p \, .
		\end{align}  
		Similarly, we obtain, using, in addition \eqref{eq-quart-rot-NO-aa-0}, that
		\begin{align}
			\MoveEqLeft \HHFBod(t) \, = \, \\
			& \int \dx{p} \Big[-\frac{i\partial_t \sigma_t(p)}2+\frac{\sigma_t(p)i\partial_t \gamma_t(p)}{2(1+\gamma_t(p))}\\
			& + \,   \Big(E(p)+\frac{\lambda}N\big((\gamma_t+N\vol|\phi_t|^2\delta)*(\hat{v}+\hat{v}(0))\big)(p)\Big)\sigma_t(p)  \\
			& + \, \frac{\lambda}{2N}\Big( \big((\sigma_t+N\vol\phi_t^2\delta)*\hat{v}\big)(p)(1+\gamma_t(p)) \\
			& + \, \big((\overline{\sigma}_t+N\vol\overline{\phi}_t^2\delta)*\hat{v}\big)(p)\frac{\sigma_t(p)^2}{1+\gamma_t(p)}\Big) \Big] e^{2i\int_0^t \dx{\tau} \Omega_\tau(p)}\ad_p\ad_{-p}\\
			& + \, h.c.
		\end{align} 
		Conjugating \eqref{eq-cub-rot-NO-cub-1} with $\bprop(t)$ and using the notation \eqref{def-bbf03}--\eqref{def-bbf12}, \eqref{eq-bog-prop} implies 
		\begin{align}
			\MoveEqLeft \Hcub(t) \, = \\
			& \frac{\lambda}{\sqrt{N}} \int \dx{\bp_3} \Big[\delta(p_1+p_2+p_3)e^{i\int_0^t \dx{\tau} \big(\Omega_\tau(p_1)+\Omega_\tau(p_2)+\Omega_\tau(p_3)\big)}\\
			&\frac1{3!} \bbf{0}{3}_t(\bp_3)\ad_{p_1}\ad_{p_2}\ad_{p_3}\\
			& + \,  \delta(p_1+p_2-p_3)e^{i\int_0^t \dx{\tau} \big(\Omega_\tau(p_1)+\Omega_\tau(p_2)-\Omega_\tau(p_3)\big)}\\
			& \frac1{2!}\bbf{1}{2}_t(\bp_3)\ad_{p_1}\ad_{p_2}a_{p_3}\\
			& + \, h.c.\Big) \, .
		\end{align} 
		Similarly, conjugating \eqref{eq-quart-rot-NO-quart-1} with $\bprop(t)$ using \eqref{eq-bog-prop} and the notation \eqref{def-bbf04}--\eqref{def-bbf22}, we obtain that
		\begin{align}
			\MoveEqLeft \Hquart(t) \, =\\
			& \frac{\lambda}{N} \int \dx{\bp_4}  \Big[ \delta(p_1+p_2+p_3+p_4)e^{i\int_0^t \dx{\tau} \big(\Omega_\tau(p_1)+\Omega_\tau(p_2)+\Omega_\tau(p_3)+\Omega_\tau(p_4)\big)}\\
			& \frac1{4!}\bbf{0}{4}_t(\bp_4) \ad_{p_1}\ad_{p_2}\ad_{p_3}\ad_{p_4}\\
			& + \, \delta(p_1+p_2+p_3-p_4)e^{i\int_0^t \dx{\tau} \big(\Omega_\tau(p_1)+\Omega_\tau(p_2)+\Omega_\tau(p_3)-\Omega_\tau(p_4)\big)}  \\
			& \frac1{3!}\bbf{1}{3}_t(\bp_4)\ad_{p_1}\ad_{p_2}\ad_{p_3}a_{p_4}\\
			& + \, h.c. \\
			& + \,  \delta(p_1+p_2-p_3-p_4)e^{i\int_0^t \dx{\tau} \big(\Omega_\tau(p_1)+\Omega_\tau(p_2)-\Omega_\tau(p_3)-\Omega_\tau(p_4)\big)}  \\
			& \frac1{(2!)^2}\bbf{2}{2}_t(\bp_4)\ad_{p_1}\ad_{p_2}a_{p_3}a_{p_4} \, . 
		\end{align} 
		This finishes the proof.
	\endprf

    \begin{lemma}[Contracted vertices]\label{lem-con-vert}
        Let $\Hcub$ and $\Hquart$ be defined as in Lemma \ref{lem-hfluc}. Then we have 
        \begin{align}
            \begin{split}
                \wick{\settowidth{\wdth}{$\Hcub$}\hspace{.33\wdth}\c1{\vphantom{\Hcub}}\hspace{.33\wdth}\c1{\vphantom{\Hcub}}\hspace{-.66\wdth}\Hcub(t)} \, &= \, \frac{\lambda\sqrt{\vol}}{\sqrt{N}} e^{i\int_0^t\dx{\tau}\Omega_\tau(0)}\ad_0 \int \dx{k} \Big[u_t(0)\Big(\big(1+2\gamma_t(k)\big)f_0(k)\big(\hat{v}(k)+\hat{v}(0)\big)\phi_t\\
                & \qquad + 2f_0(k)\sigma_t(k)\hat{v}(k)\overline{\phi}_t\Big) \, + \, v_t(0)\Big(\big(1+2\gamma_t(k)\big)f_0(k)\big(\hat{v}(k)+\hat{v}(0)\big)\overline{\phi}_t \\
                & \qquad + 2f_0(k)\overline{\sigma}_t(k)\hat{v}(k)\phi_t\Big)\Big] \, + \, \mathrm{h.c.} \, ,
            \end{split}
        \end{align}
        and 
        \begin{align}
             \wick{\settowidth{\wdth}{$\Hquart$}\hspace{.33\wdth}\c1{\vphantom{\Hquart}}\hspace{.33\wdth}\c1{\vphantom{\Hquart}}\hspace{-.66\wdth}\Hquart(t)} \, =& \, \frac{\lambda}{N} \int \dx{p}\Big[\Big(\big(\fplus \sigma_t\big)*\hat{v}(p)(1+\gamma_t(p)) \, + \, \frac{\big(\fplus \overline{\sigma}_t\big)*\hat{v}(p)\sigma_t(p)}{1+\gamma_t(p)}\\
             & \quad + \, \big((1+2\gamma_t)\fplus\big)*\big(\hat{v}+\hat{v}(0)\big)(p)\sigma_t(p)\Big)e^{2i\int_0^t \dx{\tau}\Omega_\tau(p)}\ad_{p}\ad_{-p}\,+ \, \mathrm{h.c.}\\
             & \quad + \, \Big(\big((1+2\gamma_t)\fplus)*\big(\hat{v}+\hat{v}(0)\big)(p)(1+2\gamma_t(p)) \\
             & \quad + \, 4\Re\big(\big(\fplus\sigma_t\big)*\hat{v}(p)\overline{\sigma}_t(p)\big)\Big) \ad_pa_p\Big] \, .
        \end{align}
    \end{lemma}

    \begin{proof}
        Using symmetry of $\bbf{1}{2}(p_1,p_2,p_3)$ in $(p_1\leftrightarrow p_2)$, see Lemma \ref{lem-hfluc}, we obtain that
        \begin{align}
	       \wick{\settowidth{\wdth}{$\Hcub$}\hspace{.33\wdth}\c1{\vphantom{\Hcub}}\hspace{.33\wdth}\c1{\vphantom{\Hcub}}\hspace{-.66\wdth}\Hcub(t)} \, &= \, \frac{\lambda}{\sqrt{N}} \int \dx{k} \bbf{1}{2}_t(k,0,k) f_0(k)e^{i\int_0^t\dx{\tau}\Omega_\tau(0)}\ad_0 \, + \, \mathrm{h.c.} \, ,
        \end{align}
        where 
        \begin{align}
            \MoveEqLeft \bbf{1}{2}_t(\bp_3) \, =\\
			&\sqrt{\vol}\Big(\big(u_t(p_1)u_t(p_2)u_t(p_3)\phi_t+v_t(p_1)v_t(p_2)\overline{v_t}(p_3)\overline{\phi}_t\big)\big(\hat{v}(p_1)+\hat{v}(p_2)\big)\\
			& + \big(v_t(p_1)u_t(p_2)\overline{v_t}(p_3)\phi_t+u_t(p_1)v_t(p_2)u_t(p_3)\overline{\phi}_t\big)\big(\hat{v}(p_2)+\hat{v}(p_3)\big)\\
			& + \big(u_t(p_1)v_t(p_2)\overline{v_t}(p_3)\phi_t+v_t(p_1)u_t(p_2)u_t(p_3)\overline{\phi}_t\big)\big(\hat{v}(p_1)+\hat{v}(p_3)\big) \Big) \, ,
        \end{align}
        Using the facts that $u_t=\sqrt{1+\gamma_t}$ and that $v_t=\frac{\sigma_t}{\sqrt{1+\gamma_t}}$, we thus find that
        \begin{align}
            \wick{\settowidth{\wdth}{$\Hcub$}\hspace{.33\wdth}\c1{\vphantom{\Hcub}}\hspace{.33\wdth}\c1{\vphantom{\Hcub}}\hspace{-.66\wdth}\Hcub(t)} \, & = \, \frac{\lambda\sqrt{\vol}}{\sqrt{N}} \int \dx{k} f_0(k)e^{i\int_0^t\dx{\tau}\Omega_\tau(0)}\ad_0 \\
            & \qquad \Big(\big(u_t(0)(1+\gamma_t(k))\phi_t+v_t(0)\gamma_t(k)\overline{\phi}_t\big)\big(\hat{v}(k)+\hat{v}(0)\big)\\
            & \qquad + \, \big(u_t(0)\gamma_t(k)\phi_t+v_t(0)(1+\gamma_t(k))\overline{\phi}_t\big)\big(\hat{v}(k)+\hat{v}(0)\big)\\
            & \qquad + \, \big(v_t(0)\overline{\sigma}_t(k)\phi_t+u_t(0)\sigma_t(k)\overline{\phi}_t\big)2\hat{v}(k)\Big) \, + \, \mathrm{h.c.} \\
            & = \, \frac{\lambda\sqrt{\vol}}{\sqrt{N}} \int \dx{k} f_0(k)e^{i\int_0^t\dx{\tau}\Omega_\tau(0)}\ad_0 \\
            & \qquad \Big(\big(\hat{v}(k)+\hat{v}(0)\big)\big(1+2\gamma_t(k)\big)\big(u_t(0)\phi_t+v_t(0)\overline{\phi}_t\big) \\
            & \qquad + \, 2\hat{v}(k)\big(v_t(0)\overline{\sigma}_t(k)\phi_t+u_t(0)\sigma_t(k)\overline{\phi}_t\big)\Big) \, + \, \mathrm{h.c.} \, ,
    	\end{align}
        where we also used the evenness of $\sigma_t$ and $\gamma_t$. Sorting the terms by the coefficients $u_t(0)$ and $v_t(0)$, we obtain
        \begin{align}
                \wick{\settowidth{\wdth}{$\Hcub$}\hspace{.33\wdth}\c1{\vphantom{\Hcub}}\hspace{.33\wdth}\c1{\vphantom{\Hcub}}\hspace{-.66\wdth}\Hcub(t)} \, &= \, \frac{\lambda\sqrt{\vol}}{\sqrt{N}} e^{i\int_0^t\dx{\tau}\Omega_\tau(0)}\ad_0 \int \dx{k} \Big[u_t(0)\Big(\big(1+2\gamma_t(k)\big)f_0(k)\big(\hat{v}(k)+\hat{v}(0)\big)\phi_t\\
                & \qquad + 2f_0(k)\sigma_t(k)\hat{v}(k)\overline{\phi}_t\Big) \, + \, v_t(0)\Big(\big(1+2\gamma_t(k)\big)f_0(k)\big(\hat{v}(k)+\hat{v}(0)\big)\overline{\phi}_t \\
                & \qquad + 2f_0(k)\overline{\sigma}_t(k)\hat{v}(k)\phi_t\Big)\Big] \, + \, \mathrm{h.c.} \, .
        \end{align}
        Similarly, the symmetries of $\bbf{1}{3}_t(\bp_4)$ in $(p_1\leftrightarrow p_2\leftrightarrow p_3)$, and of $\bbf{2}{2}_t(\bp_4)$ in $(p_1\leftrightarrow p_2)$ and $(p_3\leftrightarrow p_4)$, see Lemma \ref{lem-hfluc}, imply
        \begin{align}
        \wick{\settowidth{\wdth}{$\Hquart$}\hspace{.33\wdth}\c1{\vphantom{\Hquart}}\hspace{.33\wdth}\c1{\vphantom{\Hquart}}\hspace{-.66\wdth}\Hquart(t)} \, = \, & \frac{\lambda}{N} \int \dx{p}\dx{k} f_0(k)\Big(\frac{\bbf{1}{3}_t(p,-p,k,k)}{2}e^{2i\int_0^t \dx{\tau}\Omega_\tau(p)}\ad_{p}\ad_{-p} \, + \, \mathrm{h.c.} \\
		& \quad  + \, \bbf{2}{2}_t(p,k,k,p)\ad_pa_p \Big) \, ,
        \end{align}
        where
        \begin{align}
            \MoveEqLeft \bbf{1}{3}_t(\bp_4) \, := \\
			&\big(u_t(p_1)u_t(p_2)v_t(p_3)u_t(p_4)+ v_t(p_1)v_t(p_2)u_t(p_3)\overline{v_t}(p_4)\big)\big(\hat{v}(p_1+p_3)+\hat{v}(p_2+p_3)\big) \\
			& + \big(u_t(p_1)v_t(p_2)u_t(p_3)u_t(p_4)+ v_t(p_1)u_t(p_2)v_t(p_3)\overline{v_t}(p_4)\big)\big(\hat{v}(p_1+p_2)+\hat{v}(p_2+p_3)\big) \\
			& + \, \big(v_t(p_1)u_t(p_2)u_t(p_3)u_t(p_4)+ u_t(p_1)v_t(p_2)v_t(p_3)\overline{v_t}(p_4)\big)\big(\hat{v}(p_1+p_2)+\hat{v}(p_1+p_3)\big) \, ,  \\
			\MoveEqLeft \bbf{2}{2}_t(\bp_4) \, :=\\
			& \big(u_t(p_1)u_t(p_2)u_t(p_3)u_t(p_4) \, + \, v_t(p_1)v_t(p_2)\overline{v_t}(p_3)\overline{v_t}(p_3)\big)\big(\hat{v}(p_1-p_3)+\hat{v}(p_2-p_3)\big)\\
			& + \, \big(u_t(p_1)v_t(p_2)\overline{v_t}(p_3)u_t(p_4)+v_t(p_1)u_t(p_2)u_t(p_3)\overline{v_t}(p_4)\big)\big(\hat{v}(p_1+p_2)+\hat{v}(p_2-p_3)\big)\\
			& + \, \big(v_t(p_1)u_t(p_2)\overline{v_t}(p_3)u_t(p_4)+u_t(p_1)v_t(p_2)u_t(p_3)\overline{v_t}(p_4)\big)\big(\hat{v}(p_1+p_2)+\hat{v}(p_1-p_3)\big) \, .
        \end{align}
        Again using the relations between $u_t$, $v_t$ and $\gamma_t$, $\sigma_t$, we obtain
        \begin{align}
        \MoveEqLeft \wick{\settowidth{\wdth}{$\Hquart$}\hspace{.33\wdth}\c1{\vphantom{\Hquart}}\hspace{.33\wdth}\c1{\vphantom{\Hquart}}\hspace{-.66\wdth}\Hquart(t)} \,= \\
        & \frac{\lambda}{N} \int \dx{p}\dx{k} f_0(k)\\
        & \quad \Big[\Big(\big((1+\gamma_t(p))\sigma_t(k) \, + \, \frac{\sigma_t(p)\overline{\sigma}_t(k)}{1+\gamma_t(p)}\big) 2\hat{v}(p-k) + \, 2\big(\sigma_t(p)(1+\gamma_t(k)) \\
        & \qquad + \sigma_t(p)\gamma_t(k)\big)\big(\hat{v}(p-k)+\hat{v}(0)\big)\Big) \frac12 e^{2i\int_0^t \dx{\tau}\Omega_\tau(p)}\ad_{p}\ad_{-p} \, + \, \mathrm{h.c.} \\
        & \quad + \, \Big(\big((1+\gamma_t(p))(1+\gamma_t(k)) \,+ \, \gamma_t(p)\gamma_t(k)\big)\big(\hat{v}(p-k)+\hat{v}(0)\big)\\
        & \qquad + \,\big((1+\gamma_t(p))\gamma_t(k) \,+ \, \gamma_t(p)(1+\gamma_t(k))\big)\big(\hat{v}(p-k)+\hat{v}(0)\big)\\ 
        & \qquad + \, (\sigma_t(p)\overline{\sigma}_t(k) \, + \, \overline{\sigma}_t(p)\sigma_t(k))2\hat{v}(p-k)\Big)\ad_pa_p\Big] \\
        = & \, \frac{\lambda}{N} \int \dx{p}\dx{k} f_0(k) \Big[\Big(\big((1+\gamma_t(p))\sigma_t(k) \, + \, \frac{\sigma_t(p)\overline{\sigma}_t(k)}{1+\gamma_t(p)}\big) \hat{v}(p-k) \\
        & \qquad + \, \big(\sigma_t(p)(1+\gamma_t(k)) \, + \, \overline{\sigma}_t(p)\gamma_t(k)\big)\hat{v}(p-k)\Big) e^{2i\int_0^t \dx{\tau}\Omega_\tau(p)}\ad_{p}\ad_{-p} \,+ \, \mathrm{h.c.} \\ 
        &\quad + \, \Big((1+2\gamma_t(p))(1+2\gamma_t(k))\big(\hat{v}(p-k)+\hat{v}(0)\big) \\
        &\qquad + \, 4\Re\big(\sigma_t(p)\overline{\sigma}_t(k)\big)\hat{v}(p-k)\Big) \ad_pa_p\Big] \, .
        \end{align}
        Abbreviating $\fplus=\frac12(f_0(p)+f_0(-p))$, see \eqref{def-fplus} and using evenness of $\sigma$, $\gamma$, $\hat{v}$, we can simplify the expression as
        \begin{align}
             \wick{\settowidth{\wdth}{$\Hquart$}\hspace{.33\wdth}\c1{\vphantom{\Hquart}}\hspace{.33\wdth}\c1{\vphantom{\Hquart}}\hspace{-.66\wdth}\Hquart(t)} \, =& \, \frac{\lambda}{N} \int \dx{p}\Big[\Big(\big(\fplus \sigma_t\big)*\hat{v}(p)(1+\gamma_t(p)) \, + \, \frac{\big(\fplus \overline{\sigma}_t\big)*\hat{v}(p)\sigma_t(p)}{1+\gamma_t(p)}\\
             & \quad + \, \big((1+2\gamma_t)\fplus\big)*\big(\hat{v}+\hat{v}(0)\big)(p)\sigma_t(p)\Big)e^{2i\int_0^t \dx{\tau}\Omega_\tau(p)}\ad_{p}\ad_{-p}\,+ \, \mathrm{h.c.}\\
             & \quad + \, \Big(\big((1+2\gamma_t)\fplus)*\big(\hat{v}+\hat{v}(0)\big)(p)(1+2\gamma_t(p)) \\
             & \quad + \, 4\Re\big(\big(\fplus\sigma_t\big)*\hat{v}(p)\overline{\sigma}_t(p)\big)\Big) \ad_pa_p\Big] \, .
        \end{align}
        This concludes the proof.
    \end{proof}

	\begin{lemma}[HFB equations] \label{lem-HFB-equations}
		Let
		\begin{align}
		\Gamo_t(p) :=&  \gamma_t(p) \, + \, N\vol|\phi_t|^2 \delta(p) \, , \\
		\Sigo_t(p) :=&  \sigma_t(p) \, + \, N\vol \phi_t^2 \delta(p) \, .
		\end{align} 
		Then
		\begin{align}
		\MoveEqLeft \frac{i\partial_t \sigma_t(p)}2 \, - \, \frac{\sigma_t(p)i\partial_t \gamma_t(p)}{2(1+\gamma_t(p))} \, = \,\\
		&  \Big(E(p)+\frac{\lambda}{N}\Gamo_t*\big(\hat{v}+\hat{v}(0)\big)(p)\Big)\sigma_t(p)  \\
		& + \, \frac{\lambda}{2N} \Big( \big(\Sigo_t* \hat{v}\big)(p) (1+\gamma_t(p)) \, + \, \big(\Sigob_t*\hat{v}\big)(p) \frac{\sigma_t(p)^2}{1+\gamma_t(p)}\Big) \label{eq-HFB-raw-0}
		\end{align} 
		is equivalent to
		\begin{align}
			i\partial_t \gamma_t =&  \frac{\lambda}N \big[\big(\Sigo_t*\hat{v} \big) \,\overline{\sigma}_t \, - \,  \big(\Sigob_t*\hat{v} \big) \,\sigma_t\big] \, ,  \\
			i\partial_t \sigma_t =& 2\big(E+\frac{\lambda}{N}\Gamo_t*(\hat{v}+\hat{v}(0))\big)\sigma_t \, + \, \frac{\lambda}{N} \big( \Sigo_t* \hat{v}\big) (1+2\gamma_t) \, .\label{eq-sigma-HFB-lem-1}
		\end{align} 
	\end{lemma}
	\prf	
		Multiplying \eqref{eq-HFB-raw-0} by $\overline{\sigma}_t(p)$, and taking imaginary parts, the l.h.s. of \eqref{eq-HFB-raw-0} reads
		\begin{align}
		\MoveEqLeft \frac{\Re(\overline{\sigma}_t(p)\partial_t\sigma_t(p))}{2} \, - \, \frac{|\sigma_t(p)|^2\partial_t\gamma_t(p)}{2(1+\gamma_t(p))}\\
		=&  \frac14 \partial_t\big(|\sigma_t(p)|^2 \, - \, \gamma_t(p)^2\big)\\
		=&  \frac14\partial_t \gamma_t(p) \, ,
		\end{align} 
		where we employed \eqref{eq-gamma-sigma-rel-0}. Thus, we have that
		\begin{align}
		\frac14\partial_t \gamma_t =&  \frac{\lambda}{2N} \Im\Big( \big(\Sigo_t* \hat{v}\big)(1+\gamma_t)\overline{\sigma}_t \,  + \, \big(\Sigob_t*\hat{v}\big) \frac{|\sigma_t|^2\sigma_t}{1+\gamma_t}\Big) \, ,
		\end{align} 
		which, using \eqref{eq-gamma-sigma-rel-0}, is equivalent to
		\begin{align}
		\partial_t\gamma_t(p) =&  \frac{2\lambda}{N} \Im\Big( \big(\Sigo_t* \hat{v}\big)(1+\gamma_t)\overline{\sigma}_t \,  + \, \big(\Sigob_t*\hat{v}\big) \gamma_t \sigma_t\Big)\\
		=&  \frac{2\lambda}{N} \Im\Big( \big(\Sigo_t* \hat{v}\big)\overline{\sigma}_t\Big)
		\, . 
		\end{align} 
		Substituting this identity into \eqref{eq-HFB-raw-0}, we find that
		\begin{align}
		\MoveEqLeft \frac{i\partial_t \sigma_t(p)}2 \, - \, \frac{\lambda}{N}\frac{\sigma_t(p)}{2(1+\gamma_t(p))}\Big( \big(\Sigo_t* \hat{v}\big)\overline{\sigma}_t \, - \, \big(\Sigob_t* \hat{v}\big)\sigma_t\Big) \, = \,\\
		&\Big(E(p)+\frac{\lambda}{N}\Gamo_t*\big(\hat{v}+\hat{v}(0)\big)(p)\Big)\sigma_t(p)  \\
		& + \, \frac{\lambda}{2N} \Big( \big(\Sigo_t* \hat{v}\big)(p) (1+\gamma_t(p)) \, + \, \big(\Sigob_t*\hat{v}\big)(p) \frac{\sigma_t(p)^2}{1+\gamma_t(p)}\Big) \, .
		\end{align} 
		Isolating $i\partial_t \sigma_t$ on the l.h.s. and using \eqref{eq-gamma-sigma-rel-0}, we obtain that
		\begin{align}
			i\partial_t \sigma_t =&  2\big(E+\frac{\lambda}{N}\Gamo_t*(\hat{v}+\hat{v}(0))\big)\sigma_t  \, + \, \frac{\lambda}{N} \big(\Sigo_t* \hat{v}\big) (1+2\gamma_t) \, . 
		\end{align} 
		This concludes the proof.
	\endprf 
	
	\begin{lemma}[Bogoliubov dispersion] \label{lem-bog-dispersion}
		We have that $\HHFBd(t)=0$ if
		\begin{align}
			\MoveEqLeft \Omega_t(p) \, = \,\\
			&\Big(E(p)+\frac{\lambda}N\big((\gamma_t+N\vol|\phi_t|^2\delta)*(\hat{v}+\hat{v}(0))\big)(p)\Big)\big(1+2\gamma_t(p)\big) \\
			& + \, \frac{2\lambda}N\Re\Big(\big((\overline{\sigma}_t+N\vol\overline{\phi}_t^2\delta)*\hat{v}\big)(p)\sigma_t(p)\Big) \, -  \frac{\Re\big(\overline{\sigma}_t(p)i\partial_t\sigma_t(p)\big)}{1+\gamma_t(p)}\, .
		\end{align}
		If $\sigma_t$ satisfies \eqref{eq-sigma-HFB-lem-1}, we have that
		\begin{align}
			\Omo_t  \, =\, E+\frac{\lambda}{N}\Gamo_t*\big(\hat{v}+\hat{v}(0)\big) \, + \, \frac{\lambda}{N}\frac{\Re\big((\Sigo_t*\hat{v})\sigob_t\big) }{1+\gamo_t}  \, .
		\end{align}  
	\end{lemma}
	\prf
		Lemma \ref{lem-hfluc} implies that $\HHFBd(t)=0$ if
		\begin{align}
			\MoveEqLeft \Omega_t(p) \, = \,\\
			&\Big(E(p)+\frac{\lambda}N\big((\gamma_t+N\vol|\phi_t|^2\delta)*(\hat{v}+\hat{v}(0))\big)(p)\Big)\big(1+2\gamma_t(p)\big) \\
			& + \, \frac{2\lambda}N\Re\Big(\big((\overline{\sigma}_t+N\vol\overline{\phi}_t^2\delta)*\hat{v}\big)(p)\sigma_t(p)\Big) \, -  \frac{\Re\big(\overline{\sigma}_t(p)i\partial_t\sigma_t(p)\big)}{1+\gamma_t(p)} \, . 
		\end{align} 
		Substituting \eqref{eq-sigma-ren-2} in this expression and recalling definitions \eqref{def-Gamo} of $\Gamo$ and \eqref{def-Sigo} of $\Sigo$, we obtain
		\begin{align}
			\Omo_t = & - 2\big(E+\frac{\lambda}{N}\Gamo_t*(\hat{v}+\hat{v}(0))\big)\frac{|\sigo_t|^2}{1+\gamo_t}+  \big[E+\frac{\lambda}{N}\Gamo_t*\big(\hat{v}+\hat{v}(0)\big)\big](1+2\gamo_t) \\
			& - \, \frac{\lambda}{N} \Re\Big(\big( \Sigo_t* \hat{v}\big) \sigob_t\Big)\frac{1+2\gamma_t}{1+\gamo_t} \, + \, \frac{2\lambda}{N} \Re\Big(\big(\Sigo_t*\hat{v}\big) \sigob_t \Big) \, .
		\end{align}
		Recalling that $|\sigma_t|^2=\gamma_t(1+\gamma_t)$, see \eqref{eq-gamma-sigma-rel-0}, this results in
		\begin{align}
			\Omo_t  \, =\, E+\frac{\lambda}{N}\Gamo_t*\big(\hat{v}+\hat{v}(0)\big) \, + \, \frac{\lambda}{N}\frac{\Re\big((\Sigo_t*\hat{v})\sigob_t\big) }{1+\gamo_t} \, .
		\end{align}  
		This finishes the proof.
	\endprf 
	
	\section{HFB system analysis}
	
	\begin{lemma}[Generalized HFB equations]\label{lem-HFB-gen}
		Let $(\phi,\gamma,\sigma)$ satisfy either of the renormalized equations
		\begin{align}
			i\partial_t \phi_t  =&  \frac{\lambda}{N} \Big( \big(\Gamj_t*(\hat{v}+\hat{v}(0))\big) (0) \phi_t \, + \, (\Sigj_t*\hat{v}) (0) \overline{\phi}_t\Big) \\
			& - \,  2\lambda\vol\hat{v}(0) |\phi_t|^2 \phi_t \, , \label{eq-phi-HFB-genlem-1}\\
			i\partial_t \gamma_t =&  \frac{\lambda}N \big[\big(\Sigj_t*\hat{v} \big) \,\overline{\sigma}_t \, - \,  \big(\overline{\Sigj_t}*\hat{v} \big) \,\sigma_t\big] \, , \label{eq-gamma-HFB-genlem-1}\\
			i\partial_t \sigma_t =& 2\big(E+\frac{\lambda}{N}\Gamj_t*(\hat{v}+\hat{v}(0))\big)\sigma_t \, + \, \frac{\lambda}{N} \big( \Sigj_t* \hat{v}\big) (1+2\gamma_t)  \label{eq-sigma-HFB-genlem-1}
		\end{align} 
		for $j=0,1$, where 
		\begin{align}
			\Gamj =& (1+2\delta_{1,j}\fplus)\gamma \, + \, \delta_{1,j}\fplus \, + \, N\vol|\phi|^2\delta \, , \\
			\Sigj =&  (1+2\delta_{1,j}\fplus)\sigma \, + \, N\vol\phi^2\delta \, ,
		\end{align} 
		see definitions \eqref{def-Gamo}, \eqref{def-Sigo}, \eqref{def-Gamt} and \eqref{def-Sigt}. Then $(\phi,\Gamj,\Sigj)$ satisfies
		\begin{align} 
		i\partial_t \phi_t  =&  \frac{\lambda}{N} \Big( \big(\Gamj_t*(\hat{v}+\hat{v}(0))\big) (0) \phi_t \, + \, (\Sigj_t*\hat{v}) (0) \overline{\phi}_t\Big) \\
		& - \,  2\lambda\vol\hat{v}(0) |\phi_t|^2 \phi_t \, , \\
		\partial_t \Gamj_t =&  -\frac{2\lambda}N \Im\Big(\big(\overline{\Sigj_t}*\hat{v} \big) \,\Sigj_t\Big) \, , \\
		i\partial_t \Sigj_t =& 2\big(E+\frac{\lambda}{N}\Gamj_t*(\hat{v}+\hat{v}(0))\big)\Sigj_t \, + \, \frac{\lambda}{N} \big( \Sigj_t* \hat{v}\big) (1+2\Gamj_t) \\
		& - \, 4N\lambda\vol^2\hat{v}(0)|\phi|^2\phi^2\delta 
		\end{align} 
		for $j=0,1$.
	\end{lemma}
	\prf
		We start with 
		\begin{align}
			\partial_t |\phi_t|^2 = & 2\Re(\overline{\phi}_t\partial_t \phi_t )\\
			=&  2\Im(\overline{\phi}_ti\partial_t \phi_t ) \\
			=&  \frac{2\lambda}N\Im\Big((\Sigj_t*\hat{v})(0) \overline{\phi}_t^2\Big) \, .
		\end{align} 
		As a consequence, \eqref{eq-gamma-HFB-genlem-1} implies that
		\begin{align}
			\partial_t\Gamj_t =&  (1+2\delta_{1,j}\fplus)\partial_t\gamma_t \, + \, N\vol \partial_t|\phi_t|^2\delta \\
			=&  \frac{2\lambda}N \Im\Big(\big(\Sigj_t*\hat{v} \big) \,\big( (1+2\delta_{1,j}\fplus)\overline{\sigma}_t+N\vol\overline{\phi}_t^2\delta\big)\Big) \\
			=&  -\frac{2\lambda}N \Im\Big(\big(\overline{\Sigj_t}*\hat{v} \big) \,\Sigj_t\Big) \, .
		\end{align} 
		Similarly, we obtain, using \eqref{eq-phi-HFB-genlem-1} and \eqref{eq-sigma-HFB-genlem-1}, that
		\begin{align}
			i\partial_t\Sigj_t =&  (1+2\delta_{1,j}\fplus)i\partial_t\sigma_t \, + \, 2N\vol\phi_t i\partial_t\phi_t\delta \\
			=&  (1+2\delta_{1,j}\fplus)\Big(2\big(E+\frac{\lambda}{N}\Gamj_t*(\hat{v}+\hat{v}(0))\big)\sigma_t \\
			& + \, \frac{\lambda}{N} \big( \Sigj_t* \hat{v}\big) (1+2\gamma_t) \Big) \, + \, 2\lambda\vol\phi_t \delta \Big( \big(\Gamj_t*(\hat{v}+\hat{v}(0))\big)\phi_t \\
			& + \, (\Sigj_t*\hat{v})\overline{\phi}_t\Big) \, - \,  4N\lambda\vol^2\hat{v}(0) |\phi_t|^2 \phi_t^2 \\
			=& 2\big(E+\frac{\lambda}{N}\Gamj_t*(\hat{v}+\hat{v}(0))\big)\Sigj_t \, + \, \frac{\lambda}{N} \big( \Sigj_t* \hat{v}\big) (1+2\Gamj_t) \\
			& - \, 4N\lambda\vol^2\hat{v}(0)|\phi|^2\phi^2\delta \, .
		\end{align} 
		This concludes the proof.
	\endprf 
	
	\begin{lemma} \label{lem-HFB-NL-contraction}
			Assume that $\hat{v}\in L^1_{\sqrt{1+E}}\cap L^\infty(\lattice)$. Recall the definition of the nonlinearity $\vec{\cJ}=(\cJ_1,\cJ_2,\cJ_3)$ in the HFB equations, where
			\begin{align}
			\cJ_1(\phi,\Gamma,\Sigma) =&  - \, i\Big[\frac{\lambda}{N} \Big( \big(\Gamma*(\hat{v}+\hat{v}(0))\big) (0) \phi \, + \, (\Sigma*\hat{v}) (0) \overline{\phi}\Big) \, , \\
			& - \,  2\lambda\vol\hat{v}(0) |\phi|^2 \phi\Big] \\
			\cJ_2(\phi,\Gamma,\Sigma) =&  - \, \frac{2\lambda}N \Im\Big(\big(\overline{\Sigma}*\hat{v} \big) \, \Sigma\Big) \, ,\\
			\cJ_3(\phi,\Gamma,\Sigma) =&  - \, i\Big[2\big(\frac{\lambda}{N}\Gamma*(\hat{v}+\hat{v}(0))\big)\Sigma \, + \, \frac{\lambda}{N} \big( \Sigma* \hat{v}\big) (1+2\Gamma)\\
			& - \, 4N\lambda\vol^2\hat{v}(0)|\phi|^2\phi^2 \, \delta\Big] \, ,
		\end{align} 
		see \eqref{def-cj1}--\eqref{def-cj3}. Then there is a constant $C>0$ such that we have that
		\begin{align}
			\MoveEqLeft \|\vec{\cJ}(\phi_1,\Gamma_1,\Sigma_1) \, - \, \vec{\cJ}(\phi_2,\Gamma_2,\Sigma_2)\|_{\hfbspace^1}\\
			\leq & C \Big(\frac{\lambda}N(\|\hat{v}\|_{L^1_{\sqrt{1+E}}}+\|\hat{v}\|_\infty)\big(\| (\phi_1,\Gamma_1,\Sigma_1)\|_{\hfbspace^1} + \| (\phi_2,\Gamma_2,\Sigma_2)\|_{\hfbspace^1}\big) \\
			& \|(\phi_1-\phi_2,\Gamma_1-\Gamma_2,\Sigma_1-\Sigma_2)\|_{\hfbspace^1} \, + \, \lambda\vol\hat{v}(0)(|\phi_1|^2+|\phi_2|^2)\\
			& \big(1+N\vol^{\frac32}(|\phi_1|+|\phi_2|)\big)|\phi_1-\phi_2|
		\end{align} 
	\end{lemma}
	\prf
		$\cJ_1$ satisfies
		\begin{align}
			\MoveEqLeft |\cJ_1(\phi_1,\Gamma_1,\Sigma_1) \, - \, \cJ_1(\phi_2,\Gamma_2,\Sigma_2)|\\
			\leq& \frac{\lambda}{N} \Big( |\Gamma_1-\Gamma_2|*(|\hat{v}|+|\hat{v}(0)|)(0)|\phi_1| \, + \, |\Gamma_2|*(|\hat{v}|+|\hat{v}(0)|)(0)|\phi_1-\phi_2|\\
			& + \, \big(|\Sigma_1-\Sigma_2|*|\hat{v}|\big)(0) |\phi_1| \, + \, \big(|\Sigma_2|*|\hat{v}|\big)(0) |\phi_1-\phi_2|\Big) \\
			& + \,  2\lambda\vol|\hat{v}(0)| \big||\phi_1|^2 \phi_1-|\phi_2|^2 \phi_2\big| \, .
		\end{align} 
		Using the fact that
		\begin{align}
			|\phi_1|^2 \, - \, |\phi_2|^2 \, = \, \Re\big[(\overline{\phi}_1-\overline{\phi}_2)(\phi_1+\phi_2)\big] \, , \label{eq-phi2-diff-0}
		\end{align} 
		we thus obtain that
		\begin{align}
			\MoveEqLeft |\cJ_1(\phi_1,\Gamma_1,\Sigma_1) \, - \, \cJ_1(\phi_2,\Gamma_2,\Sigma_2)|\\
			\leq& \frac{\lambda}{N} \Big( 2\|\hat{v}\|_\infty\big(|\phi_1|\|\Gamma_1-\Gamma_2\|_1 \, + \, \|\Gamma_2\|_1 |\phi_1-\phi_2|\big) \\
			& + \, \|\hat{v}\|_\infty \big(\|\Sigma_1-\Sigma_2\|_1 |\phi_1| \, + \, \|\Sigma_2\|_1|\phi_1-\phi_2|\big)\Big)  \\
			& + \, 4\lambda\vol\hat{v}(0)(|\phi_1|^2+|\phi_2|^2)|\phi_1-\phi_2| \, .
		\end{align} 
		Next, we find that 
		\begin{align}
			\MoveEqLeft \|\cJ_2(\phi_1,\Gamma_1,\Sigma_1) \, - \, \cJ_2(\phi_2,\Gamma_2,\Sigma_2)\|_{L^1_{1+E}}\\
			=&  \frac{2\lambda}N\int \dx{p} \big|(\overline{\Sigma}_1-\overline{\Sigma}_2)*\hat{v} (p)\, \Sigma_1(p) \, + \, \overline{\Sigma}_2*\hat{v}(p) \, (\Sigma_1-\Sigma_2)(p)\, \big|\big(1+E(p)\big) \\
			\leq & \frac{2\lambda}N \|(\Sigma_1-\Sigma_2)*\hat{v}\|_{L^2_{1+E}} \big(\|\Sigma_1\|_{L^2_{1+E}} + \|\Sigma_2\|_{L^2_{1+E}}\big) \label{eq-J2-est-1}
		\end{align} 
		by Cauchy-Schwarz, where we also used that $\int \dx{p} f*g(p) h(p)=\int \dx{p} h*g(p) f(p)$. In addition, observe that Young's inequality implies
		\begin{align}
			\|(\Sigma_1-\Sigma_2)*\hat{v}\|_2 \, \leq\, \|\hat{v}\|_1 \|\Sigma_1-\Sigma_2\|_2 \, . \label{eq-sigma-diff*v-1}
		\end{align} 
		Using $E(p)=\frac12|p|^2$ and 
		\begin{align} 
			\sqrt{E(p)} \, \leq \, \sqrt{E(p-q)}+\sqrt{E(q)} \, ,
		\end{align} 	
		Young's inequality yields
		\begin{align}
			\|\big((\Sigma_1-\Sigma_2)*\hat{v}\big)\sqrt{E}\|_2 \leq & \, \|\big(\sqrt{E}(\Sigma_1-\Sigma_2)\big)*\hat{v}\|_2 \, + \, \|(\Sigma_1-\Sigma_2)*(\sqrt{E}\hat{v})\|_2 \\
			\leq & \, \|\sqrt{E}(\Sigma_1-\Sigma_2)\|_2\|\hat{v}\|_1  \, + \, \|\Sigma_1-\Sigma_2\|_2 \|\sqrt{E}\hat{v}\|_1 \, . \label{eq-sigma-diff*v-2}
		\end{align} 
		Collecting \eqref{eq-J2-est-1}, \eqref{eq-sigma-diff*v-1}, and \eqref{eq-sigma-diff*v-2}, we conclude that
		\begin{align}
			\MoveEqLeft \|\cJ_2(\phi_1,\Gamma_1,\Sigma_1) \, - \, \cJ_2(\phi_2,\Gamma_2,\Sigma_2)\|_{L^1_{1+E}}\\
			\leq& \frac{2\lambda}N \big(\|\Sigma_1\|_{L^2_{1+E}} + \|\Sigma_2\|_{L^2_{1+E}}\big)\|\hat{v}\|_{L^1_{\sqrt{1+E}}}\\
			& (\|\Sigma_1-\Sigma_2\|_{L^2_{1+E}}+ \|\Sigma_1-\Sigma_2\|_\infty) \, .
		\end{align} 
		With similar steps, we arrive at
		\begin{align}
			\MoveEqLeft \|\cJ_2(\phi_1,\Gamma_1,\Sigma_1) \, - \, \cJ_2(\phi_2,\Gamma_2,\Sigma_2)\|_\infty\\
			\leq& \frac{2\lambda}N \|(\Sigma_1-\Sigma_2)*\hat{v}\|_\infty \big(\|\Sigma_1\|_\infty + \|\Sigma_2\|_\infty\big) \\
			\leq& \frac{2\lambda}N \|\hat{v}\|_1 \big(\|\Sigma_1\|_\infty + \|\Sigma_2\|_\infty\big)\|\Sigma_1-\Sigma_2\|_\infty \, .
		\end{align} 
		Finally, $\cJ_3$ satisfies		
		\begin{align}
		\MoveEqLeft \|\cJ_3(\phi_1,\Gamma_1,\Sigma_1) \, - \, \cJ_3(\phi_2,\Gamma_2,\Sigma_2)\|_{L^2_{1+E}}\\
		\leq& \frac{2\lambda}{N}\Big(\|(\Gamma_1-\Gamma_2)*(\hat{v}+\hat{v}(0))\big)\Sigma_1\|_{L^2_{1+E}} \, + \, \|\Gamma_2*(\hat{v}+\hat{v}(0))\big)(\Sigma_1-\Sigma_2)\|_{L^2_{1+E}}\\
		& + \, \|(\Sigma_1-\Sigma_2)* \hat{v}\Gamma_1\|_{L^2_{1+E}} \, + \, \|\Sigma_2*\hat{v}(\Gamma_1-\Gamma_2)\|_{L^2_{1+E}} \, + \, \frac12\|(\Sigma_1-\Sigma_2)* \hat{v}\|_{L^2_{1+E}}\\
			& + \, 4N\lambda\vol^{\frac52}\hat{v}(0)\big||\phi_1|^2\phi_1^2-|\phi_2|^2\phi_2^2\big| \, . \label{eq-dj3-0}
		\end{align} 
		Again, Cauchy-Schwarz implies
		\begin{align}
			\MoveEqLeft \|(\Gamma_1-\Gamma_2)*(\hat{v}+\hat{v}(0))\big)\Sigma_1\|_{L^2_{1+E}}\\
			\leq & \, \|(\Gamma_1-\Gamma_2)*\hat{v}\Sigma_1\|_{L^2_{1+E}} \, + \, \hat{v}(0)\|\Gamma_1-\Gamma_2\|_1\|\Sigma_1\|_{L^2_{1+E}}\\
			\leq & \, \big(\|(\Gamma_1-\Gamma_2)*\hat{v}\|_{L^2_{1+E}} \, + \, \hat{v}(0)\|\Gamma_1-\Gamma_2\|_1\big) \|\Sigma_1\|_{L^2_{1+E}} \, . 
		\end{align} 
	\eqref{eq-sigma-diff*v-1} and \eqref{eq-sigma-diff*v-2} hence yield
		\begin{align}
			\MoveEqLeft \|(\Gamma_1-\Gamma_2)*(\hat{v}+\hat{v}(0))\big)\Sigma_1\|_{L^2_{1+E}}\\
			\leq & \, \big(\|\hat{v}\|_{L^1_{\sqrt{1+E}}}\|\Gamma_1-\Gamma_2\|_{L^2_{1+E}} \, + \, \hat{v}(0)\|\Gamma_1-\Gamma_2\|_1\big) \|\Sigma_1\|_{L^2_{1+E}} \, . \label{eq-j3-dgamma-1}
		\end{align} 
		Moreover, by \eqref{eq-phi2-diff-0}, we have that 
		\begin{align}
			\big||\phi_1|^2\phi_1^2-|\phi_2|^2\phi_2^2\big| \leq&  |\phi_1|^2|\phi_1^2-\phi_2^2|+|\phi_2^2\big||\phi_1|^2-|\phi_2|^2\big|\\
		\leq& (|\phi_1|^2+|\phi_2|^2)(|\phi_1|+|\phi_2|)|\phi_1-\phi_2| \, . \label{eq-phi4-diff-0}
		\end{align} 
		Then \eqref{eq-dj3-0}, \eqref{eq-j3-dgamma-1}, \eqref{eq-phi4-diff-0} and analogous estimates to those for $\cJ_2$ imply that
		\begin{align}
			\MoveEqLeft \|\cJ_3(\phi_1,\Gamma_1,\Sigma_1) \, - \, \cJ_3(\phi_2,\Gamma_2,\Sigma_2)\|_{L^2_{1+E}}\\
		\leq& \frac{2\lambda}{N}(\|\hat{v}\|_{L^1_{\sqrt{1+E}}}+\|\hat{v}\|_\infty)\big((\|\Sigma_1\|_{L^2_{1+E}}+ \|\Sigma_2\|_{L^2_{1+E}}) \|\Gamma_1-\Gamma_2\|_{L^2_{1+E}}\\
		& + \, (\|\Gamma_1\|_{L^2_{1+E}} + \|\Gamma_2 \|_{L^2_{1+E}}+\frac12)\|\Sigma_1-\Sigma_2\|_{L^2_{1+E}}\big)\\
			& + \, 4N\lambda\vol^{\frac52}\hat{v}(0) (|\phi_1|^2+|\phi_2|^2)(|\phi_1|+|\phi_2|)|\phi_1-\phi_2| \, .
		\end{align}
		 Observe that 
		 \begin{align} 
		 	\|\Gamma\|_{L^2_{1+E}} \leq & \|\Gamma\|_{L^1_{1+E}}^{\frac12} \|\Gamma\|_\infty^{\frac12}\\
		 	\leq & \, \frac12(\|\Gamma\|_{L^1_{1+E}} +\|\Gamma\|_\infty) \, ,
		 \end{align} 
		 which is why
		 \begin{align}
			\MoveEqLeft \|\cJ_3(\phi_1,\Gamma_1,\Sigma_1) \, - \, \cJ_3(\phi_2,\Gamma_2,\Sigma_2)\|_{L^2_{1+E}}\\
		\leq& \frac{\lambda}{N}(\|\hat{v}\|_{L^1_{\sqrt{1+E}}}+\|\hat{v}\|_\infty)\big((\|\Sigma_1\|_{L^2_{1+E}}+ \|\Sigma_2\|_{L^2_{1+E}}) (\|\Gamma_1-\Gamma_2\|_{L^1_{1+E}} + \|\Gamma_1-\Gamma_2\|_\infty)\\
		& + \, (\|\Gamma_1\|_{L^1_{1+E}} + \|\Gamma_1\|_\infty + \|\Gamma_2 \|_{L^1_{1+E}} + \|\Gamma_2 \|_\infty+1)\|\Sigma_1-\Sigma_2\|_{L^2_{1+E}}\big)\\
			& + \, 4N\lambda\vol^{\frac52}\hat{v}(0) (|\phi_1|^2+|\phi_2|^2)(|\phi_1|+|\phi_2|) |\phi_1-\phi_2| \, .
		\end{align}
		Similarly, we obtain that
		\begin{align}
			\MoveEqLeft \|\cJ_3(\phi_1,\Gamma_1,\Sigma_1) \, - \, \cJ_3(\phi_2,\Gamma_2,\Sigma_2)\|_\infty\\
		\leq& \frac{4\lambda}{N}\|\hat{v}\|_\infty\big((\|\Sigma_1\|_\infty+ \|\Sigma_2\|_\infty) \|\Gamma_1-\Gamma_2\|_1\\
		& + \, (\|\Gamma_1\|_1 + \|\Gamma_2\|_1+1)\|\Sigma_1-\Sigma_2\|_\infty\big)\\
			& + \, 4N\lambda\vol^{\frac52}\hat{v}(0) (|\phi_1|^2+|\phi_2|^2)(|\phi_1|+|\phi_2|) |\phi_1-\phi_2| \, .
		\end{align}
	\endprf 
	
	\prf[Proof of Lemma \ref{lem-conservation}] 
		As a consequence of $(\phi,\Gamma,\Sigma)$ being a mild solution, we have that
		\begin{align}
			\int \dx{p} \Gamma_t(p) =&  \int \dx{p} \Gamma_0(p) \, - \, \frac{2\lambda}N \int_0^t \dx{s} \Im\scp{\Sigma_s*\hat{v}}{\Sigma_s} \, . \label{eq-mass-cons-proof-1} 
		\end{align} 
		Let $\widecheck{f}$ denote the inverse Fourier transform of $f$. Observe that due to Plancherel, 
		\begin{align}
			\scp{\Sigma_s*\hat{v}}{\Sigma_s} \, = \, \int dx \, |\widecheck{\Sigma}_s(x)|^2 v(x) \label{eq-e-sigma>=0}
		\end{align} 
		is a real number, as $\|v\|_\infty\leq \|\hat{v}\|_1 <\infty$, $v$ is real-valued, and $\Sigma_s\in \ell^2(\lattice)$. As a consequence, \eqref{eq-mass-cons-proof-1} implies
		\begin{align} 
			\int \dx{p} \Gamma_t(p) \, = \, \int \dx{p} \Gamma_0(p) \, ,
		\end{align} 
		as desired.
		\par In order to show energy-conservation, we split the energy into its individual terms. From equation \eqref{eq-HFB-ren-2}, we have that
		\begin{align}
			\int \dx{p} E(p) \Gamma_t(p) \, = \, \int \dx{p} E(p) \Gamma_0(p) -\frac{2\lambda}N \int_0^t \dx{s} \scp{\Sigma_s*\hat{v}}{E\Sigma_s} \, , \label{eq-ec-egamma-1}
		\end{align}  
		where we notice that $\Sigma_s*\hat{v}\in L^2_{1+E}\cap L^\infty(\lattice)$. In particular, all terms in \eqref{eq-ec-egamma-1} are finite. Next, we have that
		\begin{align}
			\MoveEqLeft \frac{\lambda}{2N}\scp{\Gamma_t*(\hat{v}+\hat{v}(0))}{\Gamma_t}-\frac{\lambda}{2N}\scp{\Gamma_0*(\hat{v}+\hat{v}(0))}{\Gamma_0}\\
			=&  \frac{\lambda}N\int_0^t \dx{s} \scp{\Gamma_s*(\hat{v}+\hat{v}(0))}{\partial_s\Gamma_s}\\
			=&  -\frac{2\lambda^2}{N^2}\int_0^t \dx{s} \scp{\Gamma_s*(\hat{v}+\hat{v}(0))}{\Im\Big((\overline{\Sigma}_s*\hat{v})\Sigma_s\Big)}\\
			=&  -\frac{2\lambda^2}{N^2}\int_0^t \dx{s} \Im \scp{\Sigma_s*\hat{v}}{(\Gamma_s*\hat{v})\Sigma_s} \, , \label{eq-ec-g2-1}
		\end{align} 
		where the expressions in each line are finite. Using the fact that $\Sigma_s*\hat{v}\in L^2_{1+E}\cap L^\infty(\lattice)$, equation \eqref{eq-HFB-ren-2} for $\Sigma$
		\begin{align}
			\MoveEqLeft \frac{\lambda}{2N}\scp{\Sigma_t*\hat{v}}{\Sigma_t} \, - \, \frac{\lambda}{2N}\scp{\Sigma_0*\hat{v}}{\Sigma_0}\\
			=&  \frac{\lambda}N \int_0^t \dx{s} \Im \scp{\Sigma_s*\hat{v}}{i\partial_s\Sigma_s}\\
			=&  \frac{\lambda}N \Im \int_0^t \dx{s} \int \dx{p} (\overline{\Sigma}_s*\hat{v})(p)\Big(2\big(E(p)+\frac{\lambda}{N}\Gamma_s*(\hat{v}(p)+\hat{v}(0))\big)\Sigma_s(p)) \\
			& + \, \frac{\lambda}{N} \big( \Sigma_s*\hat{v}(p)\big) (1+2\Gamma_s(p)) \, - \, 4N\lambda\vol^2\hat{v}(0)|\phi_s|^2\phi_s^2\delta(p)\Big) \\
			=&  \frac{2\lambda}N \Im \int_0^t \dx{s} \Big( \scp{\Sigma_s*\hat{v}}{(E +\frac{\lambda}{N}\Gamma_s*\hat{v})\Sigma_s}\\
			& - \, 2N\lambda\vol^2\hat{v}(0)|\phi_s|^2(\overline{\Sigma}_s*\hat{v})(0)\phi_s^2 \Big) \, . \label{eq-ec-s2-1}
		\end{align} 
		In addition, equation \eqref{eq-HFB-ren-2} for $\phi$ implies
		\begin{align}
			\MoveEqLeft - N\vol^2\lambda\hat{v}(0)\big(|\phi_t|^4-|\phi_0|^4\big)\\
			=&  -2N\vol^2\lambda\hat{v}(0)\Im \int_0^t \dx{s} |\phi_s|^2\overline{\phi}_s i\partial_s \phi_s\\
			=&  -2\vol^2\lambda^2\hat{v}(0)\Im \int_0^t \dx{s} |\phi_s|^2
			(\Sigma_s*\hat{v}) (0) \overline{\phi}_s^2 \\
			=&  2\vol^2\lambda^2\hat{v}(0)\Im \int_0^t \dx{s} |\phi_s|^2
			(\overline{\Sigma}_s*\hat{v}) (0) \phi_s^2 \, . \label{eq-ec-phi4-1}
		\end{align} 
		Collecting \eqref{eq-ec-egamma-1}, \eqref{eq-ec-g2-1}, \eqref{eq-ec-s2-1}, \eqref{eq-ec-phi4-1}, we obtain that $\ehfb(\phi_t,\Gamma_t,\Sigma_t)$ is differentiable w.r.t. t, and that
		\begin{align}
			\partial_t\ehfb(\phi_t,\Gamma_t,\Sigma_t) \, = \, 0 \, .
		\end{align} 
		This concludes the proof.
	\endprf 
	
	\section{Boltzmann evolution}
	
	\begin{lemma}[Calculation of collision operators]\label{lem-cub-bol}
		Using the notation in Lemma \ref{lem-hfluc}, the cubic Boltzmann operator is given by
		\begin{align}
		\MoveEqLeft \frac1N\int_0^t\dx{s}Q_3[f_0](s,p)  \, = \, \\
		& \frac{2\lambda^2}{N}\Re \int_{[0,t]^2} \dx{\bs_2}  \mathds{1}_{s_1\geq s_2}\int \dx{\bp_3}  \Big(\frac1{2!}\big(\delta(p_1-p)+\delta(p_2-p)-\delta(p_3-p)\big) \\
		& \bbf{1}{2}_{s_1}(\bp_3)\bbfb{1}{2}_{s_2}(\bp_3)e^{i\int_{s_2}^{s_1} \dx{\tau} \big(\Omega_\tau(p_1)+\Omega_\tau(p_2)-\Omega_\tau(p_3)\big)}\delta(p_1+p_2-p_3)\\
		&\big(\fbar(p_1)\fbar(p_2)f_0(p_3)-f_0(p_1)f_0(p_2)\fbar(p_3)\big) \\
		& + \, \frac1{3!}\big(\delta(p_1-p)+\delta(p_2-p)+\delta(p_3-p)\big)\\
		& \bbf{0}{3}_{s_1}(\bp_3)\bbfb{0}{3}_{s_2}(\bp_3)e^{i\int_{s_2}^{s_1} \dx{\tau} \big(\Omega_\tau(p_1)+\Omega_\tau(p_2)+\Omega_\tau(p_3)\big)}\delta(p_1+p_2+p_3) \\
		&\big(\fbar(p_1)\fbar(p_2)\fbar(p_3)-f_0(p_1)f_0(p_2)f_0(p_3)\big)\Big) \, ,
		\end{align} 
		and the quartic Boltzmann operator is given by
		\begin{align}
			\MoveEqLeft \frac1{N^2}\int_0^t\dx{s} Q_4[f_0](s,p) \, = \, \\
			& \frac{2\lambda^2}{N^2}\Re \int_{[0,t]^2} \dx{\bs_2}  \mathds{1}_{s_1\geq s_2} \int \dx{\bp_4}  \Big(\frac1{(2!)^2}\bbf{2}{2}_{s_1}(\bp_4)\bbfb{2}{2}_{s_2}(\bp_4)\\
			& \big(\delta(p_1-p)+\delta(p_2-p)-\delta(p_3-p)-\delta(p_4-p)\big)\\
			& \delta(p_1+p_2-p_3-p_4)e^{i\int_{s_2}^{s_1} \dx{\tau} \big(\Omega_\tau(p_1)+\Omega_\tau(p_2)-\Omega_\tau(p_3)-\Omega_\tau(p_4)\big)}\\
			&\big(\fbar(p_1)\fbar(p_2)f_0(p_3)f_0(p_4)-f_0(p_1)f_0(p_2)\fbar(p_3)\fbar(p_4)\big) \\
			& + \, \frac1{3!}\bbf{1}{3}_{s_1}(\bp_4)\bbfb{1}{3}_{s_2}(\bp_4)\\
			& \big(\delta(p_1-p)+\delta(p_2-p)+\delta(p_3-p)-\delta(p_4-p)\big)\\
			& \delta(p_1+p_2+p_3-p_4)e^{i\int_{s_2}^{s_1} \dx{\tau} \big(\Omega_\tau(p_1)+\Omega_\tau(p_2)+\Omega_\tau(p_3)-\Omega_\tau(p_4)\big)}\\
			&\big(\fbar(p_1)\fbar(p_2)\fbar(p_3)f_0(p_4)-f_0(p_1)f_0(p_2)f_0(p_3)\fbar(p_4)\big)\\
			& + \, \frac1{4!}\bbf{0}{4}_{s_1}(\bp_4)\bbfb{0}{4}_{s_2}(\bp_4)\\
			& \big(\delta(p_1-p)+\delta(p_2-p)+\delta(p_3-p)+\delta(p_4-p)\big)\\
			& \delta(p_1+p_2+p_3+p_4)e^{i\int_{s_2}^{s_1} \dx{\tau} \big(\Omega_\tau(p_1)+\Omega_\tau(p_2)+\Omega_\tau(p_3)+\Omega_\tau(p_4)\big)}\\
			&\big(\fbar(p_1)\fbar(p_2)\fbar(p_3)\fbar(p_4)-f_0(p_1)f_0(p_2)f_0(p_3)f_0(p_4)\big)\Big) \, .
		\end{align}
	\end{lemma}
	\prf
		Recall that Lemma \ref{lem-hfluc} implies
		\begin{align}
			\MoveEqLeft \Hcub(t) \, := \\
			& \frac{\lambda}{\sqrt{N}} \int \dx{\bp_3} \Big(\delta(p_1+p_2+p_3)e^{i\int_0^t \dx{\tau} \big(\Omega_\tau(p_1)+\Omega_\tau(p_2)+\Omega_\tau(p_3)\big)}\\
			&\frac1{3!} \bbf{0}{3}_t(\bp_3)\ad_{p_1}\ad_{p_2}\ad_{p_3}\\
			& + \,  \delta(p_1+p_2-p_3)e^{i\int_0^t \dx{\tau} \big(\Omega_\tau(p_1)+\Omega_\tau(p_2)-\Omega_\tau(p_3)\big)}\\
			& \frac1{2!}\bbf{1}{2}_t(\bp_3)\ad_{p_1}\ad_{p_2}a_{p_3}\\
			& + \, h.c.\Big) \, .
		\end{align}
		Using \eqref{eq-ada-dc-0}, we hence have that
		\begin{align}
		\MoveEqLeft \frac1N\int_0^t\dx{s}Q_3[f_0](s,p) \, = \, \\
		& - \, \int_{[0,t]^2} \dx{\bs_2}  \mathds{1}_{s_1\geq s_2} \Big(\frac{\jb{\wick{[[\c1 \ad_p \c2 a_p, \settowidth{\wdth}{$\Hcub$}\hspace{.25\wdth}\c3{\vphantom{\Hcub}}\hspace{-.25\wdth}\c1 \Hcub\settowidth{\wdth}{$\Hcub$}\hspace{-.25\wdth}\c1{\vphantom{\Hcub}}\hspace{.25\wdth}(s_1)], \settowidth{\wdth}{$\Hcub$}\hspace{.25\wdth}\c1{\vphantom{\Hcub}}\hspace{-.25\wdth} \c2 \Hcub\settowidth{\wdth}{$\Hcub$}\hspace{-.25\wdth} \c3{\vphantom{\Hcub}}\hspace{.25\wdth}( s_2)]}}_0}\vol \label{eq-cub-bol-con-1}\\
		& + \, \frac{\jb{\wick{[[\c2 \ad_p \c1 a_p, \settowidth{\wdth}{$\Hcub$}\hspace{.25\wdth}\c3{\vphantom{\Hcub}}\hspace{-.25\wdth}\c1 \Hcub\settowidth{\wdth}{$\Hcub$}\hspace{-.25\wdth}\c1{\vphantom{\Hcub}}\hspace{.25\wdth}(s_1)], \settowidth{\wdth}{$\Hcub$}\hspace{.25\wdth}\c1{\vphantom{\Hcub}}\hspace{-.25\wdth} \c2 \Hcub\settowidth{\wdth}{$\Hcub$}\hspace{-.25\wdth} \c3{\vphantom{\Hcub}}\hspace{.25\wdth}( s_2)]}}_0}\vol\Big) \, = \label{eq-cub-bol-con-2}\\
		& -\frac{2\lambda^2}{N\vol}\Re \int_{[0,t]^2} \dx{\bs_2}  \mathds{1}_{s_1\geq s_2}\int \dx{\bp_3} d\bk_3 \, \Big[\jb{\wick{[[\c1 \ad_p \c2 a_p,\c3 \ad_{p_1}\c4 \ad_{p_2}\c1 a_{p_3}],\c2\ad_{k_3}\c4 a_{k_2}\c3 a_{k_1}]}}_0 \\
		&\frac1{(2!)^2}\bbf{1}{2}_{s_1}(\bp_3)\bbfb{1}{2}_{s_2}(\bk_3)\delta(p_1+p_2-p_3)\delta(k_1+k_2-k_3) \\
		& e^{i\int_{0}^{s_1} \dx{\tau} \big(\Omega_\tau(p_1)+\Omega_\tau(p_2)-\Omega_\tau(p_3)\big)-i\int_{0}^{s_2} \dx{\tau} \big(\Omega_\tau(k_1)+\Omega_\tau(k_2)-\Omega_\tau(k_3)\big)}\\
		& + \, \jb{\wick{[[\c1 \ad_p \c2 a_p,\c2 \ad_{p_1}\c3 \ad_{p_2}\c4 \ad_{p_3}],\c1 a_{k_1}\c3 a_{k_2}\c4 a_{k_3}]}}_0\\
		&\frac1{(3!)^2}\bbf{0}{3}_{s_1}(\bp_3)\bbfb{0}{3}_{s_2}(\bk_3)\delta(p_1+p_2+p_3)\delta(k_1+k_2+k_3) \\
		& e^{i\int_0^{s_1} \dx{\tau} \big(\Omega_\tau(p_1)+\Omega_\tau(p_2)+\Omega_\tau(p_3)\big)-i\int_0^{s_2} \dx{\tau} \big(\Omega_\tau(k_1)+\Omega_\tau(k_2)+\Omega_\tau(k_3)\big)}\\
		&+ \, \mbox{all\ contractions\ of\ the\ form \eqref{eq-cub-bol-con-1}--\eqref{eq-cub-bol-con-2}}\Big] \, . \label{eq-f-main-0}
		\end{align} 
		We recognize the collision kernels $\bbf{1}{2}$ and $\bbf{0}{3}$. Now observe that
		\begin{align}
			\jb{\wick{[[\c1 \ad_p \c2 a_p, A_1 \c2 \ad_{p_j} A_2 ], B_1 \c1 a_{k_\ell}B_2]}}_0 =&  \delta(p-p_j)\big(\fbar(p)-f_0(p)\big)\jb{\wick{[[A_1 \c1 \ad_{p_j} A_2 ], B_1 \c1 a_{k_\ell}B_2]}}_0\\
			=&  \delta(p-p_j)\jb{\wick{[[A_1 \c1 \ad_{p_j} A_2 ], B_1 \c1 a_{k_\ell}B_2]}}_0 \, , \\
			\jb{\wick{[[\c1 \ad_p \c2 a_p, A_1 \c1 a_{p_j} A_2 ], B_1 \c2 \ad_{k_\ell}B_2]}}_0 =&  -\delta(p-p_j)\jb{\wick{[[A_1 \c1 \ad_{p_j} A_2 ], B_1 \c1 a_{k_\ell}B_2]}}_0 \, .
		\end{align} 
		In particular, we have that
		\begin{align}
			\MoveEqLeft \jb{\wick{[[\c1 \ad_p \c2 a_p,\c2 \ad_{p_1}\c3 \ad_{p_2}\c4 \ad_{p_3}],\c1 a_{k_1}\c3 a_{k_2}\c4 a_{k_3}]}}_0\\
			& + \, \mbox{all\ contractions\ of\ the\ form \eqref{eq-cub-bol-con-1}--\eqref{eq-cub-bol-con-2}}\\
			=&  \big(\delta(p_1-p)+\delta(p_2-p)-\delta(p_3-p)\big)\\
			& \Big(\sum_{\pi\in \cS_2}\delta(p_1-k_{\pi(1)})\delta(p_2-k_{\pi(2)})\Big)\delta(p_3-k_3)\\
			&\big(f_0(p_1)f_0(p_2)\fbar(p_3)-\fbar(p_1)\fbar(p_2)f_0(p_3)\big) \, ,
		\end{align} 
		and similarly
		\begin{align}
		\MoveEqLeft \jb{\wick{[[\c1 \ad_p \c2 a_p,\c3 \ad_{p_1}\c4 \ad_{p_2}\c1 a_{p_3}],\c2\ad_{k_3}\c4 a_{k_2}\c3 a_{k_1}]}}_0\\
		& + \, \mbox{all\ contractions\ of\ the\ form \eqref{eq-cub-bol-con-1}--\eqref{eq-cub-bol-con-2}}\\
		=& \big(\delta(p_1-p)+\delta(p_2-p)+\delta(p_3-p)\big)\\
		& \sum_{\pi\in \cS_3}\delta(p_1-k_{\pi(1)})\delta(p_2-k_{\pi(2)})\delta(p_3-k_{\pi(3)})\\
		&\big(f_0(p_1)f_0(p_2)f_0(p_3)-\fbar(p_1)\fbar(p_2)\fbar(p_3)\big) \, .
		\end{align} 
		Exploiting the symmetry of $\bbf{1}{2}_s(\bk_3)$ w.r.t. permutations of $(k_1,k_2)$ and the symmetry of $\bbf{0}{3}_s(\bk_3)$ w.r.t. permutations of $(k_1,k_2,k_3)$, \eqref{eq-f-main-0} yields
		\begin{align}
			\MoveEqLeft \frac1N\int_0^t\dx{s}Q_3[f_0](s,p) \, = \, \\
			& \frac{2\lambda^2}{N}\Re \int_{[0,t]^2} \dx{\bs_2}  \mathds{1}_{s_1\geq s_2}\int \dx{\bp_3}  \Big(\frac1{2!}\big(\delta(p_1-p)+\delta(p_2-p)-\delta(p_3-p)\big) \\
			& \bbf{1}{2}_{s_1}(\bp_3)\bbfb{1}{2}_{s_2}(\bp_3)e^{i\int_{s_2}^{s_1} \dx{\tau} \big(\Omega_\tau(p_1)+\Omega_\tau(p_2)-\Omega_\tau(p_3)\big)}\delta(p_1+p_2-p_3)\\
			&\big(\fbar(p_1)\fbar(p_2)f_0(p_3)-f_0(p_1)f_0(p_2)\fbar(p_3)\big) \\
			& + \, \frac1{3!}\big(\delta(p_1-p)+\delta(p_2-p)+\delta(p_3-p)\big)\\
			& \bbf{0}{3}_{s_1}(\bp_3)\bbfb{0}{3}_{s_2}(\bp_3)e^{i\int_{s_2}^{s_1} \dx{\tau} \big(\Omega_\tau(p_1)+\Omega_\tau(p_2)+\Omega_\tau(p_3)\big)}\delta(p_1+p_2+p_3) \\
			&\big(\fbar(p_1)\fbar(p_2)\fbar(p_3)-f_0(p_1)f_0(p_2)f_0(p_3)\big)\Big) \, .
		\end{align} 
		With analogous steps, one obtains
		\begin{align}
			\MoveEqLeft \frac1{N^2}\int_0^t\dx{s} Q_4[f_0](s,p) \, = \, \\
			& \frac{2\lambda^2}{N^2}\Re \int_{[0,t]^2} \dx{\bs_2}  \mathds{1}_{s_1\geq s_2} \int \dx{\bp_4}  \Big(\frac1{(2!)^2}\bbf{2}{2}_{s_1}(\bp_4)\bbfb{2}{2}_{s_2}(\bp_4)\\
			& \big(\delta(p_1-p)+\delta(p_2-p)-\delta(p_3-p)-\delta(p_4-p)\big)\\
			& \delta(p_1+p_2-p_3-p_4)e^{i\int_{s_2}^{s_1} \dx{\tau} \big(\Omega_\tau(p_1)+\Omega_\tau(p_2)-\Omega_\tau(p_3)-\Omega_\tau(p_4)\big)}\\
			&\big(\fbar(p_1)\fbar(p_2)f_0(p_3)f_0(p_4)-f_0(p_1)f_0(p_2)\fbar(p_3)\fbar(p_4)\big) \\
			& + \, \frac1{3!}\bbf{1}{3}_{s_1}(\bp_4)\bbfb{1}{3}_{s_2}(\bp_4)\\
			& \big(\delta(p_1-p)+\delta(p_2-p)+\delta(p_3-p)-\delta(p_4-p)\big)\\
			& \delta(p_1+p_2+p_3-p_4)e^{i\int_{s_2}^{s_1} \dx{\tau} \big(\Omega_\tau(p_1)+\Omega_\tau(p_2)+\Omega_\tau(p_3)-\Omega_\tau(p_4)\big)}\\
			&\big(\fbar(p_1)\fbar(p_2)\fbar(p_3)f_0(p_4)-f_0(p_1)f_0(p_2)f_0(p_3)\fbar(p_4)\big)\\
			& + \, \frac1{4!}\bbf{0}{4}_{s_1}(\bp_4)\bbfb{0}{4}_{s_2}(\bp_4)\\
			& \big(\delta(p_1-p)+\delta(p_2-p)+\delta(p_3-p)+\delta(p_4-p)\big)\\
			& \delta(p_1+p_2+p_3+p_4)e^{i\int_{s_2}^{s_1} \dx{\tau} \big(\Omega_\tau(p_1)+\Omega_\tau(p_2)+\Omega_\tau(p_3)+\Omega_\tau(p_4)\big)}\\
			&\big(\fbar(p_1)\fbar(p_2)\fbar(p_3)\fbar(p_4)-f_0(p_1)f_0(p_2)f_0(p_3)f_0(p_4)\big)\Big) \, .
		\end{align} 
		This concludes the proof.
	\endprf 

    \begin{lemma}[Calculation of collision operators for $(\Phi,g)$]\label{lem-cub-bol-phi-g}
		Recall from \eqref{def-bpb} and \eqref{def-bprev} that $\bprev_3=(p_3,p_2,p_1)$ and $\bpb=(p_1,p_2,-p_3)$. Using the notation in Lemma \ref{lem-hfluc}, the cubic Boltzmann operator for $\Phi$ is given by
		\begin{align}
		    \MoveEqLeft \int_0^t \dx{s} Q_3^{(\Phi)}[f_0](s) \, = \\
            &\lambda^2\int_{[0,t]^2}\dx{\bs_2} \mathds{1}_{s_1\geq s_2} e^{i\int_0^{s_1}\dx{\tau} \Omega_\tau(0)}\Big[\frac12 \delta(p_1+p_2-p_3)\Big(e^{i\int_{s_2}^{s_1}\dx{\tau}(\Omega_\tau(p_1)+\Omega_\tau(p_2)-\Omega_\tau(p_3))}\\
            & \quad\quad \bbf{1}{3}_{s_1}(0,\bp_3)\bbfb{1}{2}_{s_2}(\bp_3) \, - \, e^{-i\int_{s_2}^{s_1}\dx{\tau}(\Omega_\tau(p_1)+\Omega_\tau(p_2)-\Omega_\tau(p_3))}\bbf{2}{2}_{s_1}(0,\bprev_3)\bbf{1}{2}_{s_2}(\bprev_3)\Big)\\
            & \qquad \big(f_0(p_1)f_0(p_2)\fbar(p_3)-\fbar(p_1)\fbar(p_2)f_0(p_3)\big)\\
            & \quad + \, \frac1{3!}\delta(p_1+p_2+p_3)\Big(\bbf{0}{4}_{s_1}(0,\bp_3)\bbfb{0}{3}_{s_2}(\bp_3)e^{i\int_{s_2}^{s_1}\dx{\tau}(\Omega_\tau(p_1)+\Omega_\tau(p_2)+\Omega_\tau(p_3))}\\
            & \quad\quad - \, \bbfb{1}{3}_{s_1}(\bp_3,0)\bbf{0}{3}_{s_2}(\bp_3) e^{-i\int_{s_2}^{s_1}\dx{\tau}(\Omega_\tau(p_1)+\Omega_\tau(p_2)+\Omega_\tau(p_3))}\Big)\\
            & \qquad \big(f_0(p_1)f_0(p_2)f_0(p_3)-\fbar(p_1)\fbar(p_2)\fbar(p_3)\big)\Big] \, ,
		\end{align}
        and for $g$, it is given by
        \begin{align}
            \MoveEqLeft \int_0^t \dx{s} Q_3^{(g)}[f_0](s)[J] \, =\\
            &\lambda^2\int\dx{p}J(p)\int_{[0,t]^2}\dx{\bs_2} \mathds{1}_{s_1\geq s_2} \int \dx{\bp_3}\Big[\delta(p_1+p_2-p_3)\Big(\delta(p-p_3)e^{2i\int_0^{s_1}\dx{\tau}\Omega_\tau(p_3)}\\
        & \quad e^{i\int_{s_2}^{s_1}\dx{\tau}(\Omega_\tau(p_1)+\Omega_\tau(p_2)-\Omega_\tau(p_3))} \bbf{0}{3}_{s_1}(\bpb_3) \bbfb{1}{2}_{s_2}(\bp_3) -2\delta(p-p_1)e^{-2i\int_0^{s_1}\dx{\tau}\Omega_\tau(p_1)}\\
        & \quad e^{-i\int_{s_2}^{s_1}\dx{\tau}(\Omega_\tau(p_1)+\Omega_\tau(p_2)-\Omega_\tau(p_3))}\bbfb{1}{2}_{s_1}(\bpb_3)\bbf{1}{2}_{s_2}(\bp_3)\Big)\\
        & \quad \big(f_0(p_1)f_0(p_2)\fbar(p_3)-\fbar(p_1)\fbar(p_2)f_0(p_3)\big)\\
        & + \, \delta(p-p_3)e^{2i\int_0^{s_1}\dx{\tau}\Omega_\tau(p_3)}\delta(p_1+p_2+p_3)e^{-i\int_{s_2}^{s_1}\dx{\tau}(\Omega_\tau(p_1)+\Omega_\tau(p_2)+\Omega_\tau(p_3))} \\
        & \quad \bbfb{1}{2}_{s_1}(\bpb_3)\bbf{0}{3}_{s_2}(\bp_3)\big(\fbar(p_1)\fbar(p_2)\fbar(p_3)-f_0(p_1)f_0(p_2)f_0(p_3)\big)\Big] \, .
    \end{align}
    \end{lemma}
    \begin{proof}
        With analogous calculations as in Lemma \ref{lem-cub-bol}, we obtain 
        \begin{align}
        \MoveEqLeft \frac1{N^{\frac32}}\int_0^t \dx{s} Q_3^{(\Phi)}[f_0](s) \, = \, -\frac1\vol\int_{[0,t]^2}\dx{\bs_2} \mathds{1}_{s_1\geq s_2} \jb{[[a_0,\wick{
        \settowidth{\wdth}{$\Hquart$}\hspace{.25\wdth}\c1{\vphantom{\Hquart}}\hspace{.25\wdth}\c2{\vphantom{\Hquart}}\hspace{.25\wdth}\c3{\vphantom{\Hquart}}\hspace{-.75\wdth}\Hquart(s_1)],\settowidth{\wdth}{$\Hcub$}\hspace{.25\wdth}\c1{\vphantom{\Hcub}}\hspace{.25\wdth}\c2{\vphantom{\Hcub}}\hspace{.25\wdth}\c3{\vphantom{\Hcub}}\hspace{-.75\wdth}\Hcub
        }
        (s_2)]}_0 \\
        = \, & \frac{\lambda^2}{N^{\frac32}} \int_{[0,t]^2}\dx{\bs_2} \mathds{1}_{s_1\geq s_2} e^{i\int_0^{s_1}\dx{\tau} \Omega_\tau(0)}\Big[\frac12 \delta(p_1+p_2-p_3)e^{i\int_{s_2}^{s_1}\dx{\tau}(\Omega_\tau(p_1)+\Omega_\tau(p_2)-\Omega_\tau(p_3))}\\
        & \quad \bbf{1}{3}_{s_1}(0,\bp_3)\bbfb{1}{2}_{s_2}(\bp_3)\big(f_0(p_1)f_0(p_2)\fbar(p_3)-\fbar(p_1)\fbar(p_2)f_0(p_3)\big)\\
        & \quad + \, \frac1{3!}\delta(p_1+p_2+p_3)\Big(\bbf{0}{4}_{s_1}(0,\bp_3)\bbfb{0}{3}_{s_2}(\bp_3)e^{i\int_{s_2}^{s_1}\dx{\tau}(\Omega_\tau(p_1)+\Omega_\tau(p_2)+\Omega_\tau(p_3))}\\
        & \quad - \, \bbfb{1}{3}_{s_1}(\bp_3,0)\bbf{0}{3}_{s_2}(\bp_3) e^{-i\int_{s_2}^{s_1}\dx{\tau}(\Omega_\tau(p_1)+\Omega_\tau(p_2)+\Omega_\tau(p_3))}\Big)\\
        & \qquad \big(f_0(p_1)f_0(p_2)f_0(p_3)-\fbar(p_1)\fbar(p_2)\fbar(p_3)\big)\Big] \, .
        \\
        &\quad + \, \frac12 \bbf{2}{2}_{s_1}(0,\bp_3)\bbf{1}{2}_{s_2}(\bp_3)\delta(p_1-p_2-p_3)e^{i\int_{s_2}^{s_1}\dx{\tau}(\Omega_\tau(p_1)-\Omega_\tau(p_2)-\Omega_\tau(p_3))}\\
        &\qquad \big(f_0(p_1)\fbar(p_2)\fbar(p_3)-\fbar(p_1)f_0(p_2)f_0(p_3)\big) \, .
    \end{align}
    In order to conclude the computation of $Q_3^{(\Phi)}[f_0](t)$, we substitute $p_1\leftrightarrow p_3$ in the last term.
    \par Similarly, we compute
    \begin{align}
        \MoveEqLeft \int_0^t \dx{s} Q_3^{(g)}[f_0](s)[J] \, = \, -\frac{N}{\vol}\int\dx{p}J(p)\int_{[0,t]^2}\dx{\bs_2} \mathds{1}_{s_1\geq s_2} \jb{[[a_pa_{-p},\wick{
        \settowidth{\wdth}{$\Hcub$}\hspace{.25\wdth}\c1{\vphantom{\Hcub}}\hspace{.25\wdth}\c2{\vphantom{\Hcub}}\hspace{.25\wdth}\c3{\vphantom{\Hcub}}\hspace{-.75\wdth}\Hcub(s_1)],\settowidth{\wdth}{$\Hcub$}\hspace{.25\wdth}\c1{\vphantom{\Hcub}}\hspace{.25\wdth}\c2{\vphantom{\Hcub}}\hspace{.25\wdth}\c3{\vphantom{\Hcub}}\hspace{-.75\wdth}\Hcub
        }
        (s_2)]}_0 \\
        = \, & \lambda^2\int\dx{p}J(p)\int_{[0,t]^2}\dx{\bs_2} \mathds{1}_{s_1\geq s_2} \int \dx{\bp_3}\Big[\delta(p-p_3)e^{2i\int_0^{s_1}\dx{\tau}\Omega_\tau(p_3)}\\
        & \Big(\delta(p_1+p_2-p_3) e^{i\int_{s_2}^{s_1}\dx{\tau}(\Omega_\tau(p_1)+\Omega_\tau(p_2)-\Omega_\tau(p_3))} \bbf{0}{3}_{s_1}(\bpb_3) \bbfb{1}{2}_{s_2}(\bp_3) \\
        & \big(f_0(p_1)f_0(p_2)\fbar(p_3)-\fbar(p_1)\fbar(p_2)f_0(p_3)\big) \\
        & + \, \delta(p_1+p_2+p_3)e^{-i\int_{s_2}^{s_1}\dx{\tau}(\Omega_\tau(p_1)+\Omega_\tau(p_2)+\Omega_\tau(p_3))}\bbfb{1}{2}_{s_1}(\bpb_3)\bbf{0}{3}_{s_2}(\bp_3) \\
        & \big(\fbar(p_1)\fbar(p_2)\fbar(p_3)-f_0(p_1)f_0(p_2)f_0(p_3)\big)\Big) \\
        & + 2\delta(p-p_1)e^{-2i\int_0^{s_1}\dx{\tau}\Omega_\tau(p_1)}\delta(p_1+p_2-p_3)e^{-i\int_{s_2}^{s_1}\dx{\tau}(\Omega_\tau(p_1)+\Omega_\tau(p_2)-\Omega_\tau(p_3))}\\
        & \bbfb{1}{2}_{s_1}(\bpb_3)\bbf{1}{2}_{s_2}(\bp_3)\big(\fbar(p_1)\fbar(p_2)f_0(p_3)-f_0(p_1)f_0(p_2)\fbar(p_3)\big) \, .
    \end{align}
    Rearranging the terms yields the result.
    \end{proof}
	
	\section{Trace estimates}

	\begin{lemma}[Number operator moments] \label{lem-num-mom}
		We have for all $\ell\in\N_0$ that
		\begin{equation}
			\jb{(\nb+\vol)^{\frac\ell2}}_0 \, \leq \, C_{\ell,\fd} \vol^{\frac\ell2} \, .  
		\end{equation}
	\end{lemma}
	
	\begin{lemma}[Operator product bound]\label{lem-prod}
		Let $A_j\in\cP[a,\ad]$ be monomials in $a,\ad$, $\gamma_j>0$, and $k_j\in\N$ be such that
		\begin{align}
		\|P_m A_j P_{m-\osg(A_j)}\|	\, \leq \, \gamma_j (m+\vol)^{k_j/2} \label{eq-Aj-bd-ass-0}
		\end{align} 
		for all $j\in\{1,\ldots,\ell\}$ and all $m\in\N_0$. Then we have that
		\begin{align}
		|\nu(\prod_{j=1}^\ell A_j)|  \, \leq \, \Big(\prod_{j=1}^\ell \gamma_j\Big) \nu\Big( \big(\cN+\sum_{m=1}^\ell|\osg(A_m)|+\vol \big)^{\sum_{j=1}^\ell k_j/2}\Big)
		\end{align} 
		for any state $\nu$. 
	\end{lemma}
	
	\begin{lemma}\label{lem-fg-bound}
		Given a test function $J\in L^2\cap L^\infty(\lattice)$, let
		\begin{align}
		f[J]  := & \int_\lattice \, dp \, J(p) \ad_pa_p \, ,\\
		g[J] := & \int_\lattice \, dp \, J(p) a_{-p}a_p \, .
		\end{align} 
		Then we have 
		\begin{align}
		\|P_m f[J]  P_n\| \leq & \delta_{m,n}\|J\|_\infty m \, , \\
		\|P_m g[J]  P_n\| \leq & \, \delta_{n,m+2}(\|J\|_2+\|J\|_\infty)(m+1+\vol) \, .
		\end{align} 
	\end{lemma}	

	For the following standard result, we need to introduce some notation. For a proof of the statement, we refer, e.g., to \cite{bachfrsi,bafrsi}. Denote
    \begin{align}
        a_p^{(1)} \, := \, \ad_p \, , \qquad a_p^{(-1)} \, := \, a_p \, .
    \end{align}
    Given a finite ordered subset $J=\{j_1<j_2<\ldots<j_r\}\subset \N$ and $\sigma_{j_k}\in\{\pm1\}$, we define the ordered product
	\begin{align}
	\prod_{j\in J} a_{p_j}^{(\sigma_j)} \, : = \, a_{p_{j_1}}^{(\sigma_{j_1})}\ldots a_{p_{j_r}}^{(\sigma_{j_r})} \, .
	\end{align}
	In addition, we abbreviate
	\begin{align}
	\bp_J \, := \, (p_{j_k})_{k=1}^r \, ,
	\end{align} 
	as well as 
	\begin{align}
	a^{(\sigma)}(\bp_J) \, := \, \prod_{j\in J} a^{(\sigma)}_{p_j} \, . \label{eq-a-set-0}
	\end{align} 
    Note that we have that
    \begin{align}\label{eq-com-sign-1}
        [\nb,\prod_{j\in J} a_{p_j}^{(\sigma_j)}] \, = \, \sum_{j\in J}\sigma_j \prod_{j\in J} a_{p_j}^{(\sigma_j)} \, .
    \end{align}
	Furthermore, we define the sets
	\begin{align}
	J_{\pm} \, := \, \{j\in J\mid \sigma_j=\pm 1\} 
	\end{align} 
	and the Wick-ordered product
	\begin{align}
	:\prod_{j\in J}a_{p_j}^{(\sigma_j)}: \, := \, \ad(\bp_{J_+})a(\bp_{J_-})
	\end{align} 	
	with all creation operators to the left, and all annihilation operators to the right. 	
	\par Finally, in order to keep track of the correct scaling, it is useful to work with the rescaled $\ell^2(\lattice)$-norm
	\begin{align}
	\|H\|_{L^2(\lattice)} \, = \, \frac{1}{\sqrt{\vol}} \|H\|_{\ell^2(\lattice)} \, .
	\end{align} 
	More generally, we also define
	\begin{align}
	\|H\|_{L^\infty_{\bp_m} L^2_{\bk_n}((\lattice)^{m+n})} \, := \, \sup_{\bp_m\in (\lattice)^m}\Big( \int_{(\lattice)^n} d\bk_n |H(\bp_m,\bk_n)|^2\Big)^{\frac12} \, , \\
	\|H\|_{ L^2_{\bk_n}L^\infty_{\bp_m}((\lattice)^{m+n})} \, := \, \Big( \int_{(\lattice)^n} d\bk_n \sup_{\bp_m\in (\lattice)^m}|H(\bp_m,\bk_n)|^2\Big)^{\frac12} \, ,
	\end{align} 
	where in the case $n=0$, this norm reduces to $\|H\|_{L^\infty_{\bp_m}((\lattice)^m)}$, and in the case $m=0$, to $\|H\|_{L^2_{\bk_n}((\lattice)^n)}$.
	
	\begin{lemma}[Wick-ordered operator bound] \label{lem-wick-ord-bd}
		Let $M\in\N_0$, $n\in\N$, $J:=\{1,\ldots,n\}$, $\sigma_j\in\{\pm1\}$ for all $j\in J$. Let $H:(\lattice)^n\to \C$, and $g_j:\lattice\to\C$ be given functions. Then the following holds true
		\begin{enumerate}
			\item If $J_\pm\neq\emptyset$, we have that
			\begin{fleqn}[\parindent]
				\begin{align}
				\MoveEqLeft \Big\|\int_{(\lattice)^n} d\bp_n H(\bp_n)\delta(\sum_{j=1}^np_j\sigma_j) \prod_{j=1}^ng_j(p_j):\prod_{j=1}^na_{p_j}^{(\sigma_j)} :P_M \Big\|\\
				\leq & \, \| |H|^{\frac12}  \prod_{j\in J_-} g_j(p_j) \delta(\sum_{j=1}^n\sigma_j p_j)^{\frac12}\|_{L^\infty_{\bp_{J_+}}L^2_{\bp_{J_-}}}\| |H|^{\frac12}  \prod_{j\in J_+} g_j(p_j) \delta(\sum_{j=1}^n\sigma_j p_j)^{\frac12}\|_{L^\infty_{\bp_{J_-}}L^2_{\bp_{J_+}}}\\
				&(M+\sum_{j=1}^n\sigma_j)_{|J_+|}^{\frac12}(M)_{|J_-|}^{\frac12}\mathds{1}_{M\geq |J_-|} \, , 
				\end{align} 
				where
				\begin{align}
				(x)_m \, := \, \prod_{k=0}^{m-1}(x-k) \label{def-falling-fac}
				\end{align} 
				denotes the falling factorial. 
				\item If $J_+=\emptyset$ and $n\geq2$, we find that 
				\begin{align}
				\MoveEqLeft \Big\|\int_{(\lattice)^n} d\bp_n H(\bp_n)\delta(\sum_{j=1}^np_j) a(\bp_n) P_M \Big\|\\
				\leq & \,  \big(\|H\|_{L^2_{\bp_{n-2}L^\infty_{p_{n-1},p_n}}}^2(M-n+1)+\vol \|\delta(\sum_{j=1}^np_j)^{\frac12} H\|_{L^2_{\bp_n}}^2\big)^{\frac12} (M)_{n-1}^{\frac12} \, .
				\end{align} 
			\end{fleqn}
			Similarly, in the case $J_-=\emptyset$, we have that 
			\begin{align}
			\MoveEqLeft \Big\|\int_{(\lattice)^n} d\bp_n \, H(\bp_n) \delta(\sum_{j=1}^np_j) \ad(\bp_n) P_M\Big\| \\
			\leq & \big(\|H\|_{L^2_{\bp_{n-2}}L^\infty_{p_{n-1},p_n}}^2(M+1)+\vol \|\delta(\sum_{j=1}^np_j)^{\frac12} H\|_{L^2_{\bp_n}}^2\big)^{\frac12} (M+n)_{n-1}^{\frac12} \, .
			\end{align} 
			\item If $n=1$, we obtain
			\begin{align}
			\|a_0 P_M\| \, = \, \sqrt{M\vol} \, .
			\end{align} 
		\end{enumerate}
	\end{lemma}

    \section{Bounds on fluctuation dynamics}

 
	\begin{lemma}[Bogoliubov dynamics]\label{lem-bog-dynamics}
        Assume that $\vol\geq1$. For any $t\geq0$, let $u_t$, $v_t$ be defined as in \eqref{def-up}, such that $\gamma_t:=|v_t|^2\in L^1\cap L^\infty(\lattice)$. Then, for any $\ell\in\N$, there exist constants $C_\ell>0$ such that for any $t>0$ we have that
        \begin{align}
            \big\|(\nb+\vol)^{\frac\ell2} \bog[k_0](\nb+\vol)^{-\frac\ell2}\big\| \, &\leq \, C_\ell(1+\|\gamma_0\|_1+\|\gamma_0\|_\infty)^{\frac\ell2} \, ,\\
            \big\|(\nb+\vol)^{\frac\ell2} \bogd[k_t](\nb+\vol)^{-\frac\ell2}\big\| \, &\leq \, C_\ell(1+\|\gamma_t\|_1+\|\gamma_t\|_\infty)^{\frac\ell2} \, .
        \end{align}
    \end{lemma}
    \begin{proof}
        Recall the facts that $u_t(p)^2=1+|v_t(p)|^2=1+\gamma_t$, see \eqref{def-gamma}, and that $\gamma_t$ is even. Employing \eqref{eq-bog-rot} and Lemma \ref{lem-fg-bound}, we obtain
        \begin{align}\label{eq-nb-bog-transform-0}
            \begin{split}
            \bogd[k_0]\nb \bog[k_0] \, &= \, \int_{\lattice} \dx{p} (u_0(p)\ad_p +\overline{v}_0(p)a_{-p})(u_0(p)a_p+v_0(p)\ad_{-p}) \\
            &= \, \nb + \int_{\lattice} \dx{p} \big(2\gamma_0 (p) \ad_p a_p +(\sigma_0(p)\ad_p\ad_{-p}+\mathrm{h.c.}) + \, \vol \gamma_0(p)\big) \, .
            \end{split}
        \end{align}
        Lemma \ref{lem-fg-bound} implies
        \begin{align}
            \|P_m f[1+2\gamma_0]P_n\| \, &\leq \, \delta_{m,n} (1+2\|\gamma_0\|_\infty )m \, , \label{eq-fgamma-bd}\\
            \|P_mg[\overline{\sigma}_0]P_n\| \, & \leq \, \delta_{n,m+2}(\|\sigma_0\|_2+\|\sigma_0\|_\infty)(m+1+\vol) \, . \label{eq-gsigma-bd-pre}
        \end{align}
        Due to $|\sigma|^2=\gamma(\gamma+1)$, see \eqref{eq-gamma-sigma-rel-0}, we have
        \begin{align}
            \|\sigma\|_2 \, &\leq \, \|\gamma\|_1^{\frac12}(1+\|\gamma\|_\infty)^{\frac12} \, \leq \, \frac12(1+\|\gamma\|_1+\|\gamma\|_\infty) \, ,\\
            \|\sigma\|_\infty \, & \leq \, \|\gamma\|_\infty^{\frac12}(1+\|\gamma\|_\infty)^{\frac12} \, \leq\, \frac12(1+2\|\gamma\|_\infty) \, .
        \end{align}
        As a consequence, \eqref{eq-gsigma-bd-pre} implies
        \begin{equation}\label{eq-gsigma-bd}
            \|P_mg[\overline{\sigma}_0]P_n\| \, \leq \, \frac{\delta_{n,m+2}}2(2+\|\gamma_0\|_1+3\|\gamma_0\|_\infty)(m+1+\vol) \, .
        \end{equation}
        Using $\vol\geq1$, Lemma \ref{lem-prod} then yields
        \begin{align}
            \scp{\psi}{\bogd[k_0](\nb+\vol)^\ell \bog[k_0]\psi} \, &= \, \scp{\psi}{\big(\bogd[k_0]\nb\bog[k_0]+\vol\big)^\ell \psi}\\
            &\lesssim_\ell \,  (1+\|\gamma_0\|_1+\|\gamma_0\|_\infty)^\ell\scp{\psi}{\big(\nb+\vol+\ell\big)^\ell \psi}\\
            &\lesssim_\ell \,  (1+\|\gamma_0\|_1+\|\gamma_0\|_\infty)^\ell\scp{\psi}{(\nb+\vol)^\ell \psi} \, .
        \end{align}
        With analogous steps, we obtain
        \begin{align}
            \scp{\psi}{\bog[k_t](\nb+\vol)^\ell \bogd[k_t]\psi} \, \lesssim_\ell \, (1+\|\gamma_t\|_1+\|\gamma_t\|_\infty)^\ell\scp{\psi}{(\nb+\vol)^\ell \psi} \, .
        \end{align}
        This concludes the proof.
    \end{proof}

	\begin{lemma}[BEC fluctuation dynamics]\label{lem-weyl-fluct-dynamics}
		Assume $\vol\geq1$. Let $\phi\in L^\infty(0,t)$ for all $t>0$. Define $\whfluc(t):=\weyld[\sqrt{N\vol}\phi_t] e^{-it\cH_N}\weyl[\sqrt{N\vol}\phi_0]$. Then there are $C_\ell, K_\ell>0$ such that
		\begin{align}
		\MoveEqLeft\Big\|(\nb+\vol)^{\frac\ell2} \whfluc(t)(\nb+\vol)^{-\frac\ell2}\big(1+\frac{\nb}{N\vol}\big)^{-\frac12}\Big\| \\
        & \leq \, C_\ell \Big(1+\frac{(\|\gtr_0\|_1+\|\gtr_0\|_\infty+1)^2}N\Big)^\ell e^{K_\ell\vd\lambda\vol(1+\frac{\|\gtr_0\|_1}N)t}
		\end{align}
		for all $\ell\in\N$.
	\end{lemma}
	
	\prf
	We need to reprove this statement as, compared to \cite{chenhott}, here, $\hat{v}(0)\neq0$ and $\phi$ is not stationary. Nonetheless, we follow the steps of the proof in \cite{busasc}. We show the statement by induction on $\ell$. Let $\psi\in\cF$ be arbitrary with $\|\psi\|=1$. 
    \begin{enumerate}
        \item \underline{Induction basis:} Define
    \begin{align}
        \HBECt \, & := \, \sqrt{N\vol}\ad_0 (-i\partial_t\phi_t +\lambda\vol\hat{v}(0)|\phi_t|^2\phi_t)\, + \, \mathrm{h.c.}\\
        \HHFBodt \, & := \, \frac{\lambda\vol}2 \int \dx{p} \hat{v}(p) \big(\phi_t^2\ad_p\ad_{-p}+\overline{\phi}_t^2a_pa_{-p}\big)\\
		\Hcubt \, & := \, \frac{\lambda\sqrt{\vol}}{\sqrt{N}} \int \dx{\bp_2} \hat{v}(p_2)\big(\ad_{p_1}\ad_{p_2}a_{p_{12}} \phi_t \, + \, \ad_{p_{12}}a_{p_2}a_{p_1}\overline{\phi}_t\big) . 
    \end{align}
    Substituting \eqref{eq-HFB-system} and recalling definitions \eqref{def-gtr-str}, we find that
    \begin{align}
        \HBECt \, =& \, -\frac{\lambda\sqrt{\vol}}{\sqrt{N}} \ad_0\Big( \big(\gtr_t*(\hat{v}+\hat{v}(0))\big) (0) \phi_t + (\str_t*\hat{v}) (0) \phi_t\Big) + \mathrm{h.c.}
    \end{align}
    Then \eqref{eq-har-dyn-0} and \eqref{eq-weyl-HN-0} imply that
	\begin{align}\label{eq-dt-nbt-0}
        \begin{split}
	\MoveEqLeft i\partial_t \scp{\whfluc(t)\psi}{(\nb+\vol)\whfluc(t)\psi}\\
	&= \scp{\whfluc(t)\psi}{[\nb,\HBECt+\HHFBodt + \Hcubt]\whfluc(t)\psi} \, . 
        \end{split}
	\end{align}
	Recall from \eqref{eq-com-sign-1} that for any monomial $A\in\cP[a,\ad]$  we have
	\begin{align} 
	[\nb,A] =&  \osg(A) A \, .
	\end{align} 
	Employing Lemmata \ref{lem-fg-bound} and \ref{lem-wick-ord-bd} together with Young's inequality, \eqref{eq-dt-nbt-0} yields
	\begin{align}
	\MoveEqLeft|\partial_t\scp{\whfluc(t)\psi}{(\nb+\vol)\whfluc(t)\psi}|\\
	\leq\, & C(\|\hat{v}\|_1+\|\hat{v}\|_\infty)\lambda\Big(\frac{\vol}{\sqrt{N}}(\|\gtr_t\|_1+\|\str_t\|_2+\|\str_t\|_\infty)|\phi_t|\\
    & \scp{\whfluc(t)\psi}{(\nb+\vol)^{\frac12}\whfluc(t)\psi}\\
    & + \, |\phi_t|^2\vol\scp{\whfluc(t)\psi}{(\nb+\vol)\whfluc(t)\psi} \\
	& + \, \frac{|\phi_t|\sqrt{\vol}}{\sqrt{N}}\scp{\whfluc(t)\psi}{(\nb+\vol)^{\frac32}\whfluc(t)\psi} \Big) \\
    \leq \, & C(\|\hat{v}\|_1+\|\hat{v}\|_\infty)\lambda\vol \Big(\frac{(\|\gtr_t\|_1+\|\gtr_t\|_\infty)^2}N\\
    & + \, |\phi_t|^2\scp{\whfluc(t)\psi}{(\nb+\vol)\whfluc(t)\psi} \\
	& + \, \frac{1}{N\vol}\scp{\whfluc(t)\psi}{(\nb+\vol)^2\whfluc(t)\psi} \Big) \, , \label{eq-nb-gron-ib-0}
    \end{align}
  where in the last step, we used Cauchy-Schwarz together with Lemma \ref{lem-mass-transfer}. 
	Our goal now is to bound $\frac{1}{N}\scp{\whfluc(t)\psi}{(\nb+\vol)^2\whfluc(t)\psi}$ in terms of $\scp{\whfluc(t)\psi}{(\nb+\vol)\whfluc(t)\psi}$ and time-dependent $\psi$-independent terms, in order to close the estimate. Using
	\begin{align}
	\cW[\sqrt{N\vol}\phi_t]a_p \weyld[\sqrt{N\vol}\phi_t] \, = \, a_p \, - \, \sqrt{N\vol}\phi_t\delta(p) \, , \label{eq-weyl-0}
	\end{align}  
	we derive that
	\begin{align}
	[\nb,\weyld[\sqrt{N\vol}\phi_t]] \, = \, &   -\sqrt{N\vol}\weyld[\sqrt{N\vol}\phi_t](\overline{\phi}_ta_0+\phi_t\ad_0)\\
    &+N\vol \weyld[\sqrt{N\vol}\phi_t] \\
	= \, &  -\big(\sqrt{N\vol}(\overline{\phi}_ta_0+\phi_t \ad_0)+N\vol \big)\weyld[\sqrt{N\vol}\phi_t] \, , \label{eq-nbw-hc-com}\\
	[\nb,\cW[\sqrt{N\vol}\phi_0]] = \, &  \cW[\sqrt{N\vol}\phi_0]\big(\sqrt{N}(a_0+\ad_0)+N\vol\big) \, .\label{eq-nbw-com}
	\end{align} 
	From these identities and using $[\nb,\cH_N]=0$ and $\phi_0=\vol^{-\frac12}$, we obtain that
	\begin{align}
	[\nb,\whfluc(t)] =&  [\nb,\weyld[\sqrt{N\vol}\phi_t]]e^{-it\cH_N}\cW[\sqrt{N\vol}\phi_0] \\
	&+ \, \weyld[\sqrt{N\vol}\phi_0]e^{-it\cH_N}[\nb,\cW[\sqrt{N\vol}\phi_0]] \\
	=&  -\sqrt{N\vol}(\overline{\phi}_ta_0+\phi_t \ad_0)\whfluc(t) \,  + \, \sqrt{N}\whfluc(t)(a_0+\ad_0) \, . \label{eq-nb-tUn-com}
	\end{align}
	As a consequence, we have that
	\begin{align}
	\MoveEqLeft \scp{\whfluc(t)\psi}{\nb^2\whfluc(t)\psi}\\
	= \, &  \scp{\nb\whfluc(t)\psi}{\whfluc(t)\nb\psi} \, + \, \sqrt{N}\scp{\nb\whfluc(t)\psi}{\whfluc(t)(a_0+\ad_0)\psi}  \,\\
	& - \, \sqrt{N\vol}\scp{\nb\whfluc(t)\psi}{(\overline{\phi}_ta_0+\phi_t\ad_0)\whfluc(t)\psi}\, .
	\end{align} 
	Using Cauchy-Schwarz and employing the fact that $\whfluc(t):\cF\to\cF$ is a unitary transformation, we thus obtain
	\begin{align}\label{eq-nb2-bd-1}
    \begin{split}
	\MoveEqLeft\scp{\whfluc(t)\psi}{\nb^2\whfluc(t)\psi}\\
	\leq \, & \|\nb\whfluc(t)\psi\|\Big(\|\whfluc(t)\nb\psi\|+\sqrt{N}\big(\sqrt{\vol}\|(\overline{\phi}_ta_0+\phi_t\ad_0)\whfluc(t)\psi\| \\
    &+ \, \|\whfluc(t)(a_0+\ad_0)\psi\|\big)\Big) \\
	\leq \, & \|\nb\whfluc(t)\psi\|\Big(\|\nb\psi\|+\sqrt{N}\big(\sqrt{\vol}|\phi_t|(\|a_0\whfluc(t)\psi\|\\
	& + \, \|\ad_0\whfluc(t)\psi\|) + \, \|a_0\psi\| \, + \, \|\ad_0\psi\|\big)\Big) \, . 
	\end{split}
    \end{align}
    Lemma \ref{lem-wick-ord-bd} implies
	\begin{align}\label{eq-a0-nb-bd-1}
    \begin{split}
	\|a_0 \psi\|^2 \, =& \, \sum_{M=0}^\infty\| P_Ma_0 \psi\|^2 \, = \, \sum_{M=1}^\infty \|a_0 P_{M-1}\psi\|^2\\
	\leq & \, \sum_{M=1}^\infty (M-1)\vol \|P_{M-1}\psi\|^2 \, = \,  \vol\|\sqrt{\nb}\psi\|^2 \, . 
    \end{split}
	\end{align} 
	Similarly, we have that
	\begin{align}
	\|\ad_0 \psi\| \, \leq \, \sqrt{\vol}\|\sqrt{\nb+1}\psi\|  \label{eq-a0+-nb-bd-1}
	\end{align} 	
	Employing \eqref{eq-a0-nb-bd-1} and \eqref{eq-a0+-nb-bd-1}, \eqref{eq-nb2-bd-1} implies
	\begin{align}
	\MoveEqLeft\scp{\whfluc(t)\psi}{\nb^2\whfluc(t)\psi}\\
	\leq & \, \|\nb\whfluc(t)\psi\|\Big(\|\nb\psi\|+2\sqrt{N\vol}\big(\sqrt{\vol}|\phi_t|\|\sqrt{\nb+1}\whfluc(t)\psi\| \\
    &\, + \, \|\sqrt{\nb+1}\psi\|\big)\Big) \, .
	\end{align}
	Young's inequality then implies
	\begin{align}
	\MoveEqLeft\scp{\whfluc(t)\psi}{\nb^2\whfluc(t)\psi}\\
	\leq \, & \frac12\scp{\whfluc(t)\psi}{\nb^2\whfluc(t)\psi} + C\Big(\scp{\psi}{\big(\nb^2+N\vol(\nb+1)\big)\psi}\\
    & \, + \, N\vol^2|\phi_t|^2\scp{\whfluc(t)\psi}{(\nb+1)\whfluc(t)\psi}\Big)
    \end{align}
	Rearranging terms, we thus obtain that 
	\begin{align}\label{eq-nb2/N-0}
        \begin{split}
	\MoveEqLeft\frac1{N\vol}\scp{\whfluc(t)\psi}{\nb^2\whfluc(t)\psi} \\
	\leq & \, C\Big(\vol|\phi_t|^2\scp{\whfluc(t)\psi}{(\nb+1)\whfluc(t)\psi}+\scp{\psi}{\big(\nb +1+\frac{\nb^2}{N\vol}\big)\psi}\Big) \, . 
    \end{split}
	\end{align}
	Substituting this into \eqref{eq-nb-gron-ib-0} and using $\vol\geq1$, we obtain that
	\begin{align}
	\MoveEqLeft|\partial_t\scp{\whfluc(t)\psi}{(\nb+\vol)\whfluc(t)\psi}|\\
	\leq \, & C\vd\lambda \vol\Big(\frac{(\|\gtr_t\|_1+\|\gtr_t\|_\infty)^2}N\\
    & + \vol|\phi_t|^2\scp{\whfluc(t)\psi}{(\nb+\vol)\whfluc(t)\psi}+\scp{\psi}{\big(\nb +\vol+\frac{\nb^2}{N\vol}\big)\psi}\Big) \, .
	\end{align}
    Lemma \ref{lem-mass-transfer} then implies
    \begin{align}
	\MoveEqLeft|\partial_t\scp{\whfluc(t)\psi}{(\nb+\vol)\whfluc(t)\psi}|\\
	\leq \, & C\vd\lambda \vol\Big(e^{C\vd(1+\frac{\|\gtr_0\|_1}N)\lambda t}\frac{(\|\gtr_0\|_1+\|\gtr_0\|_\infty+1)^2}N\\
    & + \Big(1+\frac{\|\gtr_0\|_1}N\Big)\scp{\whfluc(t)\psi}{(\nb+\vol)\whfluc(t)\psi}+\scp{\psi}{\big(\nb +\vol+\frac{\nb^2}{N\vol}\big)\psi}\Big) 
	\end{align}
    Gr\"onwall's inequality then implies
    \begin{align}
        \MoveEqLeft \scp{\whfluc(t)\psi}{(\nb+\vol)\whfluc(t)\psi} \\
        & \leq \, C_1 \Big(1+\frac{(\|\gtr_0\|_1+\|\gtr_0\|_\infty+1)^2}N\Big) e^{K_1\vd(1+\frac{\|\gtr_0\|_1}N)\lambda \vol t}\\
        & \qquad \scp{\psi}{\big(\nb +\vol+\frac{\nb^2}{N\vol}\big)\psi} \, .
    \end{align}
    

	\item \underline{Induction Step:} Assume that 
	\begin{align}\label{eq-nb-IH-1}
    \begin{split}
	\MoveEqLeft\scp{\whfluc(t)\psi}{(\nb+\vol)^j\whfluc(t)\psi}\\
	\leq  \, & C_j	\Big(1+\frac{(\|\gtr_0\|_1+\|\gtr_0\|_\infty+1)^2}N\Big)^j e^{K_j\vd\lambda\vol(1+\frac{\|\gtr_0\|_1}N)t}\\
    & \, \scp{\psi}{(\nb+\vol)^j\big(1+\frac{\nb}{N\vol}\big)\psi} 
        \end{split}
	\end{align}
	for all $1\leq j\leq \ell$ and some constants $C_j, K_j$, and any $\psi\in\cF$. We compute
	\begin{align}\label{eq-nb-IH-1-pf-1}
    \begin{split}
	\MoveEqLeft i\partial_t\scp{\whfluc(t)\psi}{(\nb+\vol)^{\ell+1}\whfluc(t)\psi}\\
	=& \, \scp{\whfluc(t)\psi}{\big[(\nb+\vol)^{\ell+1},\HBECt+\HHFBodt+\Hcubt]\whfluc(t)\psi}\\
	=& \, \sum_{j=1}^{\ell+1}\big\langle\whfluc(t)\psi,(\nb+\vol)^{j-1}\big[\nb,\HBECt+\HHFBodt+\Hcubt\big]\\
    & \quad (\nb+\vol)^{\ell+1-j}\whfluc(t)\psi\big\rangle \, . 
	\end{split}
    \end{align}
	Let
	\begin{align}
	A_{\mathrm{cub}}[\hat{v}] \, : = \, \int \dx{\bp_3} \hat{v}(p_2)\delta(p_1+p_2-p_3)\ad_{p_1}\ad_{p_2}a_{p_3} \, .
	\end{align} 
	Using \eqref{eq-com-sign-1} and recalling Lemma \ref{lem-fg-bound}, \eqref{eq-nb-IH-1-pf-1} yields
	\begin{align}
    \begin{split}
	\MoveEqLeft i\partial_t\scp{\whfluc(t)\psi}{(\nb+\vol)^{\ell+1}\whfluc(t)\psi}\\
	=& \, 2\lambda\sum_{j=1}^{\ell+1} \Im\Big(-\frac{\sqrt{\vol}}{\sqrt{N}} \big( \big(\gtr_t*(\hat{v}+\hat{v}(0))\big) (0) \phi_t + (\str_t*\hat{v}) (0) \overline{\phi}_t\big)\\
    & \qquad \scp{\whfluc(t)\psi}{(\nb+\vol)^{j-1}\ad_0(\nb+\vol)^{\ell+1-j}\whfluc(t)\psi}\\
    & \qquad + \vol\overline{\phi}_t^2\scp{\whfluc(t)\psi}{(\nb+\vol)^{j-1}g[\hat{v}](\nb+\vol)^{\ell+1-j}\whfluc(t)\psi}\\
    & \qquad - \frac{\sqrt{\vol}\phi_t}{\sqrt{N}}\scp{\whfluc(t)\psi}{(\nb+\vol)^{j-1}A_{\mathrm{cub}}[\hat{v}](\nb+\vol)^{\ell+1-j}\whfluc(t)\psi}\\
	=& \, 2\lambda \sum_{j=1}^{\ell+1}\Im\sum_{m,n=0}^\infty (m+\vol)^{j-1}(n+\vol)^{\ell+1-j}\\
    & \, \Big\langle\whfluc(t)\psi,P_m\Big(-\frac{\sqrt{\vol}}{\sqrt{N}} \big( \big(\gtr_t*(\hat{v}+\hat{v}(0))\big) (0) \phi_t \\
    & \, + (\str_t*\hat{v}) (0) \overline{\phi}_t\big)\ad_0 +\vol\overline{\phi}_t^2 g[\hat{v}]-\frac{\sqrt{\vol}\phi_t}{\sqrt{N}}A_{\mathrm{cub}}[\hat{v}]\Big)P_n\whfluc(t)\psi\Big\rangle \, . 
	\end{split}
    \end{align}
	Observe that we have
	\begin{align}
	P_mg[\hat{v}]P_n =&  P_m g[\hat{v}] P_{m+2} \delta_{n,m+2} \, , \\
	P_mA_{cub}[\hat{v}]P_n =&  P_mA_{cub}[\hat{v}] P_{m-1} \delta_{n,m-1} \, .
	\end{align} 
	Lemmata \ref{lem-fg-bound} and \ref{lem-wick-ord-bd}, and \eqref{eq-nb-IH-1-pf-2} then imply
	\begin{align}
	\MoveEqLeft|\partial_t\scp{\whfluc(t)\psi}{(\nb+\vol)^{\ell+1}\whfluc(t)\psi}|\\
	\leq & C\lambda\sum_{j=1}^{\ell+1} \sum_{m,n=0}^\infty (m+\vol)^{j-1}(n+\vol)^{\ell+1-j}\|P_m\whfluc(t)\psi\|\|P_n\whfluc(t)\psi\| \\
	&\Big(\frac{\sqrt{\vol}}{\sqrt{N}}\vd(\|\gtr_t\|_1+\|\gtr_t\|_\infty)|\phi_t|\|\ad_0 P_{m-1}\|\delta_{n,m-1}\\
    & \, +\vol|\phi_t|^2 \|g[\hat{v}] P_{m+2}\| \delta_{n,m+2} \, + \, \frac{\sqrt{\vol}|\phi_t|\|A_{cub}[\hat{v}]P_{m-1}\|}{\sqrt{N}}  \delta_{n,m-1} \Big) \\
	\leq \, & C\vd \lambda\sum_{j=1}^{\ell+1} \sum_{m,n=0}^\infty (m+\vol)^{j-1}(n+\vol)^{\ell+1-j}\\
	& \big(\|P_m\whfluc(t)\psi\|^2+\|P_n\whfluc(t)\psi\|^2\big) \\
    & \Big(\frac{\vol}{\sqrt{N}}(\|\gtr_t\|_1+\|\gtr_t\|_\infty)|\phi_t|m^{\frac12}\delta_{n,m-1} \\
    & \, + \, \vol|\phi_t|^2(m+2+\vol)\delta_{n,m+2} \, + \, \frac{\sqrt{\vol}|\phi_t|m^{\frac32}}{\sqrt{N}}  \delta_{n,m-1}\Big) \\
	\leq \, & C\vd \lambda\vol\sum_{j=1}^{\ell+1}\sum_{m=0}^\infty (m+\vol)^{j-1}(m+2+\vol)^{\ell+1-j}\|P_m\whfluc(t)\psi\|^2\\
	& \Big(\frac{(\|\gtr_t\|_1+\|\gtr_t\|_\infty)^2}N+|\phi_t|^2(m+2+\vol) + \frac{(m+1)^2}{N\vol}\Big)
	\, ,
    \end{align}
    where, in the last step, we applied Cauchy-Schwarz. We can rewrite the last inequality as
	\begin{align}
    \begin{split}
	\MoveEqLeft|\partial_t\scp{\whfluc(t)\psi}{(\nb+\vol)^{\ell+1}\whfluc(t)\psi}|\\
	\leq \, & K_\ell \vd \lambda \vol\Big(\frac{(\|\gtr_t\|_1+\|\gtr_t\|_\infty)^2}N\scp{\whfluc(t)\psi}{(\nb+\vol)^{\ell} \whfluc(t)\psi} \\
    & + \, |\phi_t|^2\scp{\whfluc(t)\psi}{(\nb+\vol)^{\ell+1} \whfluc(t)\psi}
	\\
	& + \, \frac1{N\vol}\scp{\whfluc(t)\psi}{(\nb+\vol)^{\ell+2} \whfluc(t)\psi}\Big) \, .
	\end{split}
    \end{align}
    The induction hypothesis \eqref{eq-nb-IH-1} and Lemma \ref{lem-mass-transfer} then imply
    \begin{align}\label{eq-nb-IH-1-pf-2}
        \begin{split}
            \MoveEqLeft|\partial_t\scp{\whfluc(t)\psi}{(\nb+\vol)^{\ell+1}\whfluc(t)\psi}|\\
	\leq \, & C_\ell \vd \lambda \vol\Big[\Big(1+\frac{(\|\gtr_0\|_1+\|\gtr_0\|_\infty+1)^2}N\Big)^\ell e^{K_\ell\vd\lambda\vol(1+\frac{\|\gtr_0\|_1}N)t}\\
    & \, \scp{\psi}{(\nb+\vol)^\ell\big(1+\frac{\nb}{N\vol}\big)\psi} \\
    & + \, \Big(1+\frac{\|\gtr_0\|_1}N\Big)\scp{\whfluc(t)\psi}{(\nb+\vol)^{\ell+1} \whfluc(t)\psi}
	\\
	& + \, \frac1N\scp{\whfluc(t)\psi}{(\nb+\vol)^{\ell+2} \whfluc(t)\psi}\Big] \, .
        \end{split}
    \end{align}
	We claim that 
	\begin{align}\label{eq-nb-IH-2}
    \begin{split}
	\MoveEqLeft \frac1{N\vol}\scp{\whfluc(t)\psi}{(\nb+\vol)^{j+1}\whfluc(t)\psi}\\
	\leq \,& C_j \Big(\Big(1+\frac{(\|\gtr_0\|_1+\|\gtr_0\|_\infty+1)^2}N\Big)^j e^{K_j\vd\lambda\vol(1+\frac{\|\gtr_0\|_1}N)t}\\
    & \, \scp{\psi}{(\nb+\vol)^j\big(1+\frac{\nb}{N\vol}\big)\psi} \\
	& + \, \Big(1+\frac{\|\gtr_0\|_1}N\Big)\scp{\whfluc(t)\psi}{(\nb+\vol)^j\whfluc(t)\psi}\Big) 
    \end{split}
	\end{align} 
	for all $1\leq j\leq \ell+1$ and all $\psi\in\cF$. \eqref{eq-nb-IH-2} for $j=\ell+1$ together with \eqref{eq-nb-IH-1-pf-2} then implies
	\begin{align}
	\MoveEqLeft|\partial_t\scp{\whfluc(t)\psi}{(\nb+\vol)^{\ell+1}\whfluc(t)\psi}|\\
	\leq \, & C_\ell \vd \Big(1+\frac{\|\gtr_0\|_1}N\Big)\lambda \vol \\
	& \Big( \scp{\whfluc(t)\psi}{(\nb+\vol)^{\ell+1}\whfluc(t)\psi} + e^{K_{\ell+1}\vd\lambda\vol(1+\frac{\|\gtr_0\|_1}N)t}\\
    & \, \Big(1+\frac{(\|\gtr_0\|_1+\|\gtr_0\|_\infty+1)^2}N\Big)^{\ell+1}\scp{\psi}{(\nb+\vol)^{\ell+1}\big(1+\frac{\nb}{N\vol}\big)\psi} \, . 
    \end{align}
	Gr\"onwall's Lemma then implies
	\begin{align}
	\MoveEqLeft\scp{\whfluc(t)\psi}{(\nb+\vol)^{\ell+1}\whfluc(t)\psi}|\\
	\leq \, & C_{\ell+1}e^{K_{\ell+1}\vd\lambda\vol(1+\frac{\|\gtr_0\|_1}N)t}\Big(1+\frac{(\|\gtr_0\|_1+\|\gtr_0\|_\infty+1)^2}N\Big)^{\ell+1}\\
    & \scp{\psi}{(\nb+\vol)^{\ell+1}\big(1+\frac{\nb}{N\vol}\big)\psi} \, .
    \end{align}
	Thus, proving \eqref{eq-nb-IH-2} for $j=\ell+1$ concludes the proof. We have proved \eqref{eq-nb-IH-2} for $j=1$ in Step 1, \eqref{eq-nb2/N-0}. We have that \eqref{eq-nb-IH-2} also holds for $j=0$: In fact, for
    \begin{align}
        \scp{\Psi}{(\overline{\phi}_ta_0+\phi_t\ad_0)\Psi} \, &\leq \, |\phi_t|\|\Psi\|(\|a_0\Psi\|+\|\ad_0\Psi\|) \\
        &\leq \, 2|\phi_t|\sqrt{\vol}\|\Psi\|\|\sqrt{\nb+1}\Psi\|\\
        & \leq \, \scp{\Psi}{(\nb+1+\vol|\phi_t|^2)\Psi} \, ,
    \end{align}
    due to \eqref{eq-a0-nb-bd-1} and \eqref{eq-a0+-nb-bd-1} followed by Young's inequality, we find that 
	\begin{align}
	\MoveEqLeft\cW[\sqrt{N\vol}\phi_t]\nb \weyld[\sqrt{N\vol}\phi_t]\\
	= \, &\nb-\sqrt{N\vol}(\overline{\phi}_ta_0+\phi_t\ad_0)+N\vol^2|\phi_t|^2 \\
	\leq \, &  2(\nb+\vol+N\vol^2|\phi_t|^2) \, ,
    \end{align}
	which then commutes with $e^{-it\cH_N}$. 
	\par Suppose \eqref{eq-nb-IH-2} holds up to some $1\leq j\leq\ell$. Applying \eqref{eq-nb-tUn-com}, we have that 
	\begin{align}\label{eq-nb/N-ind-0}
    \begin{split}
	\MoveEqLeft\frac1{N\vol}\scp{\whfluc(t)\psi}{(\nb+\vol)^{j+2}\whfluc(t)\psi}\\
	=& \frac1{N\vol}\scp{(\nb+\vol)^{j+1} \whfluc(t)\psi}{\whfluc(t)(\nb+\vol)\psi}\\
	& -\frac1{\sqrt{N\vol}}\scp{(\nb+\vol)^{j+1} \whfluc(t)\psi}{(\overline{\phi}_ta_0+\phi_t\ad_0)\whfluc(t)\psi} \\
	& + \,\frac1{\sqrt{N}\vol}\scp{(\nb+\vol)^{j+1} \whfluc(t)\psi}{\whfluc(t)(a_0+\ad_0)\psi} \, . 
	\end{split}
    \end{align}
	Whenever $\phi_t\neq0$, we can bound the second term by
	\begin{align}\label{eq-nb-U-a0-0}
    \begin{split}
	\MoveEqLeft\frac1{\sqrt{N\vol}}\scp{(\nb+\vol)^{j+1} \whfluc(t)\psi}{(\overline{\phi}_ta_0+\phi_t\ad_0)\whfluc(t)\psi}\\
	\leq & \alpha\vol|\phi_t|^2 \scp{ \whfluc(t)\psi}{(\nb+\vol)^{j+1}\whfluc(t)\psi} \\
	& + \, \frac{\scp{\whfluc(t)\psi}{(\overline{\phi}_ta_0+\phi_t\ad_0)(\nb+\vol)^{j+1}(\overline{\phi}_ta_0+\phi_t\ad_0)\whfluc(t)\psi}}{\alpha N\vol^2|\phi_t|^2} \, . 
	\end{split}
    \end{align}
	Employing \eqref{eq-a0-nb-bd-1} and \eqref{eq-a0+-nb-bd-1}, we find that, for any $\tilde{\psi}\in\cF$ and $k\in\N_0$,
	\begin{align}
    \begin{split}
	\MoveEqLeft\scp{\tilde{\psi}}{(\overline{\phi}_ta_0+\phi_t\ad_0)(\nb+\vol)^k(\overline{\phi}_ta_0+\phi_t\ad_0)\tilde{\psi}}\\
	\leq \, & |\phi_t|^2 (\|(\nb+\vol)^{\frac{k}2}a_0\tilde{\psi}\| + \|(\nb+\vol)^{\frac{k}2}\ad_0\tilde{\psi}\|)^2\\
	= \, & |\phi_t|^2 (\|a_0(\nb+\vol-1)^{\frac{k}2}\tilde{\psi}\| + \|\ad_0(\nb+\vol+1)^{\frac{k}2}\tilde{\psi}\|)^2\\
	\leq \, & C_k\vol|\phi_t|^2 \scp{\tilde{\psi}}{(\nb+\vol)^{k+1}\tilde{\psi}} \, ,
	\end{split}
    \end{align}
	i.e., 
	\begin{align}\label{eq-a0nba0-bd-t}
	(\overline{\phi}_ta_0+\phi_t\ad_0)(\nb+\vol)^k(\overline{\phi}_ta_0+\phi_t\ad_0)\, \leq\, C_k\vol|\phi_t|^2 (\nb+\vol)^{k+1} \, . 
	\end{align} 
    Analogously, we have that
    \begin{equation}
        \label{eq-a0nba0-bd-0}
	(a_0+\ad_0)(\nb+\vol)^k(a_0+\ad_0)\, \leq\, C_k\vol (\nb+\vol)^{k+1} \, . 
    \end{equation}
	Employing \eqref{eq-a0nba0-bd-t} and choosing $\alpha>0$ sufficiently large, \eqref{eq-nb-U-a0-0} implies
	\begin{align}\label{eq-nb/N-ind-0-2}
    \begin{split}
	\MoveEqLeft\frac1{\sqrt{N\vol}}\scp{(\nb+\vol)^{j+1} \whfluc(t)\psi}{(\overline{\phi}_ta_0+\phi_t\ad_0)\whfluc(t)\psi}\\
	\leq &C_j\vol|\phi_t|^2\scp{ \whfluc(t)\psi}{(\nb+\vol)^{j+1}\whfluc(t)\psi} \\
	& + \, \frac{1}{4 N\vol}\scp{ \whfluc(t)\psi}{(\nb+\vol)^{j+2}\whfluc(t)\psi} \, . 
	\end{split}
    \end{align}
	We bound the third term in \eqref{eq-nb/N-ind-0} by 
	\begin{align}\label{eq-nb/N-ind-0-3}
	\MoveEqLeft \frac1{\sqrt{N}\vol}\big|\scp{(\nb+\vol)^{j+1} \whfluc(t)\psi}{\whfluc(t)(a_0+\ad_0)\psi}\big|\\
	\leq & \frac1\vol\scp{ \whfluc(t)(a_0+\ad_0)\psi}{(\nb+\vol)^{j}\whfluc(t)(a_0+\ad_0)\psi}\\
	& + \, \frac1{4N\vol}\scp{ \whfluc(t)\psi}{(\nb+\vol)^{j+2}\whfluc(t)\psi} \, . 
	\end{align} 
    In particular, \eqref{eq-nb/N-ind-0}, \eqref{eq-nb/N-ind-0-2}, and \eqref{eq-nb/N-ind-0-3} imply
    \begin{align} \label{eq-nb/N-ind-0-first-simpl}
        \begin{split}
        \MoveEqLeft\frac1{N\vol}\scp{\whfluc(t)\psi}{(\nb+\vol)^{j+2}\whfluc(t)\psi}\\
        \leq & \, \frac2{N\vol}|\scp{(\nb+\vol)^{j+1} \whfluc(t)\psi}{\whfluc(t)(\nb+\vol)\psi} |\\
        & \, + \, C_j\vol|\phi_t|^2\scp{ \whfluc(t)\psi}{(\nb+\vol)^{j+1}\whfluc(t)\psi} \\
        & \, + \, \frac1\vol\scp{ \whfluc(t)(a_0+\ad_0)\psi}{(\nb+\vol)^{j}\whfluc(t)(a_0+\ad_0)\psi} \, .
        \end{split}
    \end{align}
	For the first term in \eqref{eq-nb/N-ind-0-first-simpl}, we apply \eqref{eq-nb-tUn-com} to the left and obtain
	
	\begin{align}\label{eq-nb-un-nb-0}
    \begin{split}
	\MoveEqLeft\frac1{N\vol}\scp{(\nb+\vol)^{j+1} \whfluc(t)\psi}{\whfluc(t)(\nb+\vol)\psi}\\
	= \, & \frac1{N\vol}\scp{ \whfluc(t)(\nb+\vol)\psi}{(\nb+\vol)^j\whfluc(t)(\nb+\vol)\psi}\\
	& -\frac1{\sqrt{N\vol}}\scp{(\nb+\vol)^j (\overline{\phi}_ta_0+\phi_t\ad_0)\whfluc(t)\psi}{\whfluc(t)(\nb+\vol)\psi} \\
	& + \,\frac1{\sqrt{N}\vol}\scp{(\nb+\vol)^j \whfluc(t)(a_0+\ad_0)\psi}{\whfluc(t)(\nb+\vol)\psi} \, .
    \end{split}
	\end{align}
    With analogous steps to above, we bound the second term in \eqref{eq-nb-un-nb-0} by
    \begin{align}\label{eq-un-a0-nb-un-nb}
    \begin{split}
	\MoveEqLeft\frac1{\sqrt{N\vol}}\big|\scp{(\nb+\vol)^j (\overline{\phi}_ta_0+\phi_t\ad_0)\whfluc(t)\psi}{\whfluc(t)(\nb+\vol)\psi} \big|\\
	\leq \, & C_j\vol|\phi_t|^2\scp{ \whfluc(t)\psi}{(\nb+\vol)^{j+1}\whfluc(t)\psi} \\
	& + \, \frac1{N\vol} \scp{ \whfluc(t)(\nb+\vol)\psi}{(\nb+\vol)^j\whfluc(t)(\nb+\vol)\psi} \, ,
	\end{split}
    \end{align}
    and the third term by
    \begin{align}\label{eq-a0-un-nb-un-nb}
    \begin{split}
	\MoveEqLeft\frac1{\sqrt{N}\vol}\big|\scp{(\nb+\vol)^j \whfluc(t)(a_0+\ad_0)\psi}{\whfluc(t)(\nb+\vol)\psi}\big|\\
	\leq \, & \frac1\vol \scp{ \whfluc(t)(a_0+\ad_0)\psi}{(\nb+\vol)^j\whfluc(t)(a_0+\ad_0)\psi}\\
	& + \, \frac1{N\vol} \scp{ \whfluc(t)(\nb+\vol)\psi}{(\nb+\vol)^j\whfluc(t)(\nb+\vol)\psi} \, .
	\end{split}
    \end{align}
    In particular, we can bound \eqref{eq-nb-un-nb-0} by
    \begin{align}
        \MoveEqLeft\frac1{N\vol}|\scp{(\nb+\vol)^{j+1} \whfluc(t)\psi}{\whfluc(t)(\nb+\vol)\psi}|\\
        \leq \, &  \frac3{N\vol} \scp{ \whfluc(t)(\nb+\vol)\psi}{(\nb+\vol)^j\whfluc(t)(\nb+\vol)\psi} \\
        & + \, C_j\vol|\phi_t|^2\scp{ \whfluc(t)\psi}{(\nb+\vol)^{j+1}\whfluc(t)\psi}\\
        & + \, \frac1\vol \scp{ \whfluc(t)(a_0+\ad_0)\psi}{(\nb+\vol)^j\whfluc(t)(a_0+\ad_0)\psi} \, .
    \end{align}
    As a consequence, \eqref{eq-nb/N-ind-0-first-simpl} implies 
    \begin{align} \label{eq-Nbj+2/N-bd}
        \begin{split}
        \MoveEqLeft\frac1{N\vol}\scp{\whfluc(t)\psi}{(\nb+\vol)^{j+2}\whfluc(t)\psi}\\
        \leq & \, \frac6{N\vol} \scp{ \whfluc(t)(\nb+\vol)\psi}{(\nb+\vol)^j\whfluc(t)(\nb+\vol)\psi}\\
        & \, + \, C_j\vol|\phi_t|^2\scp{ \whfluc(t)\psi}{(\nb+\vol)^{j+1}\whfluc(t)\psi} \\
        & \, + \, \frac2\vol\scp{ \whfluc(t)(a_0+\ad_0)\psi}{(\nb+\vol)^{j}\whfluc(t)(a_0+\ad_0)\psi} \, .
        \end{split}
    \end{align}
    
	For the first term in \eqref{eq-nb-un-nb-0}, we use the induction hypothesis \eqref{eq-nb-IH-2}, and obtain
	
	\begin{align}\label{eq-Nbj+2/N-first}
    \begin{split}
	\MoveEqLeft\frac1{N\vol}\scp{ \whfluc(t)(\nb+\vol)\psi}{(\nb+\vol)^j\whfluc(t)(\nb+\vol)\psi}\\
    \leq \, & C_{j-1} \Big(\Big(1+\frac{(\|\gtr_0\|_1+\|\gtr_0\|_\infty+1)^2}N\Big)^{j-1} e^{K_{j-1}\vd\lambda\vol(1+\frac{\|\gtr_0\|_1}N)t}\\
    & \, \scp{\psi}{(\nb+\vol)^{j+1}\big(1+\frac{\nb}{N\vol}\big)\psi} \\
	& + \, \Big(1+\frac{\|\gtr_0\|_1}N\Big)\scp{\whfluc(t)(\nb+\vol)\psi}{(\nb+\vol)^{j-1}\whfluc(t)(\nb+\vol)\psi}\Big) \\
    \leq \, & C_j	\Big(1+\frac{(\|\gtr_0\|_1+\|\gtr_0\|_\infty+1)^2}N\Big)^j e^{K_j\vd\lambda\vol(1+\frac{\|\gtr_0\|_1}N)t}\\
    & \, \scp{\psi}{(\nb+\vol)^{j+1}\big(1+\frac{\nb}{N\vol}\big)\psi} 
	\end{split}
    \end{align}
	where in the last step, we employed the induction hypothesis \eqref{eq-nb-IH-1}. 
    \par For the third term in \eqref{eq-nb-un-nb-0}, we use the induction hypothesis \eqref{eq-nb-IH-1} to obtain
    \begin{align}\label{eq-Nbj+2/N-third}
        \begin{split}
        \MoveEqLeft\frac2\vol\scp{ \whfluc(t)(a_0+\ad_0)\psi}{(\nb+\vol)^{j}\whfluc(t)(a_0+\ad_0)\psi}\\
        &\leq \, \frac{C_j}{\vol} \Big(1+\frac{(\|\gtr_0\|_1+\|\gtr_0\|_\infty+1)^2}N\Big)^j e^{K_j\vd\lambda\vol(1+\frac{\|\gtr_0\|_1}N)t}\\
        & \, \scp{(a_0+\ad_0)\psi}{(\nb+\vol)^j\big(1+\frac{\nb}{N\vol}\big)(a_0+\ad_0)\psi} \\
        & \leq \, C_j \Big(1+\frac{(\|\gtr_0\|_1+\|\gtr_0\|_\infty+1)^2}N\Big)^j e^{K_j\vd\lambda\vol(1+\frac{\|\gtr_0\|_1}N)t}\\
        & \, \scp{\psi}{(\nb+\vol)^{j+1}\big(1+\frac{\nb}{N\vol}\big)\psi} \, ,
        \end{split}
    \end{align}
    where, in the last step, we employed \eqref{eq-a0nba0-bd-0}.
	\par Substituting \eqref{eq-Nbj+2/N-first} and \eqref{eq-Nbj+2/N-third} into \eqref{eq-Nbj+2/N-bd}, we arrive at
	\begin{align}
	\MoveEqLeft\frac1{N\vol}\scp{\whfluc(t)\psi}{(\nb+\vol)^{j+2}\whfluc(t)\psi}\\
	\leq \, & C_j \Big(1+\frac{(\|\gtr_0\|_1+\|\gtr_0\|_\infty+1)^2}N\Big)^{j+1} e^{K_j\vd\lambda\vol(1+\frac{\|\gtr_0\|_1}N)t}\\
    & \scp{\psi}{(\nb+\vol)^{j+1}\big(1+\frac{\nb}{N\vol}\big)\psi} \\
    & + \, C_j \vol|\phi_t|^2\scp{ \whfluc(t)\psi}{(\nb+\vol)^{j+1}\whfluc(t)\psi} \, .
    \end{align}
	This concludes the proof.
    \end{enumerate}
	\endprf

    \begin{lemma}[Expressions for $\HBEC$ and $\HHFB$]\label{lem-HBEC-HHFB-with-HFB-eq}
        Assume that $(\phit,\gamt,\sigt)$ satisfy \eqref{eq-HFB-ren-2} and that $\Omt$ satisfies \eqref{eq-Omt-explicit}. Then we have
        \begin{align}
          \MoveEqLeft\HBEC(t) \, = \\
          & -\frac{\lambda\sqrt{\vol}}{\sqrt{N}}\Big[u_t(0)\Big(2(\fplus\sigma_t*\hat{v})(0)\overline{\phi}_t+\big((1+2\gamma_t)\fplus\big)*(\hat{v}+\hat{v}(0))(0)\phi_t\Big)\\
          & + \, v_t(0)\Big(2(\fplus\sigma_t*\hat{v})(0)\phi_t+\big((1+2\gamma_t)\fplus\big)*(\hat{v}+\hat{v}(0))(0)\overline{\phi}_t\Big)e^{i\int_0^t \dx{s} \Omega_s(0)}\ad_0 \\
          & + \, \mathrm{h.c.} \, ,\\
          \HHFBd(t) \, =& \, -\frac{\lambda}N\int \dx{p} \Big[2\Re\Big(\big((\fplus\sigtb_t)*\hat{v}\big)(p)\sigt_t(p)\Big) \\
          & \quad + \, \big((\fplus(1+2\gamt))*(\hat{v}+\hat{v}(0))\big)(p)\big(1+2\gamt_t(p)\big)\Big] \ad_pa_p \, ,\\
          \HHFBod(t) \, = & \, -\frac{\lambda}N\int \dx{p} \Big[\Big(\big(\fplus(1+2\gamt_t)\big)*(\hat{v}+\hat{v}(0))\Big)(p)\sigt_t(p) \\
         & \quad + \, \big((\fplus\sigt_t)*\hat{v}\big)(p)(1+\gamt_t(p)) \\
         & \quad + \, \big((\fplus\sigtb_t)*\hat{v}\big)(p)\frac{\sigt_t(p)^2}{1+\gamt_t(p)}\Big]e^{2i\int_0^t \dx{\tau} \Omt_\tau(p)}\ad_p\ad_{-p} \, + \, \mathrm{h.c.} 
      \end{align}
    \end{lemma}
	\begin{proof}
	    We start by computing $\HBEC(t)$. By Lemma \ref{lem-hfluc}, we have that
         \begin{align}\label{eq-HBEC-raw}
		 	\MoveEqLeft \HBEC(t) \, = \, \\
		 	& \sqrt{N\vol}\Big[u_t(0)\Big( -i\partial_t \phi_t  +  \lambda\vol\hat{v}(0)|\phi_t|^2 \phi_t +  \frac{\lambda}{N}(\sigma_t* \hat{v})(0)\overline{\phi}_t\\
		 	& + \, \frac{\lambda}{N} \big(\gamma_t*(\hat{v}+\hat{v}(0))\big) (0) \phi_t\Big) \, + \, v_t(0)\Big( -\overline{i\partial_t \phi_t} +  \lambda\vol\hat{v}(0) |\phi_t|^2\overline{\phi}_t \\
		 	& +  \frac{\lambda}{N}(\overline{\sigma}_t* \hat{v})(0)\phi_t + \frac{\lambda}{N} \big(\gamma_t*(\hat{v}+\hat{v}(0))\big) (0) \overline{\phi}_t\Big)\Big] e^{i\int_0^t \dx{s} \Omega_s(0)}\ad_0 \, + \, \mathrm{h.c.} 
		 \end{align} 
      Recalling \eqref{eq-HFB-ren-2}, $\phit$ satisfies
      \begin{align}
          i\partial_t \phit_t  \, = & \,  \frac{\lambda}{N} \Big( \big(\Gamt_t*(\hat{v}+\hat{v}(0))\big) (0) \phit_t \, + \, (\Sigt_t*\hat{v}) (0) \phitb_t\Big) \\
	       & - \, 2\lambda\vol\hat{v}(0) |\phit_t|^2 \phit_t \, .  
      \end{align}
      With that, \eqref{eq-HBEC-raw} implies
      \begin{align}
          \MoveEqLeft\HBEC(t) \, = \\
          & -\frac{\lambda\sqrt{\vol}}{\sqrt{N}}\Big[u_t(0)\Big(2(\fplus\sigma_t*\hat{v})(0)\overline{\phi}_t+\big((1+2\gamma_t)\fplus\big)*(\hat{v}+\hat{v}(0))(0)\phi_t\Big)\\
          & + \, v_t(0)\Big(2(\fplus\sigma_t*\hat{v})(0)\phi_t+\big((1+2\gamma_t)\fplus\big)*(\hat{v}+\hat{v}(0))(0)\overline{\phi}_t\Big)e^{i\int_0^t \dx{s} \Omega_s(0)}\ad_0 \\
          & + \, \mathrm{h.c.} 
      \end{align}
      Similarly, Lemma \ref{lem-hfluc} implies
      \begin{align}\label{eq-HHFBd-raw}
		 \MoveEqLeft \HHFBd(t) \, = \, \\
		 & \int \dx{p} \Big[ -\Omega_t(p)  -  \frac{\Re\big(\overline{\sigma}_t(p)i\partial_t\sigma_t(p)\big)}{1+\gamma_t(p)}\\
		 & +\,  \Big(E(p)+\frac{\lambda}N\big((\gamma_t+N\vol|\phi_t|^2\delta)*(\hat{v}+\hat{v}(0))\big)(p)\Big)\big(1+2\gamma_t(p)\big) \\
		 & + \, \frac{2\lambda}N\Re\Big(\big((\overline{\sigma}_t+N\vol\overline{\phi}_t^2\delta)*\hat{v}\big)(p)\sigma_t(p)\Big)\Big]\ad_p a_p  
		 \end{align}
      By \eqref{eq-HFB-ren-2}, we have that
      \begin{align}
          i\partial_t \sigt_t \, = \, 2\big(E+\frac{\lambda}{N}\Gamt_t*(\hat{v}+\hat{v}(0))\big)\sigt_t \, + \, \frac{\lambda}{N} \big( \Sigt_t* \hat{v}\big) (1+2\gamt_t) \, ,
      \end{align}
      while \eqref{eq-Omt-explicit} implies
      \begin{align}
          \Omt_t \, = \, E+\frac{\lambda}N\big(\Gamt_t*(\hat{v}+\hat{v}(0))\big) \, + \, \frac{\lambda}N\frac{\Re\big(\Sigtb*\hat{v})\sigt_t\big)}{1+\gamt_t} \, .
      \end{align}
      Using the fact that $|\sigma|^2=\gamma(\gamma+1)$, \eqref{eq-HHFBd-raw} thus implies
      \begin{align}
          & \HHFBd(t) \, = \, \\
          &\quad \int \dx{p}\Big[-\frac{\lambda}N\frac{\Re\big(\Sigtb*\hat{v})(p)\sigt_t(p)\big)}{1+\gamt_t(p)} \, + \, \frac{2\lambda}N\Re\Big(\big((\sigtb_t+N\vol(\phitb_t)^2\delta)*\hat{v}\big)(p)\sigt_t(p)\Big)\\
          & \qquad - \, \frac{\lambda}{N} \Re\Big(\big( \Sigtb_t* \hat{v}\big)(p)\sigt(p)\Big) \frac{1+2\gamt_t(p)}{1+\gamt_t(p)} \, - \, \frac{\lambda}N\big(\Gamt_t*(\hat{v}+\hat{v}(0))\big)(p)(1+2\gamt_t(p))\\
          & \qquad + \, \frac{\lambda}N\big((\gamt_t+N\vol|\phit_t|^2\delta)*(\hat{v}+\hat{v}(0))\big)(p)\big(1+2\gamt_t(p)\big)\Big] \ad_pa_p \, .
      \end{align}
      Simplifying the terms, recalling definitions \eqref{def-Gamt} and \eqref{def-Sigt}, we obtain
      \begin{align}
          \HHFBd(t) \, =& \, -\frac{\lambda}N\int \dx{p} \Big[2\Re\Big(\big((\fplus\sigtb_t)*\hat{v}\big)(p)\sigt_t(p)\Big) \\
          & \quad + \, \big((\fplus(1+2\gamt))*(\hat{v}+\hat{v}(0))\big)(p)\big(1+2\gamt_t(p)\big)\Big] \ad_pa_p \, .
      \end{align}
      Finally, Lemma \ref{lem-hfluc} yields
      \begin{align}\label{eq-HFBod-raw}
		 \HHFBod(t) \, = & \, \int \dx{p} \Big[-\frac{i\partial_t \sigma_t(p)}2+\frac{\sigma_t(p)i\partial_t \gamma_t(p)}{2(1+\gamma_t(p))}\\
		 & + \,   \Big(E(p)+\frac{\lambda}N\big((\gamma_t+N\vol|\phi_t|^2\delta)*(\hat{v}+\hat{v}(0))\big)(p)\Big)\sigma_t(p)  \\
		 & + \, \frac{\lambda}{2N}\Big( \big((\sigma_t+N\vol\phi_t^2\delta)*\hat{v}\big)(p)(1+\gamma_t(p)) \\
		 & + \, \big((\overline{\sigma}_t+N\vol\overline{\phi}_t^2\delta)*\hat{v}\big)(p)\frac{\sigma_t(p)^2}{1+\gamma_t(p)}\Big) \Big] e^{2i\int_0^t \dx{\tau} \Omega_\tau(p)}\ad_p\ad_{-p} \, + \, \mathrm{h.c.} 
		 \end{align}
      By \eqref{eq-HFB-ren-2}, we have that
      \begin{align}
        \begin{split}
        \partial_t \gamt_t \, =& \,\frac{2\lambda}N \Im\big((\Sigt_t*\hat{v} ) \sigtb_t\big) \, = \, \frac{\lambda}{iN}\big((\Sigt_t*\hat{v} ) \sigtb_t-(\Sigtb_t*\hat{v} ) \sigt_t\big)\, , \\
	i\partial_t \sigt_t \, =& \, 2\big(E+\frac{\lambda}{N}\Gamt_t*(\hat{v}+\hat{v}(0))\big)\sigt_t \, + \, \frac{\lambda}{N} \big( \Sigt_t* \hat{v}\big) (1+2\gamt_t) \, .
	\end{split}
    \end{align}
    Substituting these into \eqref{eq-HFBod-raw}, yields
    \begin{align}
		 \MoveEqLeft \HHFBod(t) \, = \, \\
		 & \int \dx{p} \Big[-\Big(E(p)+\frac{\lambda}{N}\Gamt_t*(\hat{v}+\hat{v}(0))(p)\Big)\sigt_t(p) \\
      & - \, \frac{\lambda}{2N} \big( \Sigt_t* \hat{v}\big) (p)(1+2\gamt_t(p))\\
      & + \, \frac{\lambda}{2N}\Big((\Sigt_t*\hat{v})(p) \gamt_t(p)-(\Sigtb_t*\hat{v})(p) \frac{\sigt_t(p)^2}{1+\gamt_t(p)}\Big) \\
		 & + \, \Big(E(p)+\frac{\lambda}N\big((\gamt_t+N\vol|\phit_t|^2\delta)*(\hat{v}+\hat{v}(0))\big)(p)\Big)\sigt_t(p)  \\
		 & + \, \frac{\lambda}{2N}\Big( \big((\sigt_t+N\vol(\phit_t)^2\delta)*\hat{v}\big)(p)(1+\gamt_t(p)) \\
		 & + \, \big((\sigtb_t+N\vol(\phitb_t)^2\delta)*\hat{v}\big)(p)\frac{\sigt_t(p)^2}{1+\gamt_t(p)}\Big) \Big] e^{2i\int_0^t \dx{\tau} \Omt_\tau(p)}\ad_p\ad_{-p} \, + \, \mathrm{h.c.} 
		 \end{align}
      We simplify this expression as
      \begin{align}
          \HHFBod(t) \, = & \, -\frac{\lambda}N\int \dx{p} \Big[\Big(\big(\fplus(1+2\gamt_t)\big)*(\hat{v}+\hat{v}(0))\Big)(p)\sigt_t(p) \\
         & \quad + \, \big((\fplus\sigt_t)*\hat{v}\big)(p)(1+\gamt_t(p)) \\
         & \quad + \, \big((\fplus\sigtb_t)*\hat{v}\big)(p)\frac{\sigt_t(p)^2}{1+\gamt_t(p)}\Big]e^{2i\int_0^t \dx{\tau} \Omt_\tau(p)}\ad_p\ad_{-p} \, + \, \mathrm{h.c.} 
      \end{align}
	\end{proof}
	\bibliographystyle{plain}  
	\bibliography{references}     
	
\end{document}